\documentclass
[11pt,letterpaper]{article}

\usepackage{mdframed}
\usepackage{macros}
\usepackage{amsthm}
\usepackage{thmtools}
\usepackage{turnstile}
\newcommand{\offdiag}{\mathsf{offdiag}}

\newtheorem{algorithm}{Algorithm}
\numberwithin{algorithm}{subsection}

\begin{document}

\title{Polynomial-Time Power-Sum Decomposition of Polynomials}

\author{Mitali Bafna\thanks{Harvard University, \texttt{mitalibafna@g.harvard.edu}. Supported in part by a Siebel Scholar award, a Simons Investigator Award and NSF Award CCF 2152413.}\and 
Jun-Ting Hsieh\thanks{Carnegie Mellon University, \texttt{juntingh@cs.cmu.edu}. Supported in part by NSF CAREER Award \#2047933.}\and
Pravesh K. Kothari\thanks{Carnegie Mellon University, \texttt{praveshk@cs.cmu.edu}. Supported by  NSF CAREER Award \#2047933, Alfred P. Sloan Fellowship and a Google Research Scholar Award.}\and
Jeff Xu\thanks{Carnegie Mellon University, \texttt{jeffxusichao@cmu.edu}. Supported in part by NSF CAREER Award \#2047933.}}

\date{\today}
\maketitle

\begin{abstract}
We give efficient algorithms for finding power-sum decomposition of an input polynomial $P(x)= \sum_{i\leq m} p_i(x)^d$ with component $p_i$s. The case of linear $p_i$s is equivalent to the well-studied tensor decomposition problem while the quadratic case occurs naturally in studying identifiability of non-spherical Gaussian mixtures from low-order moments. 

Unlike tensor decomposition, both the unique identifiability and algorithms for this problem are not well-understood. For the simplest setting of quadratic $p_i$s and $d=3$, prior work of \cite{GeHK15} yields an algorithm only when $m \leq \widetilde{O}(\sqrt{n})$. On the other hand, the more general recent result of~\cite{GKS20} builds an algebraic approach to handle any $m=n^{O(1)}$ components but only when $d$ is large enough (while yielding no bounds for $d=3$ or even $d=100$) and only handles an inverse exponential noise.

Our results obtain a substantial quantitative improvement on both the prior works above even in the base case of $d=3$ and quadratic $p_i$s. Specifically, our algorithm succeeds in decomposing a sum of $m \sim \widetilde{O}(n)$ generic quadratic $p_i$s for $d=3$ and more generally the $d$th power-sum of $m \sim n^{2d/15}$ generic degree-$K$ polynomials for any $K \geq 2$. Our algorithm relies only on basic numerical linear algebraic primitives, is \emph{exact} (i.e., obtain arbitrarily tiny error up to numerical precision), and handles an inverse polynomial noise when the $p_i$s have random Gaussian coefficients.

Our main tool is a new method for extracting the linear span of $p_i$s by studying the linear subspace of low-order partial derivatives of the input $P$. For establishing polynomial stability of our algorithm in average-case, we prove inverse polynomial bounds on the \emph{smallest} singular value of certain correlated random matrices with low-degree polynomial entries that arise in our analyses. Since previous techniques only yield significantly weaker bounds, we analyze the \emph{smallest} singular value of matrices by studying the \emph{largest} singular value of certain deviation matrices via \emph{graph matrix decomposition} and the trace moment method.

\end{abstract}

\thispagestyle{empty}
\setcounter{page}{0}
\newpage
\thispagestyle{empty}
\setcounter{page}{0}
\begin{spacing}{1.05}
    \tableofcontents
\end{spacing}

\thispagestyle{empty}
\setcounter{page}{0}
\newpage


\section{Introduction}

An $n$-variate polynomial $P(x)$ admits a power-sum decomposition if it can be written as $P(x) = \sum_{i \leq m} p_i(x)^d$ for some low-degree polynomials $p_i$s. This work is about the algorithmic problem of computing such a decomposition when it exists and the related structural question of when such a decomposition, if it exists, is unique.

When $p_i$s are \emph{linear} forms $\iprod{v_i,x}$ for $v_i \in \R^n$, the task of decomposing $P$ is equivalent to decomposing the corresponding coefficient tensor $\sum_i v_i^{\otimes d}$ into rank $1$ components. For $d=2$, this corresponds to rank decomposition of matrices, which is unique only in degenerate settings. For $d=3$, while the problem is already NP-hard~\cite{Has90}, there is a long line of work on identifying natural sufficient conditions (e.g., Kruskal's condition~\cite{MR444690}) that imply uniqueness of decomposition in all but degenerate settings. There are known efficient algorithms for decomposing tensors satisfying such non-degeneracy conditions and such algorithms form basic primitives in \emph{tensor methods}~\cite{Harshman1970Foundations, MR907286, MR1238921, LCC07, BhaskaraCMV14, GeHK15, AnandkumarGHKT15,GeM15, HopkinsSS15, HSSS15, MaSS16, MW19, KivvaP20}. 
An influential line of work has developed efficient learning algorithms for a long list of interesting statistical models (under appropriate assumptions) including Mixtures of Spherical Gaussians~\cite{HK,GeHK15}, Independent Component Analysis~\cite{LCC07}, Hidden Markov Models~\cite{MosselRoch:05}, Latent Dirichlet Allocations~\cite{AnandkumarFHKL12}, and Dictionary Learning~\cite{MR3388192-Barak15} via reductions to tensor decomposition. Higher-degree power-sum decomposition is a natural generalization of the tensor decomposition problem and is equivalent to the well-studied problem of reconstructing certain classes of arithmetic circuits~\cite{kss14,10.1145/3313276.3316360, GKS20} with connections (see surveys \cite{SY10,Sap15,CKW11}) to algebraic circuit lower bounds and derandomization.

\paragraph{Tensor Decomposition with Symmetries} Higher-degree power-sum decomposition  is equivalent to a strict generalization of tensor decomposition where the components are symmetrized under a natural group action. For example, when $p_i(x) = x^\top A_i x$ are homogeneous quadratic polynomials for $n \times n$ matrices $A_i$, the coefficient tensor of $P$ has the form $\E_{\sigma \sim \bbS_{6}} \sum_{i \leq m} \sigma(A_i^{\otimes 3})$ where $\bbS_{6}$ is the symmetric group on $6$ elements and acts\footnote{For e.g., for a symmetric matrix $A$, $\E_{\sigma \sim \bbS_{6}}[A^{\otimes 3}((i_1, i_2, i_3), (j_1, j_2, j_3))] = \E [A(e_1) A(e_2) A(e_3)]$ where the expectation is over the choice of a uniformly random perfect matching $(e_1, e_2, e_3)$ of $\{i_1, i_2, i_3, j_1, j_2, j_3\}$.} by permuting the $6$ indices involved in any entry of $A_i^{\otimes 3}$. If not for the action of $\sigma$, the coefficient tensor would simply be a sum of tensor powers of vectorized $A_i$s. The group action, however, has a drastic effect on the identifiability and algorithms for the problem.
Specifically, the symmetrization causes the resulting tensor to have a large rank and thus any decomposition algorithm must strongly exploit the symmetries to succeed. In fact, in Section~\ref{sec:identifiability-cubics-of-quadratics}, we exhibit a simple example of a sum of cubics of quadratics on $2$ variables whose components are \emph{not} uniquely identifiable even though the corresponding coefficient matrices of the quadratics are linearly independent. This is in contrast to the well-known result~\cite{Harshman1970Foundations,MR1238921} that 3rd order tensors with linearly independent components are uniquely identifiable and efficiently decomposable. 
This is similar to other orbit recovery problems that also reduce to tensor decomposition with symmetries such as multi-reference alignment and the cryo-EM~\cite{PerryWBRS17,MW19} problem where even establishing information-theoretic identifiability for generic parameters is significantly more challenging.

\paragraph{The Quadratic Case} Despite being the natural next step after linear $p_i$s, power-sum decomposition of quadratic $p_i$s is not well understood. In a seminal work, Ge, Huang and Kakade~\cite{GeHK15} (GHK from now) proved that the first $6$ moments of a mixture of $m\sim \sqrt{n}$ non-spherical Gaussians with smoothed parameters \emph{exactly} identify and (noise-resiliently) recover the $m$ sets of means and covariances. Their analysis involves giving an algorithm (and uniqueness proof) for decomposing sums of cubics of \emph{smoothed} quadratic \emph{positive definite} polynomials but naturally generalizes to arbitrary smoothed quadratics. This is a striking result that exhibits a large gap between smoothed/generic parameters and arbitrary ones for mixtures of Gaussians as it is known that we need $\Omega(m)$ moments (an $n^{\Omega(m)}$-size object) to uniquely identify the parameters of arbitrary mixtures of $m$ Gaussians~\cite{MoitraValiant:10}. Their approach uses a conceptually elegant  ``desymmetrize+tensor-decompose" strategy by first undoing the effect of the group action and then applying tensor decomposition. While their approach can potentially be extended to $m \geq \sqrt{n}$, it seems to encounter an inherent barrier at $m \geq n^{2/3}$ as we explain in Section~\ref{sec:overview}. Nevertheless, GHK conjectured that it should be possible to handle $m \approx n^{1-\delta}$ generic components for any $\delta>0$ given $O(1)$-degree mixture moments which, in our context, corresponds to decomposing a sum of higher constant degree powers of quadratics.

\paragraph{The Garg-Kayal-Saha Algorithm} 
In a beautiful work of Garg, Kayal and Saha~\cite{GKS20} (GKS from now), they suggest that there is an inherent barrier to extending the ``desymmetrize+tensor-decompose'' based approach of~\cite{GeHK15}. Instead, they work by exploiting an intriguing connection to algebraic circuit lower bounds and develop algorithms to recover any \emph{polynomial} number of generic components from their power-sums of large enough degree.  This algorithm however has two important deficiencies. 

First, their strategy yields a decomposition algorithm for degree-$d$ power-sums only when $d$ is very large compared to the degree of the component $p_i$s. In particular, they do not obtain any result for the simplest interesting setting of $d=3$rd (or even $100$th) power of quadratics\footnote{Indeed, while their bounds can likely be somewhat optimized, the smallest power of quadratics that their algorithm (as currently analyzed) succeeds in decomposing must be larger than $2^{335}$.}. As a result, their techniques seem unsuitable to answer natural questions such as whether 6th moments of mixtures of non-spherical Gaussians (with generic parameters) can uniquely identify $m \geq n^{0.51}$ components of Gaussian mixture in $n$ dimensions, or, whether a sum of $m \approx n$ cubics of generic quadratics can be uniquely decomposed. Second, their algorithm relies on algebraic methods for finding simultaneous vector-space decomposition. The resulting algorithm is not error-resilient and does not appear to handle even a small (\eg $\exp(-n)$ in each entry) amount of noise in the input polynomial. In fact, GKS suggest finding a stable algorithm for power-sum decomposition as an open question.

\paragraph{This Work} In this paper, we give a conceptually simple algorithm that substantially improves the quantitative results in~\cite{GKS20} for decomposing power-sums of low-degree polynomials. Somewhat surprisingly, our algorithm follows the ``desymmetrize+decompose'' approach similar to \cite{GeHK15} while circumventing the barriers suggested by~\cite{GKS20}. A key component is an efficient algorithm to extract the linear span of the coefficient tensors of (powers of) $p_i$s from the subspace of ``coordinate restrictions" of \emph{partial derivatives} of $P=\sum_{i\leq m} p_i^d$ for $d \geq 3$.  As a consequence of our algorithm, we obtain substantially improved guarantees even for the simplest non-trivial setting of sum of cubics of quadratics and handle $m \sim n$ components.

We give an error-tolerant implementation of our algorithm and prove that when each $p_i$ has independent random Gaussian coefficients, the resulting algorithm tolerates an inverse polynomial amount of adversarial noise in the coefficients of the input polynomial. A key technical step in such an analysis requires establishing inverse polynomial lower bounds on the singular values of certain correlated random matrices whose entries are low-degree polynomials in the coefficients of $p_i$s. Standard results (e.g., from~\cite{BhaskaraCMV14}) for analyzing smallest singular values yield significantly weaker bounds in our setting. Instead, we rely on a new elementary but nimble method that lower bounds the \emph{smallest} singular value of correlated random matrices by reducing the task to upper-bounding the much better understood \emph{largest} eigenvalue of certain \emph{deviation} matrices. Our analyses of the spectral norm of such matrices use the trace moment method combined with \emph{graphical matrix decompositions} of random matrices that appear naturally in the analyses of sum-of-squares lower bound witnesses~\cite{HKP15,BHK19} for average-case refutation problems. In particular, these sharper bounds are crucial in allowing us to handle $m \sim \tO(n)$ components for decomposing sums of cubics of quadratics.

\subsection{Our results} 

Our main result gives a polynomial time algorithm (in the standard bit complexity model with exact rational arithmetic) for decomposing a sum of $d$-th powers of generic (e.g., smoothed) polynomials. We note that just as in standard tensor decomposition, sums of squares of low-degree polynomials are uniquely decomposable only in degenerate settings (see Section~\ref{app:quadratic-failure}), so cubics of quadratics (i.e., $d=3$) is the simplest non-trivial setting in this context.

\begin{theorem}[Decomposing Power-Sums of Smoothed Polynomials]
\label{thm:main-theorem-hd-intro-smoothed}
  There is an algorithm that takes input an $n$-variate degree-$Kd$ (for $d$ a multiple of $3$) polynomial of the form $\wh{P}(x) = \sum_{i \leq m} \wh{A}_i(x)^{d}$ where $\wh{A}_i = A_i + G_i$ for an arbitrary degree-$K$ polynomial $A_i$ and a degree-$K$ polynomial $G_i$ with independent $\cN(0,\rho^2)$ coefficients, runs in time polynomial in the size of its input and $1/\rho$, and has the following guarantee:
  with probability at least $0.99$ over the draw of $G_i$s and internal randomness, it outputs the set $\{\wh{A}_i \mid i \leq m\}$ up to permutation (and signs, if $d$ is even) whenever
  \begin{itemize}
      \item $m \leq \tO(n)$ for $d=3$, $K=2$,
      \item $m \leq \tO(n^2)$ for $d=6$, $K=2$,
      \item $m \leq \tO(n^{2d/9})$ for any $d \geq 9$ and $K=2$,
      \item $m \leq \tO(n^{2Kd/3(5K-4)})$ for all $d \geq 9$ and $K \geq 2$.
  \end{itemize}
\end{theorem}

The theorem above works more generally for any model of smoothing that independently perturbs the coefficients of each $A_t$ with a distribution that allots a probability of at most $1/n^{O(d)}$ to any single point. In particular, a fine-enough discretization of any continuous smoothing suffices. As observed in~\cite{GKS20,GeHK15}, identifying components of non-spherical mixtures of Gaussians from low-degree moments is equivalent\footnote{This follows from the fact that for $x \in \R^n$, the $2d$-th moment of $\cN(0,\Sigma)$ in direction $x$ equals $\E_{y \sim \cN(0,\Sigma)}[ \iprod{y,x}^{2d}] = \frac{(2d)!}{2^d d!} \E[\iprod{y,x}^2]^d = \frac{(2d)!}{2^d d!} (x^{\top}\Sigma x)^d$. Consequently, $\E_{ y \sim \sum_i w_i \cN(0, \Sigma_i)} [ \iprod{x,y}^{2d}] = \frac{(2d)!}{2^d d!} \sum_{i} w_i \paren{x^{\top} \Sigma_i x}^d$.} to decomposing the power-sum of quadratic polynomials. Thus, as an immediate corollary of the theorem above, we obtain:

\begin{corollary}[Moment Identifiability of Smoothed Mixtures of Gaussians] \label{cor:mog}
The parameters of a zero-mean mixture of Gaussians $\sum_{i\leq m} w_i \cN(0,\Sigma_i)$, with arbitrary mixture weights $w_i$ and smoothed\footnote{Any continuous smoothing suffices for this result. For e.g., for an arbitrary $\wh{\Sigma}_i \succeq 0$, for $\rho = n^{-O(1)}$, add an independent and uniformly random entry from $[-\rho, \rho]$ to every off-diagonal entry of $\wh{\Sigma}_i$ and a uniformly random entry from $[n \rho,2 n \rho]$ to every diagonal entry of $\wh{\Sigma}_i$ to produce $\Sigma_i$. Note that the resulting matrix $\Sigma_i$ is positive semidefinite.} covariances $\Sigma_i$, are uniquely identifiable from the first $2d$ moments for any $m \leq \tO(n^{2d/9})$. For $d=3$ and $6$, the bound improves to $m \leq \tO(n)$ and $m\leq \tO(n^2)$ respectively.
\end{corollary}

\paragraph{Error-Resilience for Random Components}
When $P(x) = \sum_i A_i(x)^{d}+E(x)$ where each $A_i$ has independent, standard Gaussian coefficients, we prove that the our algorithm above in fact is error-resilient and tolerates an inverse polynomial error in every coefficient of the input $\wh{P}$. Indeed, Theorem~\ref{thm:main-theorem-hd-intro-smoothed} above is obtained essentially as a corollary (combined with simple algebraic tools) of this stronger analysis for random components; see Section~\ref{sec:generic-A-appendix}.

\begin{theorem}[Power-sum Decomposition of Random Polynomials, See Theorems~\ref{thm:main-theorem-hd}, \ref{thm:main-theorem-high-degree}] \label{thm:main-theorem-hd-intro}
There is a polynomial time algorithm that takes input an $n$-variate degree-$Kd$ (for $d$ a multiple of $3$) polynomial of the form $\wh{P}(x) = \sum_{i \leq m} A_i(x)^{d}+E(x)$ where $A_i$ is a degree-$K$ polynomial with independent $\cN(0,1)$ coefficients, and $E(x)$ is an arbitrary polynomial of degree $Kd$, and has the following guarantees: 
with probability at least $0.99$ over the draw of $A_i$s and internal randomness, it outputs the set $\{\wt{A}_i \mid i \leq m\}$ that contains an estimate of each $A_i$ up to permutation (and signs, if $d$ is even) with an error of at most $n^{O(1)}\Norm{E}_F^{1/d}$ whenever
\begin{itemize}
    \item $m \leq \tO(n)$ for $d=3$ and $K=2$,
    \item $m \leq \tO(n^2)$ for $d=6$ and $K=2$,
    \item $m \leq \tO(n^{2d/9})$ for any $d \geq 9$ and $K=2$,
    \item $m \leq \tO(n^{2Kd/3(5K-4)})$ for all $d \geq 9$ and $K \geq 2$.
\end{itemize}
\end{theorem}

\subsection{Discussion and comparison to prior works}
Theorem~\ref{thm:main-theorem-hd-intro} shows that our algorithm tolerates an inverse polynomial amount of noise in each entry when the component $A_i$s are random. Theorem~\ref{thm:main-theorem-hd-intro-smoothed} is in fact an immediate corollary of our analysis for the random case combined with standard tools. Our result for generic (as opposed to random) $p_i$s only handles an inverse exponential amount of noise. We believe that the same algorithm should handle inverse polynomial noise (i.e., is well-conditioned) in any reasonable smoothed analysis model. However, establishing such a result likely requires new  techniques for analyzing condition numbers of matrices with dependent, low-degree polynomial entries in independent random variables. 

For the simplest setting of sums of cubics of quadratics (i.e., $K=2$ and $d = 3$), our theorem yields a polynomial time algorithm that succeeds whenever $m\leq \tO(n)$. This improves on the algorithm implicit in~\cite{GeHK15} that succeeds\footnote{Their algorithm succeeds more generally for smoothed $A_i$s but in addition, needs access to $\sum_i A_i(x)^{2}$.} for $m \leq \tO(\sqrt{n})$. As we discuss in Section~\ref{sec:overview}, natural extensions of their techniques to higher degree power-sums also appear to break down for $m \geq n$.

The work of~\cite{GKS20} recently found a more sophisticated algorithm (that works in general on all large enough fields) that relies on simultaneous decomposition of vector spaces that escapes this barrier. In particular, they showed that for any $K$, $m=n^{O(1)}$, there is an algorithm that succeeds in decomposing a sum of $m$ $d$th powers of generic degree-$K$ polynomials for \emph{large enough $d$}. Their algorithm however requires that $d$ be very large as a function of $K$ and $\log_n m$ and in particular, does not work for $d=3$ (or even $100$) for example. Their algorithm relies on exact algorithms for certain algebraic operations and does not appear to tolerate any more than an inverse exponential (in $n$) amount of noise in the input.

The corollary above immediately improves the moment identifiability of mixtures of smoothed centered Gaussians shown in both the works above. Extending our algorithm to the ``asymmetric" case of sums of products of quadratics (instead of powers) will allow the above corollary to succeed for Gaussians with arbitrary mean, but we do not pursue this goal in this paper. We also note that unlike~\cite{GeHK15}, our theorem above does not immediately yield a polynomial time algorithm for learning mixtures of smoothed Gaussians from samples (similar to~\cite{GKS20}). This is because samples from the mixture only give us access to the corresponding sum of powers of quadratics with inverse polynomial additive error in each entry while our current analysis for the case of smoothed components only handles an inverse exponential error. 

\paragraph{Open Questions} 
Despite the progress in this work, we are far from understanding identifiability and algorithms for power-sum decomposition. Our result shows unique identifiability for sums of $\sim n$ cubics of quadratics. Could this be improved to $n^2$? Conversely, could we produce evidence of hardness of decomposing sums of $\omega(n)$ cubics of quadratics? Analogous questions arise for higher-degree polynomials and we mention one that eludes the current approach in both our work and~\cite{GKS20}: is it possible to obtain efficient algorithms that succeed in decomposing sums of $m$ $d$-th powers of degree-$K$ polynomials where $m$ grows as $n^{f(K)d}$ for some $f(K) \rightarrow \infty$ as $K \rightarrow \infty$?

In a different direction, a natural question is to generalize our result to obtain a polynomial time algorithm that decomposes power-sums of smoothed polynomials while tolerating an inverse polynomial entrywise error. Our current analysis obtains such a guarantee for power-sums of random polynomials but can only handle an inverse exponential error in the smoothed setting. We suspect that this goal requires new tools to analyze the smallest singular values of matrices whose entries are low-degree polynomials in independent Gaussians with \textit{non-zero means}.


\subsection{Brief overview of our techniques} 

Given (the special case of) sum of cubics of quadratics $P(x) = \sum_{i\leq m} (x^{\top} A_i x)^3$ for $n \times n$ symmetric matrices $A_i$ with coefficient tensor $\sum_{i\leq m} \Sym_{6}(A_i^{\otimes 3})$, the main idea of the algorithm in~\cite{GeHK15} is a conceptually simple ``desymmetrize + tensor-decompose'' approach. Here, desymmetrization reverses the effect of the polynomial symmetry and yields $\sum_i A_i^{\otimes 3}$, and one can then apply standard tensor decomposition. While $\Sym_6$ is a linear operator on 6th order tensors with an $\Omega(n^6)$-dimensional kernel, it turns out that it is invertible when restricted to tensors where the component $A_i$s are restricted to a \emph{known} generic subspace. The work of~\cite{GeHK15} shows how to estimate the span of $A_i$s -- i.e.\ this subspace -- for $m \leq \tO(\sqrt{n})$. But their techniques do not seem to extend to any $m \gg n^{2/3}$. Indeed, Garg, Kayal and Saha~\cite{GKS20} comment that reduction to tensor decomposition of the sort above cannot yield algorithms that work for $m \gg n$. As a result, they build a considerably more sophisticated approach that relies on an algebraic algorithm for simultaneous decomposition of vector spaces.

Our main idea comes as a surprise in the light of this discussion: we in fact give a conceptually simple ``desymmetrize+tensor-decompose'' based algorithm that substantially improves the bounds obtained in~\cite{GKS20}. Our key idea is a ``Span Finding algorithm'' that recovers the linear span of $A_i$s restricted to any $O(\sqrt{n})$ variables by computing the linear span of \emph{restrictions} of partial derivatives of $P$ and intersecting it with an appropriately constructed random subspace (see Section~\ref{sec:overview} for a more detailed overview).

Our algorithm is implemented using error-resilient numerical linear algebraic operations. In particular, to establish polynomial stability (Theorem~\ref{thm:main-theorem-hd-intro}) for random $A_i$, we need to understand the \emph{smallest} singular values (to obtain well-conditionedness) of certain correlated random matrices arising in our analyses. These random matrices are rather complicated with entries computed as low-degree polynomials (much smaller than the ambient dimension) of independent random variables. Standard techniques for analyzing such bounds (such as the ``leave-one-out'' method~\cite{TV09, TaoV10, RudVer:08} employed in prior works on tensor decomposition~\cite{BhaskaraCMV14, MaSS16}) are inadequate for our purposes and yield weaker bounds (which, in particular, do not allow us to handle $m \sim \sqrt{n}$ for sum of cubics of quadratics, for example).

Instead, we rely on a new elementary method that establishes singular value lower bounds by studying \emph{spectral norm} upper bounds of certain associated deviation matrices. We analyze and prove strong bounds on the spectral norm of such matrices using the graphical matrix decomposition technique that was introduced in~\cite{HKP15, BHK19} and recently used and refined in several works~\cite{graphmatrixbounds, MohantyRX20, GJJPR20, PotechinR20,HK22,sparseindset} on establishing sum-of-squares lower bounds for average-case problems and reducing the bounds to understanding certain combinatorial problems on graphs associated with the matrix. As far as we know, our work is the first use of this technique to prove singular value \emph{lower bounds} and condition numbers in algorithms. We believe that the graphical matrix decomposition toolbox will find further applications in the analyses of numerical algorithms.


\section{Technical Overview} \label{sec:overview}
In this section, we give a high-level overview of our algorithm and the key ideas that go into its design and analysis. Let's fix $P(x) = \sum_{t \leq m} A_t(x)^d+E(x)$ where $A_t(x)$ are homogeneous polynomials of degree $K$ in $n$ indeterminates $x_1, x_2, \ldots, x_n$. Throughout this paper, we will abuse notation slightly and use $A_t$ to also denote the $K$-th order coefficient tensor of the associated polynomial. We will also use $\wt{O}$ to suppress $\polylog(n)$ factors. To begin with, we will focus on the case of \emph{generic} $A_t$s -- this simply means that $A_t$s do not satisfy any of some appropriate finite collection of polynomial equations. Eventually, as we explain in Section~\ref{sec:overview-our-approach}, these equations will simply correspond to full-rankness of certain matrices that arise in our analyses. We will discuss a new method to prove strong polynomial condition number bounds for random $A_t$s in the following section. The results for smoothed/generic $A_t$s then follow via standard, simple tools.

Just like the special case of tensor decomposition (i.e., when $A_t$ are linear forms), the decomposition is not uniquely identifiable from a sum of their quadratics (i.e., $d=2$) except in degenerate cases (see Section~\ref{app:quadratic-failure}). Thus, the simplest non-trivial setting turns out to be $d=3$. 

In this section, we will focus on the simplest setting of $K=2$ (and thus, $A_t$ are simply $n \times n$ matrices) and $d=3$. This, by itself, is an important special case and captures the question of identifiability of parameters from the $6$th moments of a mixture of $m$ $n$-dimensional Gaussians with zero-mean and smoothed covariance matrices, and our main results (Theorems~\ref{thm:main-theorem-hd-intro-smoothed} and~\ref{thm:main-theorem-hd-intro}) improve the current best identifiability results (Corollary~\ref{cor:mog}).

\paragraph{Structure of the Coefficient Tensor} Up to a constant scaling, the coefficient tensor of $P$ equals $\sum_{t\leq m} \Sym(A_t^{\otimes 3})$. Here, $\Sym=\Sym_6$ acts on $A_t^{\otimes 3}$ by averaging over entries obtained by permuting the $6$ elements involved. That is, for a uniformly random permutation $\pi:[6]\rightarrow [6]$,
\[
\Sym(A_t^{\otimes 3})(a_1, a_2, \ldots, a_6)= \E_{\pi \sim \bbS_6}\Brac{A_t(a_{\pi(1)},a_{\pi(2)})A_t(a_{\pi(3)}, a_{\pi(4)})A_t(a_{\pi(5)},a_{\pi(6)})}\mper
\] 

\paragraph{Relationship to Tensor Decomposition} 
It is natural to compare our input to the related, \emph{desymmetrized} tensor $\sum_t A_t^{\otimes 3}$, given which, we can immediately obtain the $A_t$s by applying standard tensor decomposition algorithms~\cite{Harshman1970Foundations, MR1238921} (see Fact~\ref{fact:tensor-decomposition-algo}) whenever $A_t$s are linearly independent as vectors in $\binom{n+1}{2}$ dimensions. Our input, however, is not even close to a low-rank tensor because of the action of $\Sym_6$ that generates essentially maximal rank terms even starting from a single generic $A_t$. Indeed, this effect is visible for just \emph{bivariate} polynomials. In Section~\ref{sec:identifiability-cubics-of-quadratics}, we construct two different (and in fact, $\Omega(1)$-far in Frobenius norm) collections of robustly linearly independent bivariate quadratic polynomials such that the sums of their cubics have the same coefficient tensors. Thus, even though such $A_t$s can be uniquely and efficiently recovered from $\sum_t A_t^{\otimes 3}$ via standard tensor decomposition, it is information theoretically impossible to do so given $\sum_t \Sym (A_t^{\otimes 3})$.

\paragraph{The Ge-Huang-Kakade~\cite{GeHK15} Approach} 
The discussion above presents a conceptually simple way forward: if we could somehow compute the desymmetrized tensor (i.e., undo the effect of the group action) from the input, then we have reduced the problem to standard tensor decomposition. This is a bit tricky as the linear operation $\Sym_6$ on $6$th order tensors is a contraction that maps a $\binom{n+1}{2}^3 \sim n^6/8$-dimensional space into a $\binom{n+5}{6} \sim n^6/720$ dimensional subspace and is clearly not invertible (in fact, has a $\Omega(n^6)$-dimensional kernel) on arbitrary $6$th order tensors. The main idea in GHK is to observe that $\Sym_6$ can be invertible \emph{when restricted} to $6$th order tensors in some smaller subspace. In particular, let $B_1, B_2, \ldots, B_m$ be a basis for the span of the matrices $A_t$. Then, the desymmetrized coefficient tensor of $P$ is a linear combination of $B_i \otimes B_j \otimes B_k$ -- a subspace of $m^3$ dimension which is $\ll n^6/720$ if $m \ll n^2$. Proving such a claim requires analysis of the rank (and singular values, for polynomial error-stability) of the matrix representing $\Sym_6$ on the linear span of $A_t$s and GHK managed to prove it for any $m \ll \sqrt{n}$.

To obtain the span of $A_t$s, GHK rely on access to $P_4 = \sum_{t\leq m} \Sym_4(A_t^{\otimes 2})$ in addition to the input tensor above. Plugging in $e_a, e_b$ in the first two modes of this tensor yields an $n \times n$ matrix (i.e., a 2-D slice) of the form: $\sum_t A_t[a,b] A_t + \sum_t A_t[a] \ot A_t[b]$ where $A_t[i]$ is the $i$-th column of $A_t$. As $a,b$ vary, the first term generates the subspace of the span of $A_t$s. However, each such 2-D slice has an additive ``error'' that lies in the span of the rank 1 forms in the 2nd term above. The GHK idea is to zero out the rank 1 terms by projecting the 2-D slices to a subspace $\calS^\perp$, where $\calS$ contains the span of the rank 1 terms. To compute $\cal{S}$, they choose a subset $H \subseteq [n]$ and plug in $a,b,c \in H$ into three modes of $P_4$. The resulting $1$-D slices are linear combinations of the columns $A_t[a]$ for $a\in H$ and $t \in [m]$. If $m |H| \ll n$, then all $A_t[a]$ are linearly independent generically, while if $|H|^3 \gg m|H|$, then there are enough slices to generate the span of $A_t[a]$ for all $a \in H$ and $t \in [m]$. This trade-off is optimized at $m \sim n^{2/3}$ and $|H| \sim n^{1/3}$. Given a good estimate of $\calS$, we can now plug in $a,b \in H$ in two modes of $P_4$ and recover the span of $A_t$ (restricted to columns in $H$) by projecting the resulting 2-D slices off $\cal{S}$. Repeating for disjoint choices of $H$ completes the argument. In order to analyze the linear independence (and condition numbers) of the vectors arising in this analysis, GHK need to work with a somewhat smaller $m \sim \sqrt{n}$ in their argument.


\paragraph{Key Bottleneck in the GHK Approach} In our situation, we only have the sum of cubics $P$ as input (but not $P_4$). But even given $P_4$, the crucial bottleneck is the need for recovering the span of a subset of columns of the $A_t$s. With more sophisticated analyses, given the above trade-offs, it's plausible that a sum of $d$-th powers of $A_t$ allows handling $m$ as large as $n^{1-O(1/d)}$, but there appears to be an inherent barrier at $m \sim n$. The GHK approach also seems to get unwieldy as it involves plugging in standard basis vectors in several modes of the tensor. This leads to more ``spurious'' terms that one must zero-out (instead of just the rank-1 terms for $P_4$).

Thus, even given higher powers, the GHK approach appears to have a natural break-point at $m \sim n$, and even handling $m \gg \sqrt{n}$ seems to require somewhat unwieldy analysis. In fact, in their recent work, Garg, Kayal and Saha~\cite{GKS20} commented (see Page 17) \emph{"However, we believe such an approach cannot be made to handle larger number of summands (say poly(n)) even in the quadratic case as the lower bounds for
sums of powers of quadratics need substantially newer ideas than the linear case..."}.
%
%
%
\paragraph{The Garg-Kayal-Saha~\cite{GKS20} Approach} In their beautiful recent work, GKS managed to find a different approach that escapes the above obstacles and showed an algorithm  (that works on both finite fields and $\bbQ$) that for any $K$ and $m = n^{O(1)}$, manages to decompose $P(x) = \sum_{t \leq m} A_t(x)^d$ for large enough $d$ (and generic degree-$K$ polynomials $A_t$). As discussed before, their approach requires $d$ to be a large enough constant as a function of $K$ and $\log_n m$ (though they remarked that the bounds could likely be improved, already for $K=2$, they need $d \geq 2^{335}$ and $m \leq n^{d/1100}$). Their main idea, however, is relevant to our approach so we briefly describe it here. 

We restrict our attention to the quadratic case ($K=2$) from here on.
The GKS approach relies on the linear span of \emph{partial derivatives} of the input polynomial $P$. In fact, taking $r$th partial derivatives of $P$ is essentially the same (though, more principled and easier to analyze) as ``plugging in'' all possible standard basis vectors in $r$ modes of the input coefficient tensor as in GHK. GKS observed that for $r < d$, the $n_{r}= \binom{n+r-1}{r}$ many $r$-th partial derivatives of $P$ are all of the form $\sum_{t \leq m} A_t(x)^{d-r} Q_t(x)$ for some degree-$r$ polynomials $Q_t$. This linear subspace is \emph{strictly contained within} the space of all polynomial multiples of $A_t(x)^{d-r}$ -- the containment is strict because the latter space is of dimension $\sim m n_{r} \gg n_{r}$ for generic $A_t$s. However, if we were to project each of the $n_{r}$ partial derivatives down to be a function of some small enough $\ell=o(n)$ variables $y$, then, the dimension counting above is no longer an obstruction to the span being all multiples of the projected $A_t(x)^{d-r}$. Indeed, for generic $A_t$, the subspace $\calU$ of the \emph{projected} partial derivatives does in fact equal the subspace $\calV$ of all multiples of $B_t(y)^{d-r}$, where $B_t = M^{\top}A_tM$ (the projection of $A_t$) and $B_t(y) = A_t(My) = y^\top M^\top A_t My$ for an $n \times \ell$ projection matrix $M$.

\paragraph{Key Bottleneck in the GKS Approach} If we take $r=d-1$, then, it appears that the partial derivatives give us access to the subspace of span of \emph{multiples of} $B_t$s (of degree $2d-r = d+1$ for $K=2$). If we could extract the span of quadratics $B_t$ from this subspace, we could implement desymmetrization and tensor decomposition to obtain at least the $B_t$s (i.e., the projected $A_t$s). 

Unfortunately, this hope did not materialize for GKS who managed only to recover the span of $B_t(y)^{d-r}$ for $r<2d/3$. This is because their analysis of a certain ``multi-GCD'' requires that the subspaces $B_t(y)^{d-r}y_T$ for $|T| =r$ for each $t\in [m]$ only have trivial (i.e., $0$) pairwise intersection. This condition is impossible if $r\geq 2(d-r)$ or $r \geq 2d/3$; for example, if $d=3$ and $r=2$, then the degree-4 polynomial $B_t(y)B_{t'}(y)$ is clearly in the subspaces corresponding to both $t$ and $t'$, which is a non-trivial intersection! Thus, the GKS analysis is restricted to work with $r < 2d/3$ and in particular, only manages to recover the span of $B_t(y)^{d-r}$ (for $d-r>d/3$). This route rules out the desymmetrization + tensor decomposition approach.

As a result, GKS used a more complicated sequence of operations that involves taking projections of the partial derivatives and algorithms for simultaneous decomposition of vector spaces into irreducibles which they analyze by studying the associated ``adjoint algebra''. The two-step projection step requires that $d$ be very large as a function of $\log_n m$ (and degree $K$ of the $A_t$s).

\paragraph{Summary} The ``desymmetrize + tensor decompose'' approach of GHK is elegant and simple but suffers from an inherent bottleneck for going beyond $m \sim n$ (or even $n^{2/3}$) for sums of cubics (or higher powers) of quadratics and gets unwieldy as $d$ gets large. The GKS approach manages to handle any $m = n^{O(1)}$ components but only for very large $d$ and relies on a somewhat complicated algebraic algorithm. While GKS do not do this, finding a polynomially conditioned variant of their algorithm will likely require significant effort.  

\subsection{Our approach and outline of our algorithm} \label{sec:overview-our-approach}
Somewhat surprisingly, we manage to find an algorithm that achieves the best of both worlds. Our algorithm relies on the conceptually simple approach of desymmetrizing the input tensor (as in GHK) while at the same time managing to not only hit $m \sim n$ when $K=2$ and $d=3$ but also get a substantially improved trade-off compared to GKS for all $m,d,K$. Further, we find a polynomially stable implementation of our algorithm when $A_t$s are random by establishing condition number upper bounds on the structured random matrices that arise in our analysis. 

In the following, we explain the main components in our algorithm and analysis: insights that rescue the simple ``desymmetrize + tensor decompose'' approach, the resulting algorithm, and a new method to prove strong condition number upper bounds on structured random matrices. We will focus on the case when the $A_t$s have independent $\cN(0,1)$ entries in the following section. For this setting, we obtain an algorithm with polynomial error-stability guarantees. Our result for generic (or smoothed) $A_t$ is a simple corollary of this result using standard tools.

\paragraph{Recovering the Span of $B_t$s} Recall that the GKS observation shows that given a polynomial $P(x) = \sum_{t \leq m} A_t(x)^{d}$ for quadratic $A_t$s, the subspace $\calU$ spanned by $r$-th partial derivatives of $P$, when \emph{projected} to a sufficiently small dimension $\ell=o(n)$ variables $y$, equals the span $\calV = \spn(B_t(y)^{d-r} y_T \mid t\in[m], T\in [\ell]^r)$ for $B_t(y) = A_t(My)$ where $M$ is an $n \times \ell$ projection matrix.

GKS then perform a multi-GCD step that recovers the span of $B_t(y)^{d-r}$ from $\calV$ and their analysis requires the subspaces $\{B_t(y)^{d-r} y_T \mid T \in [\ell]^r\}$ for each $t\in[m]$ to have only trivial pairwise intersection (i.e.\ $=\{0\}$) . Our key idea is to observe that this assumption is not crucial! We can extract the span of powers of $B_t$ as long as these subspaces do not have a large intersection. As discussed before, when $r \geq 2d/3$, their analysis fails because of some obvious intersections between the above subspaces. We substantially improve their analysis by observing that for random polynomials these obvious intersections between the subspaces \emph{are the only ones} possible! 

More precisely, let's restrict to $d=3$ and consider the subspace of projected (we in fact show that simply \emph{restricting} the variables suffices) partial derivatives of order $r=2$ of $P$. Then, the subspace of restricted 2nd order partial derivatives of $P$ contains homogeneous polynomials of degree $4$.
For random $A_t$s, we fully characterize the set of quadratic polynomials $\{q_t \mid t \leq m\}$ that satisfy the polynomial equality $\sum_{t \leq m} B_t(y) q_t(y) = 0$.
Observe that for any $s\neq t\in[m]$, $q_s = B_{t}$ and $q_{t} = -B_s$ is clearly in the solution space. Such solutions span a subspace of dimension $\binom{m}{2}$. In Lemma~\ref{lem:V-singular-value-hd}, we prove that these solutions are in fact the \emph{only} solutions whenever $m \leq \wt{O}(n)$.

This understanding immediately allows us to use a simple subroutine to recover the span of $\{B_t(y) \mid t \leq m\}$.  Specifically, we take a random homogeneous quadratic polynomial $p(y)$ and let $\calV_p$ be the subspace of quartic multiples of $p$, that is, $\calV_p = \spn(p(y)y_S \mid |S| = 2)$. Then, any non-zero $f(y)=p(y)q_0(y) \in \calV \cap \calV_p$ must be a solution to $\sum_{t \leq m} B_t(y) q_t(y) = p(y)q_0(y)$. The above characterization of the solution subspace allows us to conclude that whenever $q_0$ is non-zero, 
it lies in the span of $B_t(y)$. Thus, we have confirmed that $\calV \cap \calV_p = \spn(p(y)B_t(y) \mid t \leq m)$, and dividing this subspace by $p$ immediately yields $\spn(B_t(y) \mid t\leq m)$!

Thus, to summarize, our algorithm for finding the span of $B_t$s is simple: 
\begin{enumerate}
\item Restrict all 2nd order partial derivatives of $P$ to some $\ell$ variables ($\ell=O(\sqrt{n})$ suffices),
\item Find intersection of this subspace with $\calV_p$ for a random homogeneous quadratic polynomial $p$ and divide the resulting subspace by $p$.
\end{enumerate}

The analog of this result for $d=3D$ powers of quadratics relies on a similar lemma that characterizes the solution space of $\sum_{t \leq m} B_t(y)^D q_t(y) = 0$. For sums of powers of degree $K>2$ polynomials however, the characterization gets a little more involved as unlike in the case of quadratic $B_t$s, $q_t$ will have a larger degree than $B_t^D$, which makes the solution space larger. We present this characterization in more detail in Section~\ref{sec:analysis-of-V-high-deg}.

\paragraph{Noise Resilient Implementation} For obtaining a noise-resilient version of the above method, we first need a noise-robust version of the GKS observation that the subspace $\calU$ of restricted partial derivatives equals the subspace $\calV$ spanned by multiples of $B_t(y)$, and also a robust way of obtaining a basis for $\calV$. This amounts to understanding the smallest nonzero singular value of a certain matrix that we analyze in Lemma~\ref{lem:singular-value-of-U-hd}. Similarly, we bound the nonzero singular values of the matrix of linear equations $\sum_{t \leq m} B_t(y) q_t(y) = 0$ (described above) in Lemma~\ref{lem:V-singular-value-hd}. Finally, we use a simple method in Lemma~\ref{lem:robust-int-hd} to robustly compute the intersection of two subspaces given a basis for each by looking at the largest singular values of the sum of the corresponding projection matrices, allowing us to obtain a subspace close to the span of $B_t$.

\paragraph{Desymmetrization} The above discussions show how we can estimate the span of $B_t(y)$ for a restriction of the quadratic $A_t$ to some $\ell = O(\sqrt{n})$ variables. Given this subspace, we apply desymmetrization \emph{directly to the restricted polynomial} $P(My)$. To analyze this step, we need to understand the invertibility (and condition numbers) of the matrix representing the $\Sym_6$ linear transform on the subspace of the linear span of $B_t$. We establish the condition number upper bound in Lemma~\ref{lem:desymm-mainsv-hd}, thus obtaining the desymmetrized tensor $\sum_{t \leq m} B_t^{\otimes 3}$. 

\paragraph{Aggregating Restrictions} For a given restriction (via an $n \times \ell$ matrix $M$), the above steps give us access to the tensor $\sum_{t \leq m} B_t^{\otimes 3}$ where $B_t = M^\top A_t M$ is the $\ell \times \ell$ matrix of the restricted $A_t$. We would like to piece together such restrictions to obtain $\sum_{t \leq m} A_t^{\otimes 3}$. We show how to do this by working with a simple $n^6$-size \emph{pseudorandom} set of restriction matrices $M$ such that the average over the corresponding restricted 3rd order tensor gives us the unrestricted 3rd order tensor up to a known scaling. Our construction is a simple modification of the standard construction of $6$-wise independent hash families.

\paragraph{Tensor decomposition and taking s$D$-th roots} Given an estimate of $\sum_{t \leq m} A_t^{\otimes 3}$, we can apply the standard polynomially-stable tensor decomposition algorithms (Fact~\ref{fact:tensor-decomposition-algo}) to recover the $A_t$s. When we work with higher ($d=3D$) powers of quadratics (or degree-$K$ polynomials, more generally), this step only gives us $\Sym_{KD}(A_t^{\otimes D})$. The task of recovering $A_t$ given $\Sym_{KD}(A_t^{\otimes D})$ is a certain simple ``deconvolution'' problem. We give a noise-robust algorithm for this task that relies on a simple semidefinite program analyzed in Lemma~\ref{lem:analysis-D-root-hd}.

\subsection{Overview of singular value lower bounds}
For establishing polynomial stability of our algorithm for random $A_t$s and proving Theorem~\ref{thm:main-theorem-hd-intro}, we need to understand the condition number and in particular, the smallest singular value of certain random matrices that arise in our analyses. Analyzing the smallest singular value of random matrices turns out to be more challenging than the much better understood largest singular value. For matrices with independent and \emph{identically} distributed random subgaussian entries, a sharp bound was only achieved in the breakthrough work of \cite{RudVer:08} via a sophisticated analysis via the ``leave-one-out'' distance method. The matrices that arise in our analyses are significantly more involved. The entries are not independent but are instead computed as low-degree polynomials of independent random variables that are of polynomially smaller number than the dimension of the matrix. As a result, the entries exhibit large correlations, and the leave-one-out method appears hard to implement for such matrices. 

Instead, we adopt a different, more elementary but nimble method that obtains estimates of the smallest singular values via upper bounds on the \emph{largest} singular values of certain \emph{deviation} matrices. To see this method on a simple toy example, consider an $n \times m$ matrix (for $m \ll n$) $G$ of independent $\cN(0,1)$ entries. Then, we can write $G^{\top}G = n (1 \pm O(\frac{1}{\sqrt{n}})) \cdot \Id + \offdiag(G^{\top}G)$ where $\offdiag(G^{\top}G)$ zeros out the diagonal entries of $G^{\top}G$. To establish a lower bound on the $m$th singular value of $G$, it is thus enough to observe that $\spec{\offdiag(G^{\top}G)} \leq \wt{O}(\sqrt{mn})$ with high probability.

This argument works as long as $m \leq n/\polylog(n)$ and gives a sharp (up to the leading constant) estimate on the smallest singular value. Note that in this argument, we effectively ``charge'' the spectral norm of the off-diagonal ``deviation'' matrix to the smallest entry of the diagonal part. Such a strategy works so long as all columns of $G$ are of roughly similar length.

It turns out that despite its simplicity, this technique is surprisingly resilient for our purposes and unlike methods from prior works, it easily applies to the involved matrices that arise in our analysis, yielding bounds that are essentially sharp so long as we can keep the dimensions of the matrix somewhat ``lopsided'' (i.e.\ $m \ll n$ in the example above). This turns out to not be a handicap in our setting. 

In our analysis, the problem now reduces to bounding the spectral norm of certain correlated, low-degree polynomial-entry random matrices arising from the off-diagonal part of the matrices we analyze. While this can be quite complicated, we rely on the recent advances in understanding the spectral norm of such matrices~\cite{BHK19,graphmatrixbounds,sparseindset} in the context of proving Sum-of-Squares lower bounds for average-case optimization problems. This technique relies on decomposing random matrices into a linear combination of certain structured random matrices called \emph{graph matrices}. We rely on the tools from prior works that reduce the task of analyzing the spectral norm of such matrices to analyzing combinatorial properties of the underlying ``graph''. 

This technique gets us started but hits a snag as it turns out that some of the deviation matrices simply \emph{do not have small spectral norms}. We handle such terms by proving that the large spectral norm can be ``blamed'' on having large \emph{positive} eigenvalues that cannot affect the bounds on the smallest singular value. Formally, we provide a charging argument, reminiscent of the positivity analyses in the construction of sum-of-squares lower bounds~\cite{BHK19,GJJPR20,HK22,sparseindset}, to handle such terms and establish the required bounds on the spectral norm.

While somewhat technical, the proofs of singular value lower bounds for all the matrices in our analyses follow the same blueprint. We give a more detailed exposition of these tools (by means of an example) in Section~\ref{sec:overview-singular-val} before applying them to the matrices relevant to us.

\section{Preliminaries and Notation}
\label{sec:prelims}


\paragraph{Notations and definitions}
\begin{enumerate}
    \item \textbf{Multisets and monomials:} 
    We denote $[n] = \{1,2,\dots, n\}$.
    We say that $I \in [n]^k$ is a \emph{multiset} of size $k$ when the order of elements in $I$ does not matter.
    For a multiset $I$, we let $x_I$ denote the product of variables in $I$: $\prod_{i \in I}x_i$.
    For $n,k\in \N$, we define $n_k \coloneqq \binom{n+k-1}{k}$ as the number of multisets of $n$ of size $k$ which is also equal to the number of degree-$k$ monomials in $n$ variables.

    \item \textbf{Vectors, matrices, and tensors:}
    Given a vector $v \in \R^n$, we denote its $\ell_p$ norm by $\|v\|_p = (\sum_i |v_i|^p)^{1/p}$. Given a matrix $M\in \R^{n\times n}$, we let $\|M\|_F$ denote its Frobenius norm, $\spec{M}$ denote its spectral norm, and $M[i,j]$ denote the $(i,j)$ entry of $M$.
    Given a tensor $T \in (\R^n)^{\ot d}$ of order $d$, $\|T\|_F$ denotes its Frobenius norm, and $\fl(T) \in \R^{n^d}$ denotes the flattened vector of $T$ so that $\|\fl(T)\|_2 = \|T\|_F$.
    For index $I\in [n]^d$, we denote $T[I]$ to be the entry at index $I$.
    For two tensors $T, T'$, we denote $\Iprod{T, T'} = \Iprod{\fl(T), \fl(T')}$.

    \item \textbf{Polynomial coefficients:}
    Given a degree-$d$ polynomial $p(x_1,\ldots,x_n)$, we use $p$ to denote the coefficient vector of this polynomial, with entries of the vector indexed by monomials in lexicographic ordering.
    For a multiset $I$, we let $p[I]$ denote the coefficient of $p$ corresponding to the monomial $x_I$.
    With slight abuse of notation, when $p$ is homogeneous we also use $p$ to denote the \emph{symmetric} coefficient tensor of order $d$, i.e., $p(x) = \Iprod{p, x^{\ot d}}$.
    We interchangeably use the vector and tensor views when clear from context.
    For two polynomials $p, q$, we denote $\Iprod{p,q}$ as the inner product of their coefficient vectors, and $\|p\|^2 = \Iprod{p,p}.$\footnote{The Frobenius norm of the coefficient \emph{tensor} is different from the norm of the coefficient \emph{vector} due to necessary rescalings of tensor entries, but they only differ by constant factors depending on $d$.}

    \item \textbf{Symmetrization:}
    Given any tensor $T \in (\R^n)^{\ot d}$, we let $\Sym(T)$ denote the symmetric tensor such that for any index $I\in [n]^d$,
    \[\Sym(T)[I] = \E_{\pi \in \bbS_d} T[\pi(I)]\]
    where $\bbS_d$ denotes the set of permutations.
    For example, for a degree-$d$ homogeneous polynomial $p(x) = \Iprod{p, x^{\ot d}}$ and $k\in \N$, $\Sym(p^{\ot k})$ is the coefficient tensor of $p(x)^k$, i.e., $p(x)^k = \Iprod{\Sym(p^{\ot k}), x^{\ot dk}}$.

    \item \textbf{Random polynomial:}
    A random degree-$d$ homogeneous polynomial is a degree-$d$ homogeneous polynomial whose coefficients are picked independently and randomly from $\calN(0,1)$.

    \item \textbf{Partial derivatives:}
    Given a polynomial $p(x_1,\ldots,x_n)$ we let $\partial_{i}p(x_1,\ldots,x_n)$ denote the polynomial which is the partial derivative of $p$ with respect to variable $x_i$, and $\partial_ip$ denotes the coefficient vector for the same. For a multiset $I \in [n]^k$, $\partial_{I}p(x)$ denotes the polynomial obtained by taking partial derivatives with respect to $I$.

    \item \textbf{Linear span:} For a set of vectors $v_i \in \R^{n}$ we use $\spn(v_i)$ to denote their linear span. For a set of degree-$d$ homogeneous polynomials $p_i(x)$, we use $\spn(p_i(x))$ (or $\spn(p_i)$ for simplicity) to denote the linear span of their coefficient vectors lying in $\R^{n_d}$.

    \item \textbf{Linear subspace and projection:} For a $k$-dimensional subspace $\calU$ of $\R^n$, let $\proj(\calU)$ denote the $n \times n$ matrix that projects a vector in $\R^n$ to the subspace $\calU$. More specifically, let $U \in \R^{n\times k}$ be a matrix with columns consisting of orthonormal basis vectors $\{u_i\}_{i=1}^k$ of $\calU$, then $\proj(\calU) = UU^T$.
    Further, for two $k$-dimensional subspaces $\calU, \calV$ of $\R^n$, we define the difference between them as
    $\Norm{\calU - \calV}_F \coloneqq \Norm{\proj(\calU) - \proj(\calV)}_F$.
\end{enumerate}


\paragraph{Eigenspace Perturbation Bounds}

We now state some theorems that will be used to analyze the error resilience of our algorithms. The following theorem  gives us a stability result for the singular values of matrices: 
\begin{theorem}[Weyl's theorem]\label{thm:weyl}
Given matrices $A,E \in \R^{m \times n}$ with $m \geq n$, for all $k \leq n$ we have that:
\[\sigma_k(A) - \spec{E} \leq \sigma_k(A+E) \leq \sigma_k(A) + \spec{E}.\]
\end{theorem}

The following theorem that can be found in \cite{Stewart90} analyses the singular vector spaces of matrices under perturbation:
\begin{theorem}[Wedin] \label{thm:wedin}
Given matrices $A,E \in \R^{m \times n}$ with $m \geq n$, let $A$ have singular value decomposition:
\[[U_1 ~~ U_2 ~~ U_3] \begin{bmatrix}
\Sigma_1 & 0 \\
0 & \Sigma_2 \\
0 & 0 
\end{bmatrix} [V_1 ~~ V_2]^T.\]
Let $\wt{A} = A + E$ with analogous singular value decomposition. Suppose that there exists a $\delta$ such that:
\[\min_{i,j}|[\Sigma_1]_{i,i} - [\Sigma_2]_{j,j}| > \delta, ~~\text{ and } ~~ \min_{i} |[\Sigma_1]_{i,i}| > \delta,\]
then
\[\Norm{\proj(\calU_1) - \proj(\wt{\calU_1})}_F^2 + \Norm{\proj(\calV_1) - \proj(\wt{\calV_1})}_F^2 \leq \frac{2\|E\|_F^2}{\delta^2} \mper\]
\end{theorem}

We will now state some facts about least-squares minimization and also analyze its error resilience.

\begin{definition}[Moore-Penrose Pseudo-inverse]
The pseudo-inverse of a rank $n$ matrix $M \in \R^{m \times n}$ for $m \geq n$ is denoted by $M^{\dagger}$ and equals $(M^T M)^{-1}\cdot M^T$. We have that the singular values of $M^\dagger$ are $1/\sigma_1(M) \leq  \ldots \leq 1/\sigma_n(M)$.
\end{definition}

\begin{lemma}\label{lem:ls-facts}
Suppose $M \in \R^{m \times n}$ with $m \geq n$ is a rank $n$ matrix and $b \in \R^m$. The solution to the least-squares problem: $\min_{x \in \R^n} \|Mx - b\|_2$ is $x = M^{\dagger}b.$ 
\end{lemma}

We will use the following lemma about stability of the pseudo-inverse:
\begin{lemma}\label{lem:p-inverse-stab}
Given matrices $A,E \in \R^{m \times n}$ with $m \geq n$, $\rk(A) = n$ and $\spec{E} < \sigma_n(A)$, we have that:
\[\spec{(A+E)^{\dagger} - A^\dagger} \leq \frac{\sqrt{2}\spec{E}}{\sigma_n(A)(\sigma_n(A) - \spec{E})}.\]
\end{lemma}

\begin{proof}
The following lemma can be found in \cite{Stewart90}:
\[\spec{(A+E)^{\dagger} - A^\dagger} \leq \sqrt{2}\spec{E}\spec{A^{\dagger}}\spec{(A+E)^\dagger}.\]
We know that $\spec{A^\dagger} = 1/\sigma_n(A)$ and similarly $\spec{(A+E)^\dagger} = \frac{1}{\sigma_n(A+E)} \leq \frac{1}{\sigma_n(A)-\spec{E}}$ by Weyl's theorem.
Plugging these into the above the lemma follows.

\end{proof}

Using the above we can prove the stability of the solution of least-squares minimization:

\begin{theorem}[Error-Resilience of Least-squares]\label{lem:least-sq-robust}
For all $m\times n$ matrices $A,E$ with $m \geq n$, $\rk(A) = n$ and $\spec{E} < \sigma_n(A)$ and vectors $b, e \in \R^m$, if $y = \argmin_{x \in \R^n} \|Ax - b\|_2$ and $\wt{y} = \argmin_{x \in \R^n} \|(A+E)x - (b+e)\|_2$, then we have that
\[\|y - \wt{y}\|_2 \leq \frac{\sqrt{2}\spec{E}\|b\|_2 + \sigma_n(A)\|e\|_2}{\sigma_n(A)(\sigma_n(A) - \spec{E})}.\]
\end{theorem}

\begin{proof}
We have that $y = A^\dagger b$ and $\wt{y} = (A+E)^\dagger (b+e)$. We have that:
\begin{align*}
\|\wt{y} - y\|_2 &= \Norm{(A+E)^\dagger(b+e) - A^\dagger b}_2 \\
&= \Norm{((A+E)^\dagger - A^\dagger)b + (A+E)^\dagger e}_2 \\
&\leq  \Norm{((A+E)^\dagger - A^\dagger)b }_2 + \Norm{(A+E)^\dagger e}_2 \\
&\leq \spec{((A+E)^\dagger - A^\dagger)}\|b\|_2 + \spec{(A+E)^\dagger} \|e\|_2 \mper
\end{align*}
We can use Lemma~\ref{lem:p-inverse-stab} to bound the first term by:
\[\spec{((A+E)^\dagger - A^\dagger)}\|b\|_2 \leq \frac{\sqrt{2}\spec{E}\|b\|_2}{\sigma_n(A)(\sigma_n(A) - \spec{E})} \mper \]
Using Weyl's theorem we can bound the second term by:
\[\spec{(A+E)^\dagger} \|e\|_2 \leq \frac{\|e\|_2}{\sigma_n(A) - \spec{E}} \mper \]
Adding the two bounds completes the proof.
\end{proof}

\paragraph{Preliminaries for singular value lower bounds }

\begin{fact}[Deleting rows won't decrease singular values]
\label{fact:deleting-rows-sval}
    Given a matrix $A \in \R^{n\times m}$, let $B \in \R^{n'\times m}$ be a submatrix obtained by deleting rows of $A$.
    Then,
    \begin{equation*}
        A^\top A \succeq B^\top B \mper
    \end{equation*}
\end{fact}

The following lemma will be used throughout the paper.
It is a simple result that follows directly from standard concentration results on Gaussian variables.

\begin{claim}[Norm of powers of random polynomial]
\label{claim:norm-of-polynomial-powers}
    Fix $D, K\in \N$.
    Let $p(x)$ be a degree-$K$ homogeneous polynomial in $n$ variables such that each coefficient of $p$ is sampled i.i.d.\ from $\calN(0,1)$.
    Then, with probability at least $1 - n^{-\Omega(DK)}$,
    \begin{equation*}
        \Norm{ p(x)^D }^2 = \Theta(n^{DK}) \mper
    \end{equation*}
\end{claim}
\begin{proof}
    Recall that $\Norm{ p(x)^D }$ is the norm of the coefficient vector of the degree-$DK$ polynomial $p(x)^D$.
    Viewing $p$ as an order-$K$ tensor, $\Norm{p(x)^D}$ is within a constant factor away from $\Norm{\Sym(p^{\ot D})}_F$.
    Clearly, $\Norm{p(x)^D}^2$ is a degree-$2D$ polynomial over the coefficients of $p$, and the expectation is
    \begin{equation*}
        \E_{p} \Norm{ p(x)^D }^2 = \Theta(1) \sum_{I \in [n]^{DK}} \E_p \Brac{(\Sym(p^{\ot D})[I])^2} = \Theta(n^{DK}) \mper
    \end{equation*}
    The statement of the lemma follows by standard Gaussian concentration on low-degree polynomials of Gaussians (see e.g.\ \cite{SS12}).
\end{proof}



\section{Decomposing Power-Sums of Quadratics}
In this section, we describe our efficient algorithm to decompose powers of low-degree polynomials. To keep the exposition simpler, we will analyze the algorithm for the case of quadratic $p_i$s in this section and postpone the analysis for higher-degree $p_i$s to the next section. 

Specifically, we will prove that there is a polynomially stable and exact algorithm for decomposing power-sums of \emph{random} quadratics. The same algorithm's recovery guarantees hold more generally for power-sums of \emph{smoothed} quadratic polynomials though our current analysis only derives an inverse exponential error tolerance. Our algorithms work in the standard bit complexity model for exact rational arithmetic.

\begin{theorem} \label{thm:main-theorem-hd}
    There is an algorithm that takes input parameters $n,m,D\in \N$, an accuracy parameter $\tau>0$, and the coefficient tensor $\wh{P}$ of a degree-$6D$ polynomial $\wh{P}$ in $n$ variables with total bit complexity $\size(\wh{P})$, runs in time $(\size(\wh{P})n)^{O(D)} \polylog (1/\tau)$, and outputs a sequence of symmetric matrices $\wt{A}_1,\wt{A}_2,\ldots, \wt{A}_m \in \R^{n\times n}$ with the following guarantee.

    Suppose $\wh{P}(x) = \sum_{t=1}^m A_t(x)^{3D} + E(x)$ where each $A_t$ is an $n \times n$ symmetric matrix of independent $\calN(0,1)$ entries, $\|E\|_F \leq n^{-O(D)}$ and $m \leq (\frac{n}{\polylog(n)})^D$ if $D \leq 2$ and $m \leq (\frac{n}{\polylog(n)})^{2D/3}$ if $D > 2$.
    Then, with probability at least $0.99$ over the draw of $A_i$s and internal randomness of the algorithm, for odd $D$,
    \begin{equation*}
        \min_{\pi \in \bbS_m} \max_{t \in [m]} \Norm{\wt{A}_t - A_{\pi(t)}}_F \leq \poly(n) \Paren{ \|E\|_F^{1/D} + \tau^{1/D} } \mcom
    \end{equation*}
    and for even $D$,
    \begin{equation*}
        \min_{\pi \in \bbS_m} \max_{t \in [m]} \min_{\sigma\in\{\pm1\}} \Norm{\wt{A}_t - \sigma A_{\pi(t)}}_F \leq \poly(n) \Paren{ \|E\|_F^{1/3D} + \tau^{1/3D} } \mper
    \end{equation*}
\end{theorem}
Observe that for odd $D$, we are able to recover $A_t$s up to permutation while for even $D$, we recover $A_t$ up to permutation and signings. Such a guarantee is also the best possible given $P(x)$.

\paragraph{The Algorithm: } Our proof of Theorem~\ref{thm:main-theorem-hd} uses the following algorithm (that works as stated for decomposing powers of degree-$K$ $A_t$s more generally, but we will analyze for quadratic $A_t$s in this section). 
\begin{mdframed}
  \begin{algorithm}[Decomposing Power Sums]
    \label{algo:overall-algo-hd}\mbox{}
    \begin{description}
    \item[Input:] Coefficient Tensor of a $n$-variate degree-$3KD$ polynomial $\wh{P}(x) = P(x)+E(x)$ where $P(x) = \sum_{t \in [m]} A_t(x)^{3D}$ for degree-$K$ polynomials $A_t$.

    \item[Output:] Estimates $\wt{A}_1, \wt{A}_2, \ldots, \wt{A}_m$ of the coefficient tensors of $A_1, A_2, \ldots, A_m$.
     
    \item[Operation:] \mbox{}
    \begin{enumerate}
        \item \textbf{Construct Pseudorandom Restrictions:} Construct the collection $\mathcal{S}$ of $\leq 3KD \ell$-size subsets of $[n]$, $|\mathcal{S}|=n^{O(D)}$ using the algorithm from Lemma~\ref{lem:analysis-aggregation-hd}.

        \item \textbf{Desymmetrize Pseudorandom Restrictions of Coefficient Tensor:} For each $S \in \mathcal{S}$: \label{step:pseudorandom-projections}
        \begin{enumerate}
            \item \textbf{Find Subspace of Restricted Partials:} Apply Algorithm~\ref{algo:partial-der-hd} to compute the linear span $\wt{\mathcal{V}}_D$ of coefficient vectors of $M_S$-restrictions of $2D$-th order partial derivatives of $\wh{P}$. 
            \label{step:partial-der-hd}
            
            \item \textbf{Span-finding:} Apply Algorithm~\ref{algo:span-finding-hd} to find the span of restricted $A_t(x)^D$'s. \label{step:span-finding-hd}

            \item \textbf{Desymmetrize:} Apply Algorithm~\ref{algo:desymmetrization-hd} to compute the desymmetrized restricted coefficient tensor. 
            \label{step:desym-hd}
        \end{enumerate}

        \item \textbf{Aggregate Restricted Tensors:} Use restricted desymmetrized tensors from all restrictions in the pseudorandom set to construct the desymmetrized tensor. \label{step:aggregation-hd}

        \item \textbf{Decompose Tensor:} Apply tensor decomposition to the desymmetrized tensor.
        \label{step:tensor-decomp-hd}

        \item \textbf{Take $D$-th Root of a Single Polynomial:} using Lemma~\ref{lem:analysis-D-root-hd}. 
        \label{step:D-root}
    \end{enumerate}
    \end{description}
  \end{algorithm}
\end{mdframed}

\paragraph{Algorithm Overview} In this section, we henceforth restrict our attention to quadratic $A_t$'s. Like in the case of cubics of quadratics discussed in Section~\ref{sec:overview} our algorithm first desymmetrizes the input coefficient tensor and then applies tensor decomposition to recover estimates of the individual components. Specifically (when $E=0$), given the coefficient tensor of $P$ has the form $\Sym_{6D}(\sum_{t \leq m} \Sym_{2D}(A_t^{\otimes D})^{\otimes 3})$, our goal is to ``undo" the effect of the outer application of $\Sym$ and this is accomplished in the first three steps that are direct analogs of the ones discussed in the special case analyzed in Section~\ref{sec:overview}. After performing the desymmetrization step, for higher powers of quadratics we only recover estimates of $\Sym(A^{\otimes D})$ at the end of this procedure. The final (and extra, compared to the cubic case) step in the algorithm takes $D$-th root of noisy estimates of single polynomials, i.e.\ obtains an estimate of $A(x)$ from an estimate of $A(x)^D$.

Specifically, in Step~\ref{step:partial-der-hd}, we compute the $\sim n^{2D}$ different $2D$-th order partial derivatives $\partial_{I} \wh{P}(x)$ of $\wh{P}$ as $I = \{i_1,\dots, i_{2D}\} \in [n]^{2D}$ ranges over all multisets of size $2D$. We then restrict each of these degree-$4D$ polynomials to some fixed set of $\ell = o(n)$ variables which, in order to distinguish from the original set of indeterminates $x$, we will call $y$. The effect of this restriction is to transform $A_t$ into $B_t = M^{\top}A_tM$ for a $n \times \ell$ restriction matrix $M$ defined below. 

\begin{definition}[Restriction matrix]
\label{def:restriction-matrix}
    Given a set $S \subseteq [n]$ with $|S| = \ell$, we denote $M_S \in \R^{n\times \ell}$ to be the matrix whose columns consist of standard unit vectors $e_{j}$ for $j\in S$.
    We write $P\circ M_S$ for the polynomial (in indeterminates $y$) defined by $P\circ M_S(y) = P(M_S y)$.
    
    For each $M_S$, we let $\calR_S$ be the linear operator that takes an $n \times n$ matrix $A \in \R^{n \times n}$, into $\calR_S(A)= (M_S M_S^\top) A (M_S M_S^\top)$ -- i.e., zeros out the $(i,j)$ entry of $A$ if $i$ or $j$ is not in $S$.
\end{definition}

For any restriction matrix $M$, let $B_t=M^{\top}A_tM$. Let $\calV_D$ be the span of polynomials of the form $B_t(y)^D y_T$:
\begin{equation*}
    \calV_D \coloneqq \spn\Paren{B_t(y)^D y_T \mid t \in [m], T \in [\ell]^{2D}} \mper
\end{equation*}
Then, any $2D$th order partial derivative of $P$, when restricted via $M$, is in $\calV_D$. We prove that for small enough $m, \ell$, the linear span of the restricted partials of $P$ is in fact \emph{equal} to the linear span of the polynomials $B_t(y)^D y_T$ (we prove an error-tolerant version in Section~\ref{sec:partial-der-hd}).

\begin{lemma}[Analysis of the Subspace of Restricted Partials of $\wh{P}$]
\label{lem:analysis-restricted-partials-hd}
    Fix $D\in \N$. Let $m, \ell, n \in \N$ be parameters such that $m \leq (\frac{\ell}{\polylog(\ell)})^{2D}$ if $D\leq 2$, and $m\leq (\frac{\ell}{\polylog(\ell)})^{D}$ if $D>2$, and that $m\ell^{2D} \leq (\frac{n}{\polylog(n)})^{2D}$.
    Given $\wh{P} = \sum_{t \in [m]} A_t(x)^{3D} + E(x)$ where each $A_t$ is a degree-2 homogeneous polynomial with i.i.d.\ $\calN(0,1)$ entries, and a restriction matrix $M \in \R^{n \times \ell}$, we have that with probability $1- n^{-\Omega(D)}$ over the choice of $A_t$'s, Algorithm~\ref{algo:partial-der-hd} outputs a subspace $\wt{\calV}_D$ of $\R^{\ell_{4D}}$ that satisfies:
    \[ \Norm{\wt{\calV}_D - \calV_D}_F \leq O\left(\frac{\|E\|_F}{(n\ell)^D}\right) \mcom \]
    with $B_t(y) = A_t \circ M(y)$. 
\end{lemma}

Consider $\calW_D = \spn(B_t(y)^D \mid t \in [m])$, for the next step, we show we can extract a subspace $\wt{\calW}_D \approx \calW_D$ given a basis for $\wt{\calV}_D$, by proving for a random degree-$2D$ polynomial $p(y)$ the intersection (computed in Step~\ref{step:span-finding-hd}) of $\calV_D$ with the linear span of polynomials of the form $p(y) y_T$ (for $|T|=2D$) equals that of $B_t(y)^D$ with high probability over $p$ and $B_t$'s: 

\begin{lemma}[Extracting Span of $B_t(y)^D$]\label{lem:span-finding-hd}
    Let $D,m,\ell$ be the same parameters as Lemma~\ref{lem:analysis-restricted-partials-hd}.
    Given degree-2 homogeneous polynomials $B_t$ for $t\in[m]$ in $\ell$ variables with coefficients drawn i.i.d from $\calN(0,1)$, with probability $1- \ell^{-\Omega(D)}$  Algorithm~\ref{algo:span-finding-hd} outputs $\wt{\calW}_D$ that satisfies:
    \[ \Norm{\wt{\calW}_D - \calW_D}_F \leq O\left(m\ell^{4D} \|\calV_D - \wt{\calV}_D\|_F\right).\]
\end{lemma}

Finally, we show that on the subspace of linear span of $B_t(y)^D$, the outer $\Sym_{6D}$ operation is invertible in an error-tolerant way via the least squares algorithm. This gives us a desymmetrized, $M$-restricted 3rd order tensor.

\begin{lemma}[Desymmetrization of Restricted $\wh{P}$ via Least-Squares]
\label{lem:analysis-desym-hd}
    Let $D,m,\ell\in \N$ such that $m \leq (\frac{\ell}{\polylog(\ell)})^{2D}$.
    For each $t\in[m]$, let $B_t$ be a degree-2 homogeneous polynomial in $\ell$ variables with i.i.d.\ $\calN(0,1)$ entries.
    Suppose $\wt{\calW}_D$ is a subspace of $\R^{\ell_{2D}}$ such that $\Norm{\wt{\calW}_D -\calW_D}_F \leq 1/(m^{3.5}\ell^{O(D)})$, then with probability $1- n^{\Omega(D)}$ over the choice of $B_t$'s,
    Algorithm~\ref{algo:desymmetrization-hd} outputs a tensor $\wt{T}$ such that:
    \[ \Norm{\wt{T} - \sum_{t \in [m]} \Paren{\Sym(B_t^{\ot D})}^{\ot 3} }_F \leq \poly(m) \Paren{\ell^{O(D)}\Norm{\wt{\calW}_D -\calW_D}_F + \|E\|_F} \mper\]
\end{lemma}

We show how to aggregate the desymmetrized estimates above for $n^{O(D)}$ pseudorandom restriction matrices to obtain the estimate of the unrestricted tensor we need. 

\begin{lemma}[Aggregating Pseudorandom Restrictions]
\label{lem:analysis-aggregation-hd}
    Let $D,n,\ell,m \in \N$ such that $6D \leq \ell \leq n$.
    There is an $n^{O(D)}$-time computable collection $\calS$ of subsets of $[n]$ such that each $S\in \calS$ satisfies $\ell \leq |S| \leq 6D\ell$ and that
    \begin{equation*}
        \E_{S\sim \calS} \sum_{t=1}^m \Paren{\Sym \Paren{ \calR_{S}(A_t)^{\ot D} }}^{\ot 3}
        = C \circ \sum_{t=1}^m \Paren{\Sym(A_t^{\ot D}) }^{\ot 3}
    \end{equation*}
    where $C \in (\R^n)^{\ot 6D}$ is a fixed tensor whose entries depend only on the entry locations, and each entry of $C$ has value within $((\ell/2n)^{6D}, 1)$.
\end{lemma}

Given such a partially desymmetrized tensor, an application of off-the-shelf algorithms for 3rd order tensor decomposition allows us obtain $\Sym(A_t^{\otimes D})$ for $t \leq m$ in Step~\ref{step:tensor-decomp-hd}. We will specifically use:

\begin{fact}[Stable Tensor Decomposition, symmetric case of Theorem 2.3 in \cite{BhaskaraCMV14}] \label{fact:tensor-decomposition-algo}
There exists an algorithm that takes input a $n \times n \times n$ tensor $\wt{T}$ and an accuracy parameter $\tau>0$, runs in time $(\size(T)n)^{O(1)} \polylog (1/\tau)$ and outputs a sequence of vectors $\wt{v}_1, \wt{v}_2, \ldots, \wt{v}_r$ with the following guarantee. If $\wt{T} = \sum_i v_i^{\otimes 3} + E$ for an arbitrary $n \times n \times n$ tensor $E$ and the matrix with $v_i$s as rows has a condition number (ratio of largest to $r$-th smallest singular value) at most $\kappa <\infty$. Then, 
\[
\min_{\pi \in \bbS_r} \max_{i \leq r} \Norm{\wt{v}_i - v_{\pi(i)}}_2 \leq \poly(\kappa, n) \Norm{E}_F + \tau\mper
\]
\end{fact}

To apply this fact, we will need the following bound on the condition number $\kappa$ of the matrix with $\Sym(A_t^{\otimes D})$ as columns:

\begin{lemma}[Condition number]
\label{lem:condition-number-hd}
    Under the same assumptions as Lemma~\ref{lem:analysis-restricted-partials-hd},
    let $A_D$ be the $n_{2D} \times m$ matrix whose columns are the coefficient vectors of $A_t(x)^D$ for $t\in[m]$.
    Then, with probability $1 - n^{-\Omega(D)}$, the condition number $\kappa(A_D) \leq O(1)$.
\end{lemma}
Recall that for any natural number $k$, we write $n_k = \binom{n+k-1}{k}$ for the number of distinct degree $k$ monomials in $n$ variables. 

Finally, in Step~\ref{step:D-root}, we extract $A_t$ from $\Sym(A_t^{\otimes D})$ (i.e.\ desymmetrize a single noisy power). Note that in this step, we do not need randomness/genericity of the $A_t$.
\begin{lemma}[Stable Computation of $D$-th Roots]
\label{lem:analysis-D-root-hd}
    Let $D,n \in \N$ and $\delta \geq 0$.
    Let $P \in \R^{n\times n}$ be an unknown symmetric matrix.
    Suppose $\wt{P_D}(x)$ is a homogeneous degree-$D$ polynomial in $n$ variables such that its coefficient tensor satisfies $\Norm{\wt{P_D} - \Sym(P^{\otimes D})}_F \leq \delta$.
    There is an algorithm that runs in $n^{O(D)}$ time and outputs $\wt{Q} \in \R^{n\times n}$ such that if $D$ is odd, then
    \begin{equation*}
        \Norm{\wt{Q} - P}_F \leq O(\sqrt{n}\delta^{1/D}) \mcom
    \end{equation*}
    and if $D$ is even, then
    \begin{equation*}
        \min_{\sigma \in \{\pm1\}} \Norm{ \wt{Q} - \sigma P }_{F} \leq O(n \delta^{1/3D}) \cdot \|P\|_{\max} \mper
    \end{equation*}
\end{lemma}

\paragraph{Putting things together}
We will prove each of the above lemmas and provide details of each step in the following subsections. Here, we use them to finish the proof of Theorem~\ref{thm:main-theorem-hd}.

\begin{proof}[Proof of Theorem~\ref{thm:main-theorem-hd}]
    For $D \leq 2$, we set $\ell = \sqrt{n}$ and $m \leq (\frac{n}{\polylog(n)})^D$ such that $m \leq (\frac{\ell}{\polylog(\ell)})^{2D}$.
    For $D > 2$, we set $\ell = n^{2/3}$ and $m \leq (\frac{n}{\polylog(n)})^{2D/3}$ such that $m \leq (\frac{\ell}{\polylog(\ell)})^{D}$.
    In both cases, we have $m\ell^{2D} \leq (\frac{n}{\polylog(n)})^{2D}$.
    
    We consider the collection $\calS$ of subsets of $[n]$ from Lemma~\ref{lem:analysis-aggregation-hd} with parameter $\ell$ such that $|\calS| = n^{O(D)}$ and $\ell \leq |S| \leq 6D\ell$ for all $S\in \calS$.
    Thus for $m \leq (\frac{n}{\polylog(n)})^D$, the parameters $m,n,|S|$ satisfy $m \leq (\frac{|S|}{\polylog(n)})^{2D}$ and $m|S|^{2D} \leq (\frac{n}{\polylog(n)})^{2D}$.

    Consider a set $S\in \calS$ and the corresponding restriction matrix $M_S$, and let $B_t = M_S^\top A_t M_S \in \R^{|S| \times |S|}$.
    By Lemma~\ref{lem:analysis-restricted-partials-hd}, \ref{lem:span-finding-hd} and \ref{lem:analysis-desym-hd} (assuming $\|E\|_F \leq n^{-\Omega(D)}$), after Steps~\ref{step:partial-der-hd}, \ref{step:span-finding-hd} and \ref{step:desym-hd}, we obtain tensor $\wt{T}_S \in \R^{\ell^{6D}}$ such that
    \[ \Norm{\wt{T}_S - \sum_{t \in [m]} \Paren{\Sym(B_t^{\ot D})}^{\ot 3} }_F \leq n^{O(D)} \cdot \|E\|_F \mcom\]
    with probability $1 - n^{\Omega(D)}$ over the randomness of the input.
    By union bound over $\calS$, we get the same guarantees for all $S\in \calS$ with probability $1- \frac{1}{\poly(n)}$.

    Next, observe that $\sum_{t \in [m]} \Paren{\Sym(B_t^{\ot D})}^{\ot 3} \in (\R^{\ell})^{\ot 6D}$ is simply a sub-tensor obtained by removing zero entries from the tensor $\sum_{t=1}^m \Paren{\Sym(\calR_{S}(A_t)^{\ot D}}^{\ot 3} \in (\R^n)^{\ot 6D}$ according to $S \subseteq [n]$.
    Therefore, for each $S \in \calS$ we have an estimate of $\sum_{t=1}^m \Paren{\Sym(\calR_{S}(A_t)^{\ot D})}^{\ot 3}$, then if we average over all $S \in \calS$, by Lemma~\ref{lem:analysis-aggregation-hd}, we get a tensor $\wt{R}_D' \in \R^{n^{6D}}$ such that
    \begin{equation*}
        \Norm{\wt{R}_D' - C \circ \sum_{t \in [m]} \Paren{\Sym(A_t^{\ot D})}^{\ot 3}}_F \leq n^{O(D)} \cdot \|E\|_F
    \end{equation*}
    where the error bound is by triangle inequality, and $C$ is a known tensor with entries within $((\ell/2n)^{6D}, 1)$.
    Thus, by normalizing $\wt{R}_D'$ according to $C$, we get a tensor $\wt{R}_D$ such that
    \begin{equation*}
        \Norm{\wt{R}_D - \sum_{t \in [m]} \Paren{\Sym(A_t^{\ot D})}^{\ot 3}}_F \leq n^{O(D)} \cdot \|E\|_F \mper
    \end{equation*}
    
    Next, by the tensor decomposition algorithm (Fact~\ref{fact:tensor-decomposition-algo}) and the condition number upper bound from Lemma~\ref{lem:condition-number-hd}, Step~\ref{step:tensor-decomp-hd} runs in $n^{O(D)} \polylog(\tau)$ time and outputs tensors $\wt{A^D_1},\dots, \wt{A^D_m}$ such that
    \begin{equation*}
        \min_{\pi \in \bbS_m} \max_{t \in [m]} \Norm{\wt{A^D_t} - \Sym\Paren{A_{\pi(t)}^{\ot D}}}_F \leq n^{O(D)} \|E\|_F + \tau \mper
    \end{equation*}
    Finally, by Lemma~\ref{lem:analysis-D-root-hd} we can extract $\wt{A}_t \in \R^{n \times n}$ from $\wt{A_t^D}$.
    For odd $D$, using the fact that $x^{1/D}$ is a concave function when $D \geq 1$, we get that
    \begin{equation*}
        \min_{\pi \in \bbS_m} \max_{t \in [m]} \Norm{\wt{A}_t - A_{\pi(t)}}_F \leq O(\sqrt{n}) \Paren{n^{O(D)} \|E\|_F + \tau}^{1/D}
        \leq \poly(n) \Paren{\|E\|_F^{1/D} + \tau^{1/D} } \mper
    \end{equation*}
    For even $D$, since $\|A_t\|_{\max} \leq \polylog n$ with high probability by standard concentration results, we get that
    \begin{equation*}
    \begin{aligned}
        \min_{\pi \in \bbS_m} \max_{t \in [m]} \min_{\sigma\in \{\pm1\}} \Norm{\wt{A}_t - \sigma A_{\pi(t)}}_F 
        &\leq O(n) \Paren{n^{O(D)} \|E\|_F + \tau}^{1/3D} \|A_t\|_{\max} \\
        &\leq \poly(n) \Paren{ \|E\|_F^{1/3D} + \tau^{1/3D} } \mper
    \end{aligned}
    \end{equation*}
    This completes the proof.
\end{proof}

\subsection{Proof of Lemma~\ref{lem:analysis-restricted-partials-hd}: Estimating span of partial derivatives of \texorpdfstring{$\wh{P}$}{P}}
\label{sec:partial-der-hd}

\begin{mdframed}
  \begin{algorithm}[Estimate Span of $2D$th Order Partial Derivatives of $\wh{P}$]
    \label{algo:partial-der-hd}\mbox{}
    \begin{description}
    \item[Input:] A degree $6D$ homogeneous polynomial $\wh{P}(x) = \sum_{t \in [m]} A_t(x)^{3D} + E(x)$ and restriction matrix $M \in \R^{n \times \ell}$.
      
    \item[Output:] A basis for $\wt{\calV}_D$ such that: \[\|\wt{\calV_D} - \calV_D\|_F \leq O\left(\frac{\|E\|_F}{(n\ell)^D}\right)\] for
    $\calV_D = \spn\Paren{B_t(y)^D y_T \mid t \in [m], T \in [\ell]^{2D}}$ and $B_t(y) = A \circ M(y)$.
    \item[Operation:] \mbox{}
    \begin{itemize}
        \item For all multisets $I \in [n]^{2D}$ of size $2D$:
    \begin{enumerate}
    	\item Compute partial derivative of $\wh{P}(x)$ with respect to $I$: $\partial_I \wh{P}(x)$. 
	    \item Project each partial derivative with respect to $M$ to get the polynomials: $(\partial_I \wh{P}) \circ M(y)$.
    \end{enumerate}
        \item Let $\wt{U}$ be a matrix in $\R^{\ell_{4D} \times n_{2D}}$ whose columns are the vectors $(\partial_I \wh{P}) \circ M$.
        \item Output the top $(m\ell_{2D} - \binom{m}{2})$ singular vectors of $\wt{U}$.   
    \end{itemize}
    \end{description}
  \end{algorithm}
\end{mdframed}



Our goal in this step is to obtain an estimate of $\calV_D = \spn\Paren{B_t(y)^D y_T \mid t \in [m], T \in[\ell]^{2D}}$ given the input coefficient tensor of the polynomial $\wh{P}$. Let us first restrict to $E=0$.

The main observation that helps us here is the following structure in the partial derivatives of $P$. Specifically, for $2D$-size multiset $I$:
    $\partial_I P(x) = \sum_{t=1}^m A_t(x)^{D} \cdot p_t(x)$.
In particular, if we restrict $P(x)$ via $M$, then the resulting $2D$-th order partial derivatives are all in the span of multiples of $B_t(y)^{D} = (y^{\top} B_t y)^D$. Thus, let $\calU_D$ be the subspace of the restricted partial derivatives of $P$:
\begin{equation*}
    \calU_D \coloneqq \spn\Paren{(\partial_I P) \circ M(y) \mid I\in [n]^{2D}} \mper
\end{equation*}

Then, by our discussion above, $\calU_D \subseteq \calV_D$. We will in fact argue that these subspaces are \emph{exactly} the same whenever $m$ and $\ell$ are small enough. In fact, we'll prove that any polynomial of the form $\sum_{t=1}^m B_t(y)^D Q_t(y)$, where $\deg(Q_t) = 2D$, is a linear combination of the restricted partial derivatives with coefficients that are not too large (this will be essential for our error analysis).




\begin{lemma}[$\calU_D = \calV_D$]
\label{lem:U-equals-V-hd}
    Fix $D\in \N$. Let $\ell, m, n\in \N$ such that $m\ell^{2D} \leq (\frac{n}{\polylog(n)})^{2D}$.
    Let $A_1,\dots,A_m \in \R^{n\times n}$ be random symmetric matrices with independent Gaussian entries, let $M \in \R^{n\times \ell}$ be a fixed restriction matrix.
    Then, with probability $1 - n^{-\Omega(D)}$ over the draw of $A_t$s, we have that $\calU_D = \calV_D$ and further, for any degree-$2D$ homogeneous polynomials $Q_1,\dots, Q_m$, there exists coefficients $R = \{R_I\}_{I\in[n]^{2D}} \in \R^{n_{2D}}$ such that
    \begin{equation*}
        \sum_{I\in[n]^{2D}} R_I \cdot \sum_{t=1}^m (\partial_{I} A_t^{3D}) \circ M(y) = \sum_{t=1}^m B_t(y)^D Q_t(y) \mcom
    \end{equation*}
    and moreover, $\|R\|_2 \leq O(n^{-D}) \sum_{t=1}^m \|Q_t\|_F$.
\end{lemma}


Thus, in order to recover an (approximate) basis for $\calV_D$, it is enough to obtain a basis for (a noisy estimate of) the subspace $\calU_D$. Our algorithm for this will simply populate natural elements of $\calU_D$ and then take the top few singular vectors to get a basis for $\calU_D$ (and thus also for $\calV_D$). 

\paragraph{Dimension Counting} In order to analyze this algorithm, let us first calculate the dimension of $\calV_D$. Observe that the polynomials $B_t(y)^D y_T$ for $t\in[m]$ and $T \in [\ell]^{2D}$ are \emph{not} linearly independent: as $|T| = \deg(B_t^D) = 2D$, we can have $B_t(y)^D B_s(y)^D = B_s(y)^D B_t(y)^D$ for $s\neq t$.
We will show, in a sense, that these are the \emph{only} linear dependencies in the set of polynomials $B_t(y)^D y_T$ and thus $\dim(\calV_D) = m\ell_{2D} - \binom{m}{2}$. In fact, for our error-tolerance, we will show that the matrix that populates the coefficient vectors of polynomials of the form $B_t(y)^D y_T$ has $\dim(\calV_D)$ polynomially large singular values:


\begin{lemma}[Singular value lower bound for $V$]
\label{lem:V-singular-value-hd}
    Fix $D\in \N$. Let $m, \ell \in \N$ such that $m \leq (\frac{\ell}{\polylog(\ell)})^{2D}$ if $D\leq 2$, and $m\leq (\frac{\ell}{\polylog(\ell)})^{D}$ if $D>2$.
    Let $B_1, \dots, B_m$ be degree-2 homogeneous polynomials in $\ell$ variables with i.i.d.\ standard Gaussian coefficients.
    Then, with probability $1 - \ell^{-\Omega(D)}$, the set of homogeneous degree-$2D$ polynomials $\{p_t\}_{t\in[m]}$ satisfying $\sum_{t=1}^m B_t(y)^D p_t(y) = 0$ forms a subspace of dimension $\binom{m}{2}$ spanned by the following,
    \begin{equation*}
        \calN \coloneqq \Set{(p_1,\dots,p_m) \mid p_{t_1} = B_{t_2}^D,\ p_{t_2}=-B_{t_1}^D, \text{ and } p_{s} = 0 \textnormal{ for } s \neq t_1,t_2}_{t_1 < t_2 \leq m} \mper
    \end{equation*}
    Furthermore, let $V\in \R^{\ell_{4D} \times m\ell_{2D}}$ such that where each column is the coefficient vector of the degree-$4D$ polynomial $B_t(y)^D y_T$ for $t\in[m]$ and $T \in [\ell]^{2D}$.
    Then,
    \begin{equation*}
        \rank(V) = m\ell_{2D} - \binom{m}{2}, \quad \sigma_{m\ell_{2D} - \binom{m}{2}}(V) \geq \Omega(\ell^D) \mper
    \end{equation*}
\end{lemma}

From Lemma~\ref{lem:U-equals-V-hd}, we know that $\dim(\calU_D) = m\ell_{2D} - \binom{m}{2}$. Combining Lemma~\ref{lem:U-equals-V-hd} and \ref{lem:V-singular-value-hd}, we will prove that the coefficient vectors of $2D$-th order partial derivatives form a sufficiently incoherent spanning set for $\calU_D$ by establishing the following lower bound on the $\dim(\calU_D)$-th singular value of the matrix that populates them as columns.

\begin{lemma}[Singular value lower bound of $U$]
\label{lem:singular-value-of-U-hd}
    Under the same assumptions as Lemma~\ref{lem:V-singular-value-hd}, let $U$ be the $\ell_{4D} \times n_{2D}$ matrix where each column is the coefficient vector of the degree-$4D$ polynomial $\partial_I P \circ M(y)$ for multiset $I\in[n]^{2D}$.
    Then, with probability $1- n^{-\Omega(D)}$,
    \begin{equation*}
        \sigma_{\rank(U)}(U) = \Omega((n\ell)^D)
    \end{equation*}
    where $\rank(U) = m\ell_{2D} - \binom{m}{2}$ is our candidate rank.
\end{lemma}

Finally, we can upgrade the analyses above to handle non-zero error matrices $E$. Our bounds above can be used to infer the following distance bounds between our estimate and the truth. 
\begin{lemma}
Under the same hypothesis as Lemma~\ref{lem:singular-value-of-U-hd}, with probability at least $1-n^{-\Omega(D)}$ on the choice of the $A_t$s and the internal randomness of the algorithm, 
    $\|\wt{U} - U\|_F \leq O(D^{D/2})\|E\|_F.$ 
\end{lemma}

\begin{proof}
Recall that $\wt{U}$ is a matrix in $\R^{\ell_{4D} \times n_{2D}}$ whose columns are the vectors $(\partial_I \wh{P}) \circ M$. For a multiset $I$ and polynomial $Q(x)$, recall that $Q[I]$ denotes the coefficient of $x_I$ in $Q(x)$.
For all multisets $J \in [\ell]^{4D}$ and $I \in [n]^{2D}$, we have that the $(J, I)$ entry of $U$ and $\wt{U}$ is $((\partial_{I} P) \circ M)[J]$ and $((\partial_{I}\wh{P}) \circ M)[J]$ respectively. Subtracting the two we get that the $(J,I)$-entry of $U - \wt{U}$ equals:
$((\partial_{I}(P - \wh{P})) \circ M)[J] = ((\partial_{I}E) \circ M)[J]$.

Due to the structure of $M$ (columns are $e_i$'s with no repeated columns) one can check that:
\[\sum_{I \in [n]^{2D},\ J \in [\ell]^{4D}} ((\partial_I E) \circ M) [J]^2 \leq \sum_{I \in [n]^{2D},\ J \in [n]^{4D}} (\partial_I E)[J]^2.\]
Each summand in the RHS above is at most $O(D!)E[I \cup J]$, and since $I \cup J$ ranges over multisets $R \in [n]^{6D}$ of size $6D$ we get:
\begin{equation*}
    \|U - \wt{U}\|_F \leq O(D^{D/2})\|E\|_F \mper
    \qedhere
\end{equation*}
\end{proof}

\paragraph{Completing the proof of Lemma~\ref{lem:analysis-restricted-partials-hd}}
\begin{lemma}[Lemma~\ref{lem:analysis-restricted-partials-hd} restated: Correctness of Algorithm~\ref{algo:partial-der-hd}]\label{lem:partial-der-hd}
Fix $D\in \N$. Let $m, \ell, n \in \N$ be parameters such that $m \leq (\frac{\ell}{\polylog(\ell)})^{2D}$ if $D\leq 2$, and $m\leq (\frac{\ell}{\polylog(\ell)})^{D}$ if $D>2$, and that $m\ell^{2D} \leq (\frac{n}{\polylog(n)})^{2D}$.
    Given $\wh{P} = \sum_{t \in [m]} A_t(x)^{3D} + E(x)$ where each $A_t$ is a degree-2 homogeneous polynomial with i.i.d.\ $\calN(0,1)$ entries, and a restriction matrix $M \in \R^{n \times \ell}$, we have that with probability $1- n^{-\Omega(D)}$ over the choice of $A_t$'s, Algorithm~\ref{algo:partial-der-hd} outputs a subspace $\wt{\calV}_D$ of $\R^{\ell_{4D}}$ that satisfies:
    \[ \Norm{\wt{\calV}_D - \calV_D}_F \leq O\left(\frac{\|E\|_F}{(n\ell)^D}\right) \mcom \]
    with $B_t(y) = A_t \circ M(y)$.
\end{lemma}

\begin{proof}
By our analysis above (when $E(x) = 0$) we know that the column space of $U$ defined as $\calU_D$ is equal to $\calV_D$. We have that $\dim(\calV_D) = \dim(\calU_D) = m\ell_{2D} - \binom{m}{2}$ and moreover  $\sigma_{m\ell_{2D} - \binom{m}{2}}(U) \geq \Omega((n\ell)^D)$. Recall that $\wt{\calV}_D$ is the subspace spanned by the top $m\ell_{2D} - \binom{m}{2}$ dimensional left singular vectors of $\wt{U}$. Applying Wedin's theorem to matrix $U$ (with $\calU_1$ being the top $m\ell_{2D} - \binom{m}{2}$ dimensional left singular vector space of $U$ which is equal to $\calV_D$ and the error matrix $U - \wt{U}$) we get that,
\begin{equation*}
    \Norm{\wt{\calV}_D - \calV_D }_F \leq O\left(\frac{\|U - \wt{U}\|_F}{\sigma_{m \ell_{2D} - \binom{m}{2}}(U)}\right) \leq O\left(\frac{D^D \|E\|_F}{(n\ell)^D}\right) \mper
\end{equation*}
This completes the proof.
\end{proof}

\paragraph{Structure of the subsequent sections}
In Section~\ref{sec:U-equals-V-hd} we prove Lemma~\ref{lem:U-equals-V-hd}.
The proof relies on an application of a lower bound on the singular values of a certain matrix (Lemma~\ref{lem:singular-value-of-L}) shown in  Section~\ref{sec:sval-U-equals-V}.
Next, we prove Lemma~\ref{lem:V-singular-value-hd} in Section~\ref{sec:analysis-of-V-hd} that relies on some  singular value lower bounds deferred to Sections~\ref{sec:rank-of-N} and \ref{sec:singular-value-of-VV-NN}.  
Finally, in Section~\ref{sec:analysis-of-U-hd} we prove Lemma~\ref{lem:singular-value-of-U-hd} by combining Lemmas~\ref{lem:U-equals-V-hd} and \ref{lem:V-singular-value-hd}.

\subsubsection{Proof of Lemma~\ref{lem:U-equals-V-hd}: \texorpdfstring{$\calU_D = \calV_D$}{U=V}}
\label{sec:U-equals-V-hd}

First observe that $\partial_I A_t^{3D}(x)$ is a sum of products of scalars of the form $A_t[i_1,i_2]$ and linear polynomials of the form $\iprod{A_t[j], x}$.
More specifically, the terms in $\partial_I A_t^{3D}(x)$ can be categorized by the number of linear terms $\gamma_1$ and number of scalar terms $\gamma_2$ where $\gamma_1 + 2\gamma_2 = |I|$.
To formally write out each term in $\partial_I A_t^{3D}(x)$, we need the following definition.

\begin{definition}[Bucket profile]
\label{def:bucket-profile}
    Given integers $\gamma_1, \gamma_2 \geq 0$ such that $\gamma_1 + 2 \gamma_2 = 2D$, we call $\gamma = (\gamma_1, \gamma_2)$ a \emph{bucket profile} of the $2D$ partial derivative, and we write $\gamma^*$ to denote the special bucket $(2D, 0)$.
    Let $\Gamma_{2D}$ be the set of bucket profiles.
    
    Moreover, we define the \emph{bucket partition} of $\{1,2,\dots,2D\}$ as follows: $\kappa_1(\gamma) \coloneqq (1,2,\dots,\gamma_1)$ and $\kappa_2(\gamma) \coloneqq ((\gamma_1+1, \gamma_1+2), \dots, (2D-1, 2D))$, a set of $\gamma_2/2$ pairs.
\end{definition}

\begin{remark}[Interpretation of $\gamma$ and the terms in partial derivatives of $A_t^{3D}$]
    To better understand Definition~\ref{def:bucket-profile}, we make the following remarks,
    \begin{itemize}
        \item 
        Imagine $A_t^{3D}$ as being $3D$ buckets, and taking $2D$ partial derivatives means dropping $2D$ balls in these buckets such that each bucket contains at most $2$ balls.
        Then, $\gamma_1$ (resp.\ $\gamma_2$) denotes the number of buckets with 1 (resp.\ 2) balls, hence $\gamma_1 + 2\gamma_2 = 2D$.

        \item
        Let $I = (i_1,\dots,i_{2D}) \in [n]^{2D}$ be an ordered tuple, and consider $\partial_I A_t^{3D}$ and a bucket profile $\gamma \in \Gamma_{2D}$.
        Note that there are $3D - (\gamma_1+\gamma_2) = D+\gamma_2$ empty buckets.
        Thus, $\gamma$ represents the terms in $\partial_I A_t^{3D}$ that are products of $A_t^{D+\gamma_2}$ and $\gamma_1$ linear polynomials, which can be represented as
        $\Iprod{A_t[i_{\pi(k)}], x}$ for $k\in \kappa_1(\gam)$ where $|\kappa_1(\gam)| = \gamma_1$, for some permutation $\pi\in \bbS_{2D}$.
        Here note that $A_t[i_{\pi(k)}]$ is a vector of dimension $n$.
    \end{itemize}
\end{remark}



\paragraph{Reduce $\calU_D = \calV_D$ to proving feasibility of a linear system}

With Definition~\ref{def:bucket-profile}, we can formally write out the partial derivatives in a form that is convenient for our analysis: for a multiset $I = \{i_1,\dots, i_{2D}\} \in[n]^{2D}$ of size $2D$,
\begin{equation}
\label{eq:partial-explicit-hd}
    \partial_I A_t^{3D}(x) = \sum_{\gamma \in \Gamma_{2D}} c_{\gamma} A_t(x)^{D + \gamma_2} \sum_{\pi \in \bbS_{2D}} \prod_{k\in \kappa_1(\gamma)} \Iprod{A_t[i_{\pi(k)}], x} \prod_{(k_1,k_2)\in \kappa_2(\gamma)} A_t[i_{\pi(k_1)},i_{\pi(k_2)}]
\end{equation}
where $c_{\gamma} > 0$ is a scalar depending on $\gamma$ and bounded by $O(D)^{2D}$.
Note that one can view the summation over $\pi \in \bbS_{2D}$ as a way of symmetrizing over $I$, as the partial derivative should not depend on the ordering of $I$.
Thus, we have
\begin{equation*}
    \partial_I P(x) = \sum_{t=1}^m A_t(x)^D \sum_{\gamma\in \Gamma_{2D}} A_t(x)^{\gamma_2} p_{t,I,\gamma}(x)
\end{equation*}
where $p_{t,I,\gamma}(x)$ is a homogeneous polynomial of degree $\gamma_1$ consisting of products of $\gamma_1$ linear polynomials of the form $\iprod{A_t[i], x}$.

We next set $x = My$ where $M \in \R^{n\times \ell}$ is a given restriction matrix and $y$ is an $\ell$-dimensional variable.
Then, for $|I| = 2D$,
\begin{equation}
\label{eq:projected-partial-hd}
    \partial_I P \circ M(y) = \sum_{t=1}^m B_t(y)^D \sum_{\gamma\in \Gamma_{2D}} B_t(y)^{\gamma_2} q_{t,I,\gamma}(y)
\end{equation}
where from \pref{eq:partial-explicit-hd} we see that $q_{t,I,\gamma}$ is a degree-$\gamma_1$ polynomial:
\begin{equation}
\label{eq:projected-partial-explicit-hd}
    q_{t,I,\gamma}(y) = c_{\gamma} \cdot \sum_{\pi \in \bbS_{2D}} \prod_{k\in \kappa_1(\gamma)} \Iprod{M^\top A_t[i_{\pi(k)}], y} \prod_{(k_1,k_2)\in \kappa_2(\gamma)} A_t[i_{\pi(k_1)},i_{\pi(k_2)}] \mper
\end{equation}

From \pref{eq:projected-partial-hd} it is clear that for any $I$, the projected partial derivative $\partial_I P \circ M(y)$ lies in $\calV_D$, which means $\calU_D \subseteq \calV_D$.
To prove Lemma~\ref{lem:U-equals-V-hd}, we first write out $\sum_{I\in[n]^{2D}} R_I \cdot (\partial_I A \circ M(y))$ using \pref{eq:projected-partial-hd}:
\begin{equation*}
    \sum_{I\in[n]^{2D}} R_I \cdot (\partial_I A \circ M(y))
    = \sum_{t=1}^m B_t(y)^D \sum_{\gamma \in \Gamma_{2D}}^D B_t(y)^{\gamma_2} \sum_{I\in[n]^{2D}} R_I \cdot q_{t,I,\gamma}(y)
\end{equation*}
where $\gamma = (\gamma_1, \gamma_2)$ and $q_{t,I,\gamma}(y)$ is a homogeneous polynomial of degree $\gamma_1$.

Our main idea is to focus on the bucket profile $\gamma^* = (2D, 0)$ where $q_{t,I, \gamma^*}$ is a product of $2D$ linear polynomials.
We show that the polynomials $\{q_{t,I,\gamma^*}\}_{I}$ already give us enough freedom to construct any degree-$2D$ polynomials.
Specifically, we show that the following linear system in variables $\{R_I\}_{I}$ is feasible for any $Q_1,\dots,Q_m$:

\begin{definition}[Linear system for proving $\calU_D = \calV_D$]
\label{def:linear-system-for-U-equals-V}
    Given arbitrary degree-$2D$ homogeneous polynomials $Q_1,\dots,Q_m$, we define the following linear system in variables $\{R_I\}_{I\in [n]^{2D}}$:
    \begin{equation} \label{eq:U-equals-V-constraints-hd}
    \begin{aligned}
        &\sum_{I\in[n]^{2D}} R_I \cdot q_{t,I,\gamma^*}(y) = Q_t(y), & & \forall t\in [m], \\
        &\sum_{I\in[n]^{2D}} R_I \cdot q_{t,I,\gamma}(y) = 0, & & \forall \gamma \neq \gamma^*,\ t\in[m] \mper
    \end{aligned}
\end{equation}
\end{definition}

Note that the equations in Definition~\ref{def:linear-system-for-U-equals-V} are written as polynomial equations.

\paragraph{Writing the linear system in matrix form}
For a bucket profile $\gamma$, let $J = \{j_1,j_2,\dots,j_{\gamma_1}\} \in [\ell]^{\gamma_1}$ be a multiset. From \pref{eq:projected-partial-explicit-hd}, the coefficient of $y_J$ in $q_{t,I,\gamma}$ is
\begin{equation*}
\begin{aligned}
    \wh{q_{t,I,\gamma}}(J) &= c_{\gamma}' \cdot \sum_{\pi \in \bbS_{2D}} \prod_{k\in \kappa_1(\gamma)} (A_t M)[i_{\pi(k)}, j_k] \prod_{(k_1,k_2)\in \kappa_2(\gamma)} A_t[i_{\pi(k_1)},i_{\pi(k_2)}] \\
    &= c_{\gamma}' \cdot \sum_{\pi \in \bbS_{2D}} \Paren{ \bigotimes_{k\in [\gamma_1]} (A_t M)[j_k] \otimes A_t^{\ot \gamma_2} }[\pi(I)] \\
    &= c_{\gamma}'' \cdot \Sym \Paren{ \bigotimes_{k\in [\gamma_1]} (A_t M)[j_k] \otimes A_t^{\ot \gamma_2} }[I] \mper
\end{aligned}
\end{equation*}
Here recall that we denote $T[I] = T[i_1,i_2,\dots,i_{2D}]$ for a $2D$-order tensor $T$.
Note also that $A_tM \in \R^{n\times \ell}$ consists of $\ell$ columns of $A_t$ (determined by the restriction $M$), thus $(A_tM)[j_k] \in \R^n$ is simply a column of $A_t$.
Then, the equations in \pref{eq:U-equals-V-constraints-hd} reduces to
\begin{equation} \label{eq:U-equals-V-linear-system-hd}
    \Iprod{R_I, \Sym\Paren{ \bigotimes_{k\in [\gamma_1]} (A_t M)[j_k] \otimes A_t^{\ot \gamma_2}} }
    =
    \begin{cases}
        Q_t[J], & \forall t\in[m], J\in[\ell]^{2D}, \text{ for } \gamma = (2D,0) \\
        0, & \forall t\in[m], J\in[\ell]^{\gamma_1}, \text{ for } \gamma \in \Gamma_{2D}, \gamma \neq (2D,0).
    \end{cases}
\end{equation}

The above is a linear system in $n_{2D}$ variables $\{R_I\}_I$.
Since $\gamma_1 + 2\gamma_2 = 2D$, $\gamma_1$ can range from $0, 2, \dots, 2D$, and each $J$ is a multiset of size $\gamma_1$.
Thus we have in total $m \sum_{k=0}^{D} \ell_{2k}$ linear constraints.
The following matrix defines the linear system:

\begin{definition}
\label{def:L-matrix-hd}
    For each $\gamma \in \Gamma_{2D}$, let $L_{\gamma}$ be the $m \ell_{\gamma_1} \times n^{2D}$ matrix where each row is indexed by $t\in[m]$ and multiset $J\in [\ell]^{\gamma_1}$.
    \begin{equation*}
        L_{\gamma}[ (t,J), \cdot] = \fl \Paren{ \Sym \Paren{\bigotimes_{k\in [\gamma_1]} (A_t M)[j_k] \otimes A_t^{\ot \gamma_2} }}
    \end{equation*}
    i.e., each row is the flattened vector of a symmetric $2D$-th order tensor of dimension $n$.
    Moreover, let $r_D \coloneqq m\sum_{k=0}^D \ell_{2k}$ and define $L$ to be the $r_D \times n^{2D}$ matrix formed by concatenating the rows of $L_{\gamma}$ for $\gamma \in \Gamma_{2D}$.
\end{definition}

It is clear that the linear system \pref{eq:U-equals-V-linear-system-hd} can be written as
\begin{equation} \label{eq:U-equals-V-linear-system-matrix-form}
\begin{aligned}
    L_{(2D,0)} \cdot \fl(R_I) &= 
    \begin{bmatrix}
        \fl(Q_1) \\
        \vdots \\
        \fl(Q_m)
    \end{bmatrix} \mcom \\
    \begin{bmatrix}
        L_{(2D-2,1)} \\
        \vdots \\
        L_{(0, D)}
    \end{bmatrix}
    \cdot \fl(R_I) &= 0 \mper
\end{aligned}
\end{equation}

Next, we prove a singular value lower bound for $L$.

\begin{lemma}[Singular value lower bound for $L$]
\label{lem:singular-value-of-L}
    Fix $D\in \N$, let $m,\ell,n\in \N$ such that $m \ell^{2D} \leq (\frac{n}{\polylog(n)})^{2D}$.
    Let $L \in \R^{r_D \times n^{2D}}$ be the matrix defined in Definition~\ref{def:L-matrix-hd} where $r_D = m \sum_{k=0}^{D} \ell_{2k}$.
    Then, with probability $1 - n^{-\Omega(D)}$, $\sigma_{r_D}(L) \geq \Omega(n^D)$.
\end{lemma}

We defer the proof to Section~\ref{sec:sval-U-equals-V}.
Lemma~\ref{lem:singular-value-of-L} allows us to complete the proof of Lemma~\ref{lem:U-equals-V-hd}.

\begin{proof}[Proof of Lemma~\ref{lem:U-equals-V-hd}]
    From the analysis above, we see that to prove that there exists coefficients $R = \{R_I\}_{I \in [n]^{2D}}$ such that
    \begin{equation*}
        \sum_{I\in[n]^{2D}} R_I \cdot (\partial_{I} A\circ M(y)) = \sum_{t=1}^m B_t(y)^D Q_t(y) \mcom
    \end{equation*}
    it suffices to prove that the linear constraints in \pref{eq:U-equals-V-constraints-hd} are satisfied.
    The constraints reduce to the linear system in \pref{eq:U-equals-V-linear-system-matrix-form} with the $r_D \times n^{2D}$ matrix $L$ defined in Definition~\ref{def:L-matrix-hd}.
    By Lemma~\ref{lem:singular-value-of-L}, $\sigma_{r_D}(L) \geq \Omega(n^D)$ implies that there is a solution $R$, and further
    \begin{equation*}
        \Norm{\fl(R)}_2^2 \leq O(n^{-2D}) \cdot \sum_{t=1}^m \Norm{\fl(Q_t)}_2^2 \mper
    \end{equation*}
    This completes the proof.
\end{proof}

\subsubsection{Proof of Lemma~\ref{lem:V-singular-value-hd}: Analysis for \texorpdfstring{$V$}{V}}
\label{sec:analysis-of-V-hd}

\paragraph{Overview of proof of Lemma~\ref{lem:V-singular-value-hd}.}
Consider the matrix $V^\top V \in \R^{m\ell_{2D} \times m\ell_{2D}}$.
We would like to show that the rank of $V^\top V$ is $m\ell_{2D} - \binom{m}{2}$, meaning that it has a $\binom{m}{2}$-dimensional null space.
Let $p \in \R^{m\ell_{2D}}$ be a vector in the null space, and we view $p = (p_1, p_2,\dots, p_m)$ where each $p_i \in \R^{\ell_{2D}}$ is the coefficient vector of a degree-$2D$ homogeneous polynomial.
Then, $Vp = 0$ is equivalent to
\begin{equation}
\label{eq:null-space-of-V-hd}
    \sum_{t=1}^m B_t(y)^D p_t(y) = 0 \mper
\end{equation}
Observe that the $\binom{m}{2}$ tuples $(p_1,\dots,p_m)$ in $\calN$ are indeed solutions to \pref{eq:null-space-of-V-hd} simply because
\begin{equation*}
    B_{t_1}(y)^D B_{t_2}(y)^D - B_{t_2}(y)^D B_{t_1}(y)^D \mper
\end{equation*}

Thus, our goal is to show that 1) these $\binom{m}{2}$ solutions in $\calN$ are linearly independent, and 2) the linear span of $\calN$ is exactly the set of solutions to \pref{eq:null-space-of-V-hd}.
We first define the following matrix $N \in \R^{\binom{m}{2} \times m\ell_{2D}}$ where each row is dimension $m\ell_{2D}$ representing a tuple $(p_1,\dots,p_m)\in \calN$:

\begin{definition}[Null space of $V$]
\label{def:N-matrix}
    We define $N \in \R^{\binom{m}{2} \times m\ell_{2D}}$ to be the matrix whose rows are indexed by $(t_1,t_2)$ for $t_1 < t_2 \in [m]$ and each row represents a collection of $m$ degree-$2D$ polynomials $(p_1,\dots,p_m)$ such that
    \begin{equation*}
        p_s(y) =
        \begin{cases}
            B_{t_2}(y)^D & \text{if } s = t_1, \\
            - B_{t_1}(y)^D & \text{if } s = t_2, \\
            0           & \text{otherwise}.
        \end{cases}
    \end{equation*}
\end{definition}

By definition, each row of $N$ is a solution to \pref{eq:null-space-of-V-hd}, thus $VN^\top = 0$.
We first show that $N$ is rank $\binom{m}{2}$, which implies that $\calN$ is linearly independent.

\begin{lemma}[Rank of $N$]
\label{lem:rank-of-N}
    Let $m, \ell, D \in \N$ such that $m \leq (\frac{\ell}{\polylog(\ell)})^{2D}$.
    Let $N \in \R^{\binom{m}{2} \times m\ell_{2D}}$ be the matrix defined in Definition~\ref{def:N-matrix}.
    Then, with probability $1- \ell^{-\Omega(D)}$,
    \begin{equation*}
        \sigma_{\binom{m}{2}}(N) \geq \Omega(\ell^D) \mper
    \end{equation*}
\end{lemma}

Next, we would like to show that $\spn(\calN)$ are the only solutions to \pref{eq:null-space-of-V-hd}.
To this end, we prove the following,
\begin{lemma}
\label{lem:VV-NN}
    Let $m, \ell, D \in \N$ such that $m \leq (\frac{\ell}{\polylog(\ell)})^{2D}$ if $D\leq 2$, and $m\leq (\frac{\ell}{\polylog(\ell)})^{D}$ if $D>2$, and let $N \in \R^{\binom{m}{2} \times m\ell_{2D}}$ be the matrix defined in Definition~\ref{def:N-matrix}.
    Then with probability $1-\ell^{-\Omega(D)}$,
    \begin{equation*}
        \lambda_{\min}\Paren{V^\top V + N^\top N}
        \geq \Omega(\ell^{2D}) \mper
    \end{equation*}
\end{lemma}

We defer the proofs of Lemmas~\ref{lem:rank-of-N} and \ref{lem:VV-NN} to Section~\ref{sec:rank-of-N} and \ref{sec:singular-value-of-VV-NN}.
Lemmas~\ref{lem:rank-of-N} and \ref{lem:VV-NN} immediately allow us to complete the proof of Lemma~\ref{lem:V-singular-value-hd}.

\begin{proof}[Proof of Lemma~\ref{lem:V-singular-value-hd}]
    By Definition~\ref{def:N-matrix}, we know $VN^\top = 0$ since each row of $N$ represents polynomials $(p_1,\dots,p_m)$ such that each $p_i$ is degree $2D$ and $\sum_{t=1}^m B_t(y)^D p_t(y) = 0$.
    This implies that the matrices $V^\top V$ and $N^\top N$ have orthogonal column span.
    By Lemma~\ref{lem:rank-of-N}, $\rank(N) = \binom{m}{2}$, and by Lemma~\ref{lem:VV-NN}, the row span of $N$ is exactly the null space of $V^\top V$.
    This implies that $\rank(V) = m\ell_{2D} - \binom{m}{2}$ and that $\spn(\calN)$ is exactly the set of solutions to $\sum_{t=1}^m B_t(y)^D p_t(y) = 0$.
    
    $V^\top V + N^\top N$ having smallest eigenvalue at least $\Omega(\ell^{2D})$ further shows that the smallest singular value of $V$, $\sigma_{m\ell_{2D} - \binom{m}{2}}(V)$, being lower bounded by $\Omega(\ell^D)$. This completes the proof.
\end{proof}

\subsubsection{Proof of Lemma~\ref{lem:singular-value-of-U-hd}: Analysis for \texorpdfstring{$U$}{U}}
\label{sec:analysis-of-U-hd}

For starters, we recall the definition for $\calU_D$ as the following,
\begin{equation*}
    \calU_D \coloneqq \spn\Set{\partial_{I} P \circ M(y)}_{I\in [n], |I|=2D},
\end{equation*}
where $P(x) = \sum_{t=1}^m A_t(x)^3$ and $M$ is the given restriction matrix such that $\partial_I P \circ M(y) = \partial_I P (My)$.

We defined $U$ to be the $\ell_{4D} \times n_{2D}$ matrix where each column is the coefficient vector of the degree-$4D$ polynomial $\partial_I P \circ M(y)$ for multiset $I\in[n]^{2D}$.
In this section we prove Lemma~\ref{lem:singular-value-of-U-hd} which states a singular value lower bound for $U$.

\begin{proof}[Proof of Lemma~\ref{lem:singular-value-of-U-hd}]
    Let $V$ be the $\ell_{4D} \times m\ell_{2D}$ matrix defined in Lemma~\ref{lem:V-singular-value-hd}, and let $r_V = \rank(V)$ which equals $m\ell_{2D} - \binom{m}{2}$ by Lemma~\ref{lem:V-singular-value-hd}.
    Lemma~\ref{lem:V-singular-value-hd} also states that with high probability, $\sigma_{r_V}(V) \geq \Omega(\ell^D)$.
    Moreover, Lemma~\ref{lem:U-equals-V-hd} shows that $\colspan(U) = \colspan(V)$ and hence we know that $\rank(U) = \rank(V)$.
    Consider the singular value decompositions of $U$ and $V$:
    \begin{equation*}
        U = \sum_{i=1}^{r_V} \sigma_i u_i \wt{u}_i^\top, \quad
        V = \sum_{i=1}^{r_V} \tau_i v_i \wt{v}_i^\top,
    \end{equation*}
    where the column spans coincide: $\spn\{u_i\}_{i\in[r_V]} = \spn\{v_i\}_{i\in[r_V]}$.
    
    Let $q\in \R^{m\ell_{2D}}$ be a vector in the row span of $V$, i.e.\ $\spn\{\wt{v}_i\}_{i\in[r_V]}$, such that $Vq = u_{r_V}$, which is the bottom (left) singular vector of $U$.
    Note that we can equivalently view $q \in \R^{m\ell_{2D}}$ as degree-$2D$ polynomials $Q_1(y),\dots,Q_m(y)$, and that $Vq$ represents the degree-$4D$ polynomial $\sum_{t=1}^m B_t(y)^D Q_t(y)$.
    
    Then, by Lemma~\ref{lem:U-equals-V-hd}, there exists a vector $p \in \R^{n_{2D}}$ (a flattened $2D$-th order tensor) such that $Up = Vq$ and $\|p\|_2 \leq O(n^{-D}) \|q\|_2$.
    Moreover, since $Up = u_{r_V}$, $p$ must be orthogonal to $\spn\{u_1,\dots,u_{r_V-1}\}$, implying that
    \begin{equation*}
        \|Up\|_2 \leq \sigma_{r_V} \|p\|_2 \leq \sigma_{r_V}(U) \cdot O(n^{-D}) \|q\|_2 \mper
    \end{equation*}
    We also have that 
    \begin{equation*}
        \|Vq\|_2 \geq \tau_{r_V} \|q\|_2 \geq \Omega(\ell^D) \cdot \|q\|_2 \mper
    \end{equation*}
    This proves that $\sigma_{r_V}(U) \geq \Omega((n\ell)^D)$, completing the proof.
\end{proof}

\subsection{Proof of Lemma~\ref{lem:span-finding-hd}: Span finding}
\label{sec:span-finding-hd}

\begin{mdframed}
  \begin{algorithm}[Span-Finding]
    \label{algo:span-finding-hd}\mbox{}
    \begin{description}
    \item[Input:] A basis for $\widetilde{\calV}_D$.
    \item[Output:] A basis for $\widetilde{\calW}_D$ such that: $\Norm{\widetilde{\calW}_D - \spn(B_t(y)^D \mid t \in [m])}_F \leq O(m\ell^{2D}\|\calV_D - \wt{\calV}_D\|_F)$.
     
    \item[Operation:]\mbox{}
    \begin{enumerate}
    	\item Choose a random degree-$2$ homogeneous polynomial $p(y)$ and compute the subspace $\calV_p = \spn(p(y)^D y_S \mid S \in [\ell]^{2D})$.
	    \item Let $\widetilde{\calW}_p$ be the top $m$-dimensional subspace of $\proj(\calV_p) + \proj(\widetilde{\calV_D})$, spanned by the orthonormal vectors $\wt{w_p^1},\ldots, \wt{w_p^m}$.
		\item Let $V_p \in \R^{\ell_{4D} \times \ell_{2D}}$ be the matrix whose columns are $p(y)^Dy_S$ for multisets $S \in [\ell]^{2D}$. For each $i \in [m]$: 
		compute $\wt{w_i} = \argmin_{w \in \R^{\ell_{2D}}} \Norm{V_p \cdot w - \wt{w_p^i}}_2$. 
		\item Output $\widetilde{\calW}_D = \spn(\wt{w_i})$.
      \end{enumerate}
    \end{description}
  \end{algorithm}
\end{mdframed}

\paragraph{Notation:} We will use the following notation throughout this section.
\begin{enumerate}
\item $\calV_D := \spn(B_t(y)^D \cdot y_S \mid t \in [m], S \in [\ell]^{2D})$. This subspace will be associated with the matrix $V \in \R^{\ell_{4D}\times \ell_{2D}}$ whose columns are the coefficient vectors of the polynomials $B_t(y)^D y_S$. 
\item Given a polynomial $p(y)$, define $\calV_p := \spn(p(y)^D y_S \mid S \in [\ell]^{2D})$. This subspace will be associated with the matrix $V_p \in \R^{\ell_{4D} \times \ell_{2D}}$ whose columns are the coefficient vectors of the polynomials $p(y)^Dy_S$.
\item $\calW_p := \spn(p(y)^D \cdot B_t(y)^D \mid t \in [m])$. This subspace will be associated with the matrix $W_p$ whose columns $w_p^{i}, i \in [m]$ form an orthonormal basis for $\calW_p$.
\item $\calW_D := \spn(B_t(y)^D \mid t \in [m])$. 
\item The tilde-versions of these quantities (e.g. $\wt{\calV}_D$) denote the noisy estimates given as input/estimated by the algorithm unless specified otherwise.
\end{enumerate}

The algorithm for span-finding in the case where $E(x) = 0$ is very straightforward. Given $\calV_D$, we first compute the intersection of $\calV_D$ with the subspace $\calV_p$ for a random degree $2$ homogeneous polynomial $p$. It is easy to see that the subspace $\calW_p$ lies inside the intersection. We show that in fact the intersection is \emph{equal} to $\calW_p$ when $B_t$'s are also random polynomials. Given the subspace $\calW_p$ we can now divide out by the polynomial $p(y)^D$ to get the subspace $\calW_D$ spanned by polynomials $B_t(y)^D$. To make sure that this algorithm is robust to error we need to do each of these steps carefully: take a robust intersection of subspaces and divide out by $p(y)^D$ when the polynomials might be close to a multiple of $p(y)^D$. Let us now show that the exact algorithm works: 

\begin{lemma}\label{lem:intersection-hd}
For degree-2 homogeneous polynomials $B_t(y), t\in [m]$ and $p(y)$, with coefficients chosen independently at random from $\calN(0,1)$, we have that with probability $1$ over the draw of $A_t$s, we have:
    \[\calV_D \cap \calV_p = \calW_p \mcom\]
with $\dim(\calW_p) = m$.
\end{lemma}

\begin{proof}
It is clear that $\calW_p \subseteq \calV_D \cap \calV_p$. Let us now prove the other inclusion. Every non-zero polynomial that lies in $\calV_D$ and in $\calV_p$ must satisfy the equation $p(y)^D q_0(y) = \sum_{t \in [m]} B_t(y)^D q_t(y)$ for some degree $2D$ homogeneous polynomials $q_0(y), \ldots, q_m(y)$ where $q_0(y)$ is non-zero. Letting $B_0(y) = p(y)$, Lemma~\ref{lem:V-singular-value-hd} shows that with high probability over the choice of $B_t$'s the solution space for the above linear system is the ${m+1 \choose 2}$-dimensional subspace $\calN = \calN_1 + \calN_2$, with $\calN_1 = \spn(q_0 = B_t^D, q_t = p_0^D, q_{t'} = 0, \forall t'  \in [m]\setminus \{t\} \mid t \in [m])$ and $\calN_2 = \spn(q_0 = 0, q_{t_1} = B_{t_2}^D, q_{t_2} = -B_{t_1}^D, q_t = 0, \forall t'  \in [m]\setminus \{t_1,t_2\} \mid t_1,t_2 \in [m])$. Combined with the Schwartz-Zippel lemma, this in fact shows that the same event holds with probability $1$ over the draw of the $A_t$s. Since the second subspace sets $q_0(y)$ to $0$, we get that the non-zero solutions to $q_0(y)$ must be in the subspace $\spn(B_t(y)^D \mid t \in [m])$ and therefore every polynomial in the intersection must lie in the subspace $\calN_1 = \spn(p(y)^DB_t(y)^D \mid t \in [m]) = \calW_p$. Also note that since $\calN$ has dimension ${m+1 \choose 2}$ and is spanned by ${m+1 \choose 2}$ vectors, $\calW_p = \calN_1$ must be $m$-dimensional, which completes the proof of the lemma.
\end{proof}

Given subspace $\calW_p$, it is clear that dividing any basis of $\calW_p$ by the polynomial $p(y)^D$ gives a basis for ${\calW_D = \spn(B_t(y)^D \mid t \in [m])}$. Step 3 of the algorithm can be equivalently thought of as solving for a degree $2D$ homogeneous polynomial $w(y)$ that minimizes: $\|p(y)^Dw(y) - w_p^i(y)\|_2$ ($w_p^i$ is the coefficient vector of a degree $4D$ homogeneous polynomial). In the case where $w_p^{i}$ is a multiple of $p(y)^D$, we find $w(y)$ such that $p(y)^D w(y) = w_p^i(y)$, i.e. we have successfully divided by $p(y)^D$. Thus our algorithm outputs $\calW_D$ in the case when $E(x) = 0$. Let us now analyse the error-resilience of the algorithm.

\paragraph{Error Resilience:}
We will now show that given a subspace $\wt{\calV}_D$ close to $\calV_D$, it is possible to take a ``robust intersection'' of $\wt{\calV}_D$ with $\calV_p$ to get a subspace $\wt{\calW_p}$ (Step 2 in Algorithm~\ref{algo:span-finding-hd}) that is close to $\calW_p$. We then show how to do a ``robust division'' to obtain subspace $\calW_D$. Before analysing the error-resilience of the algorithm though, let us describe a procedure for taking a robust intersection of subspaces.

\subsubsection{Robust intersection of subspaces}
Suppose that we are given two subspaces $\calV_1, \calV_2$ with the promise that they intersect in a $m$-dimensional subspace $\calW$, consider the problem of finding $\calW$. Furthermore we want a robust algorithm to do so, in the sense that given \textit{perturbed} subspaces $\wt{V_1},\wt{V_2}$ we want to find a subspace $\wt{W}$ that is close to $W$. We will do so by the following simple algorithm: Output the top $m$-dimensional eigenspace of the matrix $\wt{M} = \proj(\wt{\calV_1}) + \proj(\wt{\calV_2})$. We will now show the correctness of this algorithm:

\begin{lemma}\label{lem:robust-int}
Let $\calV_1, \calV_2$ be two $k_1,k_2$-dimensional subspaces of $\R^n$ (respectively) such that $\calV_1 \cap \calV_2 = \calW$ with $\dim(\calW) = m$. Let $M = \proj(\calV_1)+\proj(\calV_2)$ be such that $\lambda_{k_1+k_2 - m}(M) > \delta$. Let $\wt{\calV_1},\wt{\calV_2}$ be such that $\|\calV_i - \wt{\calV_i}\| \leq \gamma$. Then we have that the top $m$-dimensional eigenspace of $\wt{M} = \proj(\wt{\calV_1})+\proj(\wt{\calV_2})$ denoted by $\wt{\calW}$ is close to $\calW$: 
\[\|\calW - \wt{\calW}\| \leq O\left(\frac{\gamma}{\delta}\right).\]
\end{lemma}

\begin{proof}
First note that $\rk(M) = k_1+k_2 - m$ and every unit vector $w$ in the intersection of $V_1,V_2$ will satisfy $w^T M w = 2$. Furthermore, every vector $w \notin W$ will satisfy: $w^T M w < 2$. Hence the top $m$-dimensional eigenspace of $M$ will be \emph{equal} to $\calW$ and will correspond to eigenvalue $2$. 

We will now prove that $\lambda_{k_1+k_2 - m}(M) > \delta$ implies that $\lambda_{m+1}(M) < 2-\delta$. Consider the matrix $U = M - 2\proj(\calW) = \proj(\calV_1 \cap \calW^\perp) + \proj(\calV_2 \cap \calW^\perp)$. Since $(\calV_1 \cap \calW^\perp) \cap (\calV_2 \cap \calW^\perp) = \phi$, we have that $\rk(U) = k_1 + k_2 - 2m$ and $\lambda_{i}(U) = \lambda_{m+i}(M)$, hence it suffices to prove that, $\lambda_{1}(U) < 2-\delta$. This follows immediately from the following claim:

\begin{claim}\label{claim:subspace-int}
Let $\calU_1,\calU_2$ be two $r_1,r_2$-dimensional subspaces of $\R^n$ (resp.) such that $\calU_1 \cap \calU_2 = \phi$ and $U = \proj(\calU_1)+\proj(\calU_2)$. If $\lambda_{r_1+r_2}(U) > \delta$ then $\lambda_1(U) < 2 - \delta$.
\end{claim}

We will now apply Wedin's theorem to $M$. We have that $\|\wt{M} - M\| \leq 2\gamma$ and letting $\calU_1 = \calW$ be the top $m$-dimensional eigenspace of $M$, $\calU_2$ being the next $k-m$-dimensional eigenspace, we can check that, $\min_{i,j}|[\Sigma_1]_{i,i} - [\Sigma_2]_{j,j}| > \delta$ and $\min_{i} |[\Sigma_1]_{i,i}| = 2 > \delta$, from the conditions of the lemma. Hence applying Wedin's theorem we get that,
\[\|\proj(\calW) - \proj(\wt{\calW})\|_F^2 \leq \frac{2\gamma^2}{\delta^2} \mper \]
\end{proof}

Let us complete the above proof by proving the claim:
\begin{proof}[Proof of Claim~\ref{claim:subspace-int}]
We will prove the contrapositive by considering the case that $\lambda_{1}(U) > 2-\delta$ achieved via the unit vector $v$. 
Let $v$ be the top eigenvector for $U$ such that \[v^TUv = v^T\proj(\calU_1)v+v^T\proj(\calU_2)v >2-\delta \mper \]
Suppose $v^T\proj(\calU_1)v = 1-t >1-\delta$, we know there is a unit vector $a\in R^{r_1}$ s.t. $
\|v-\calU_1 a\|^2 \leq t
$;
Similarly, we would have a unit vector $b\in R^{r_2}$ s.t. $
\|v-\calU_2 b\|^2 \leq \delta-t  
$
By a triangle inequality, we have \[ 
\|\calU_1a-\calU_2 b\|^2 \leq \delta
\]
which upper bounds the smallest singular value of the matrix $[\calU_1, \calU_2]$ by $\delta$, i.e. $\lambda_{r_1+r_2}(\calU)\leq \delta$.
\end{proof}

\subsubsection{Error resilience of Algorithm~\ref{algo:span-finding-hd}}
\begin{lemma}\label{lem:sval-intersection-hd}
For degree-2 homogeneous polynomials $B_t(y), t\in [m]$ and $p(y)$, with coefficients chosen independently at random from $\calN(0,1)$ we have that with probability $1 - \ell^{-\Omega(D)}$:
\[\lambda_{(m+1)\ell_{2D} - {m+1 \choose 2}}(\proj(\calV_D)+\proj(\calV_p)) > \Omega\left(\frac{1}{m\ell^{2D}}\right).\]
\end{lemma}

\begin{proof}
Let $M = \proj(\calV_D)+\proj(\calV_p)$. Let $A_p D_p B_p^T$ and $ADB^T$ be the singular value decomposition of $V_p$ and $V$ respectively. We have that $[V_p ~~ V][V_p ~~ V]^T = V_pV_p^T + VV^T = A_pD_p^2A_p^T + AD^2A^T$ and $M = A_pA_p^T + AA^T$.

First note that the left singular vector subspace of $[V_p ~~ V]$ is the column space of $[V_p ~~ V]$ which is $\calV_p + \calV_D$. The latter is a subspace of dimension $\ell_{2D} + (m\ell_{2D} - {m \choose 2}) - m = (m+1)\ell_{2D} - {m+1 \choose 2}$, since by Lemma~\ref{lem:intersection-hd} we have that $\dim(\calV_D \cap \calV_p) = m$. Using Lemma~\ref{lem:V-singular-value-hd} we have that $\sigma_{(m+1)\ell_{2D} - {m+1 \choose 2}}([V_p ~~ V]) \geq \Omega(\ell^D)$ which implies that for every vector $v$ in $\calV_p+\calV_D$:
\begin{equation}\label{eq:span-find-hd}
v^T V_p V_p v^T + v^T V V v^T = v^T [V_p ~~ V][V_p ~~ V]^T v \geq \Omega(\ell^{2D}) \|v\|_2^2.
\end{equation}
We can upper bound the LHS by: 
\[\max(\sigma_{\max}(V_p)^2, \sigma_{\max}(V)^2) (v^T A_p A_p^T v + v^T AA^T v) \leq \max(\|V_p\|_F^2, \|V\|_F^2) \cdot v^T M v.\]

We have that with probability $1 - \ell^{-\Omega(D)}$,
\[\|V\|_F^2 = \sum_{\substack{t \in [m] \\ S \in [\ell]^{2D}}} \|B_t(y)^D y_S\|_2^2 = \sum_{t \in [m]} \ell_{2D}\|B_t(y)^D \|_2^2 \leq O(m\ell^{4D}),\]
where the last inequality follows from Claim~\ref{claim:norm-of-polynomial-powers}. Similarly we get that $\sigma_{\max}(V_p)^2 \leq O(\ell^{4D})$. Combining equation~\ref{eq:span-find-hd} with the above we get,
\[v^T M v \geq \Omega\left(\frac{\ell^{2D}}{m\ell^{4D}}\right)\|v\|_2^2,\] 
for all $v \in \calV_D + \calV_p$ which completes the eigenvalue lower bound for $M$.
\end{proof}

\begin{lemma}[Robust Intersection]\label{lem:robust-int-hd}
Given degree 2 homogeneous polynomials $B_t(y), t \in [m]$ and $p(y)$ with coefficients drawn independently at random from $\calN(0,1)$, with probability $1 - \ell^{-\Omega(D)}$ Step 2 of the algorithm outputs $\wt{\calW_p}$ such that:
\[\|\wt{\calW_p} - \calW_p\|_F \leq O(m\ell^{2D}\|\wt{\calV}_D - \calV_D\|_F).\]
\end{lemma}

\begin{proof}
Recall that the output of Step 2 of the algorithm outputs $\wt{\calW_p}$: the top $m$ dimensional subspace of $\wt{M} = \proj(\calV_p) + \proj(\wt{\calV}_D)$. Using the fact that $\calV_p \cap \calV_D = \calW_p$ such that $\dim(\calW_p) = m$ (Lemma~\ref{lem:intersection-hd}), we can apply Lemma~\ref{lem:robust-int} for robust subspace intersection, with subspaces $\wt{\calV}_D \approx \calV_D$ and $\calV_p$ to get that $\wt{\calW_p}$ is close to the true intersection $\calV_p \cap \calV_D = \calW_p$:
\[\|\calW_p - \wt{W_p}\|_F \leq O\left(\frac{\|\wt{\calV}_D - \calV_D\|_F}{\lambda_{(m+1)\ell_{2D} - {m+1 \choose 2}}(M)}\right) \leq O(m\ell^{2D}\|\wt{\calV}_D - \calV_D\|_F),\]
by Lemma~\ref{lem:sval-intersection-hd}.
\end{proof}

We will now show that Step 3 of Algorithm performs a robust division by the polynomial $p(y)^D$, that is, $\wt{\calW}_D$ is close to the subspace $\calW_D$.

\begin{lemma}[Robust Division]\label{lem:robust-div-hd}
Given degree 2 homogeneous polynomials $B_t(y), t \in [m]$ and $p(y)$ with coefficients drawn independently at random from $\calN(0,1)$, with probability $1 - \ell^{-\Omega(D)}$, Step 4 of the algorithm outputs $\wt{\calW}_D$ such that:
\[\|\calW_D - \wt{\calW}_D\|_F \leq O(\ell^{2D}\|\calW_p - \wt{\calW_p}\|_F).\]
\end{lemma}

\begin{proof}
Since $p$ is a random degree $2$ homogeneous polynomial we can apply Lemma~\ref{lem:V-singular-value-hd} (with $m = 1$ and $G_1 = p$) to get that $\sigma_{\ell_{2D}}(V_p) \geq \Omega(\ell^D)$. Recall that $\wt{W_p}$ is the matrix whose columns are $\wt{w_p^1},\ldots, \wt{w_p^m}$ which are orthonormal vectors that span $\wt{\calW_p}$. Since $V_p$ has rank $\ell_{2D}$ (full column rank), Step 3 of Algorithm~\ref{algo:span-finding-hd} computes the solution to the least-squares program: $\wt{w_i} = V_p^{\dagger}\wt{w_p^i}$. Let $\wt{W}$ be a matrix whose columns are $\wt{w_i}$, i.e. $\wt{W} = V_p^{\dagger}\wt{W_p}$ and analogously let $W = V_p^{\dagger}W_p$, where $W_p$ is the matrix with orthonormal columns $w_p^{i}$ that span $\calW_p$. 
We get that:
\begin{align*}
\|WW^T - \wt{W}\wt{W}^T\|_F &= 
\|V_p^{\dagger}(\wt{W_p}\wt{W_p^T} - WW^T) (V_p^{\dagger})^T\|_F \\
&\leq \spec{V_p^\dagger}^2 \cdot \|\wt{W_p}\wt{W_p^T} - W_p W_p^T\|_F \\
&= \frac{\|\calW_p - \wt{\calW_p}\|_F}{\sigma_{\ell_{2D}}(V_p)^2} 
= O\left(\frac{\|\calW_p - \wt{\calW_p}\|_F}{\ell^{2D}}\right).
\end{align*}

Now note that $\calW_D = \spn(w_i)$ is the column space of $WW^T$ and therefore also the top $m$ dimensional eigenspace, the same being true for $\wt{\calW}_D$. Since the matrices $WW^T$ and $\wt{W}\wt{W}^T$ are close in Frobenius norm we can apply Wedin's theorem (Theorem~\ref{thm:wedin}) to get that their top $m$ dimensional eigenspaces are close:
\begin{equation}\label{eq:W_D-closeness}
\|\calW_D - \wt{\calW}_D\|_F \leq O\left(\frac{\|WW^T - \wt{W}\wt{W}^T\|_F}{\lambda_{m}(WW^T)} \right) \leq O\left(\frac{\|\calW_p - \wt{\calW_p}\|_F}{\ell^{2D} \lambda_m(WW^T)}\right).
\end{equation}

Let us finish the proof by bounding $\lambda_{m}(WW^T) = \sigma_m(W)^2$. Since $W_p$ lies in the column space of $V_p$ (the polynomials $w_p^i(y)$ are multiples of $p(y)^D$) we get that, $W_p = V_p \cdot W$. We know that $\sigma_m(W_p) = 1$ ($W_p$ has orthonormal columns by construction) therefore we get that for any unit vector $v \in \R^m$:
\[1 \leq \|W_p v\|_2 = \|V_p W v\|_2 \leq \sigma_{\max}(V_p) \cdot \|Wv\|_2 \leq \|V_p\|_F \|Wv\|_2,\]
which implies that $\sigma_m(W) \geq 1/\|V_p\|_F$. Using Claim~\ref{claim:norm-of-polynomial-powers} we have that with probability $1-\ell^{-\Omega(D)}$:
\[\|V_p\|_F^2 = \sum_{S \in [\ell]^{2D}} \|p(y)^D y_S\|_2^2 = \ell_{2D} \|p(y)\|_2^2 \leq O(\ell^{4D}),\]
which implies that $\lambda_m(WW^T) = \sigma_m(W)^2 \geq \Omega(1/\ell^{4D})$.
Plugging this into the equation~\ref{eq:W_D-closeness} we get that:
\[\|\calW_D - \wt{\calW}_D\|_F \leq O(\ell^{2D}\|\calW_p - \wt{\calW_p}\|_F)\mper \qedhere \]
\end{proof}

Together with Lemma~\ref{lem:robust-int-hd} and \ref{lem:robust-div-hd}, we can complete the proof of Lemma~\ref{lem:span-finding-hd}.

\begin{lemma}[Restatement of Lemma~\ref{lem:span-finding-hd}]
Let $D,m,\ell\in \N$ such that $m \leq (\frac{\ell}{\polylog(\ell)})^{2D}$. Given degree 2 homogeneous polynomials $B_t, t\in[m]$ in $\ell$ variables with coefficients drawn i.i.d from $\calN(0,1)$, with probability $1- \ell^{-\Omega(D)}$  Algorithm~\ref{algo:span-finding-hd} outputs $\wt{\calW}_D$ that satisfies:
\[ \Norm{\wt{\calW}_D - \calW_D}_F \leq O\left(m\ell^{4D} \|\calV_D - \wt{\calV}_D\|_F\right).\]
\end{lemma}

\begin{proof}
This follows immediately by combining Lemmas~\ref{lem:robust-int-hd} and \ref{lem:robust-div-hd}. Lemma~\ref{lem:robust-int-hd} states that when the polynomials $B_t$'s and $p$ are random, with probability $1-\ell^{-\Omega(D)}$ Step 2 of the algorithm outputs $\wt{\calW_p}$ such that:
\[\|\wt{\calW_p} - \calW_p\|_F \leq O(m\ell^{2D}\|\wt{\calV}_D - \calV_D\|_F).\]
Lemma~\ref{lem:robust-div-hd} states that with probability $1-\ell^{-\Omega(D)}$ Step 4 of the algorithm outputs $\wt{\calW_D}$ such that:
\[\|\calW_D - \wt{\calW}_D\|_F \leq O(\ell^{2D}\|\calW_p - \wt{\calW_p}\|_F).\]
Combining both the inequalities the bound on $\|\calW_D - \wt{\calW}_D\|_F$ easily follows.
\end{proof}

\subsection{Proof of Lemma~\ref{lem:analysis-desym-hd}: Desymmetrization}\label{sec:desymm-hd}
\begin{mdframed}
  \begin{algorithm}[Desymmetrization]
    \label{algo:desymmetrization-hd}\mbox{}
    \begin{description}
    \item[Input:] Basis for subspace $\wt{\calW}_D$.
    \item[Output:] $\wt{T} \in \R^{\ell_{2D}^3}$ such that: \[\Norm{\wt{T} - \sum_{t \in [m]} (\Sym(B_t^{\ot D}))^{\otimes 3}}_F \leq \poly(m)\Paren{\ell^{O(D)}\|\wt{\calW}_D - \calW_D\|_F + \|E\|_F}.\]
    \item[Operation:]\mbox{}
    \begin{enumerate}
    \item Let $\wt{W}$ be the matrix whose columns $\{\wt{w_1},\ldots,\wt{w_m}\}$ form an orthonormal basis for $\wt{\calW}_D$. Compute the matrix $\wt{W}_{uniq}^{\otimes 3}$ whose columns are $\wt{w_i} \ot \wt{w_j} \ot \wt{w_k}$ indexed by multisets $[i,j,k] \subseteq [m]$.
    \item Solve the following least-squares minimization problem:
    \[\min_{Y \in \R^{m_3}}\Norm{\fl(\wh{P}\circ M) - \Sym_{6D} \cdot \wt{W}^{\otimes 3}_{uniq} \cdot Y}_F.\]
    \item Compute $Z \in \R^{m^3}$ indexed by $i,j,k \in [m]$ with $Z_{i,j,k} = Y_{[i,j,k]}/6$ if $|[i,j,k]| = 3$, $Z_{i,j,k} = Y_{[i,j,k]}/3$ if $|[i,j,k]| = 2$ and $Z_{i,j,k} = Y_{[i,j,k]}$ if  $|[i,j,k]| = 1$, where $[i,j,k] \subseteq [m]$ denotes a multiset of size $3$. 
    \item Output $\wt{T} = \wt{W}^{\otimes 3} \cdot Z$. \end{enumerate}
    \end{description}
  \end{algorithm}
\end{mdframed}

\paragraph{Notations:} We will use the following notations throughout this section.
\begin{enumerate}
\item $\calW_D := \spn( B_t(y)^D \mid t \in [m])$. Let $W$ denote the matrix with columns $\{w_1,\ldots,w_m\}$ that form an orthonormal basis for $\calW_D$.
\item Let $C_t$ denote the coefficient vector of $B_t(y)^D$, i.e. $C_t = \Sym_{2D}(B_t^{\otimes D})$,
\item For a matrix $W\in \R^{n\times m}$ with columns $w_1,\dots,w_m$, we use $W^{\ot k}$ to denote the $n^k \times m^k$ matrix whose columns are all possible tensor products of $k$ columns of $W$: $w_{i_1} \ot \cdots \ot w_{i_k}$ for $i_1,\ldots,i_k \in [m]$.
Furthermore, we denote $W^{\ot k}_{uniq}$ to be the $n^k \times m_k$ matrix whose columns are $w_{i_1} \otimes \cdots \otimes w_{i_k}$ for $i_1 \leq i_2 \leq \cdots \leq i_k \in [m]$.
\end{enumerate}

Let us first describe the algorithm in the case when $\delta_2 = 0$, i.e. given $\calW_D$ and $\Sym_{6D}(\sum_t C_t^{\ot 3}) = \Sym_{6D} (\sum_t \Sym_{2D}(B_t^{\otimes D})^{\otimes 3})$, we will obtain the unsymmetrized tensor $\sum_t \Sym_{2D}(B_t^{\otimes D})^{\otimes 3}$. From now on we will drop the subscript of $\Sym$ when it is clear from context.  therefore we will show how to recover $\sum_t C_t^{\ot 3}$.

A priori we do not have enough information to recover the tensor $\sum_t C_t^{\ot 3}$ from the tensor $\Sym(\sum_t C_t^{\ot 3})$, but we show that given the subspace spanned by the polynomials $C_t$'s we can ``desymmetrize'' to recover it. Since each vector $C_t$ belongs to $\calW$, $C_t$ can be written as $W c_t$  for some vector $c_t \in \bbR^m$. Therefore $C_t^{\otimes 3} = W^{\otimes 3} \cdot c_t^{\otimes 3}$, where $W^{\otimes 3} \in \R^{\ell_{2D}^3 \times m^3}$ is the matrix whose columns are $w_i \otimes w_j \otimes w_k$ ranging over $i,j,k \in [m]$. Summing over $t$ we get that the tensor $T = \sum C_t^{\otimes 3}$ equals $W^{\otimes 3} \cdot (\sum c_t^{\otimes 3})$. Let us write the vector $\sum c_t^{\otimes 3}$ as the vector of unknown variables $Z = [Z_{ijk}]_{i,j,k \in [m]}$. Let $\Sym_{6D} \in \R^{\ell_{6D} \times \ell_{2D}^3}$ denote the matrix that symmetrizes a vector in $\R^{\ell_{2D}^3}$ so that it is a valid coefficient vector of a degree $6D$ polynomial. We have the following linear system:
\[\Sym_{6D}\Bigparen{\sum_t C_t^{\otimes 3}} = \Sym_{6D} \cdot W^{\otimes 3} \cdot Z \mper\]
The matrix $\Sym_{6D} \cdot W^{\otimes 3}$ still does not have full column rank though (hence is not invertible), because
$\Sym_{6D} \cdot (w_i \otimes w_j \otimes w_k) = \Sym_{6D} \cdot (w_j \otimes w_i \otimes w_k) = \ldots $ for any multiset $[i,j,k] \subseteq [m]$. To fix this, consider the matrix $W^{\otimes 3}_{uniq} \in \R^{\ell_{2D}^3 \times m_3}$ with columns that are the ``unique'' columns of $W^{\otimes 3}$: $w_i \otimes w_j \otimes w_k$ for $i \leq j \leq k$. Let $Y \in \R^{m_3}$ be the corresponding vector of unknowns indexed naturally using multisets $S \subseteq [m]$ of size $3$. Giving a particular setting of a symmetric vector $Z \in \R^{m^3}$ ($Z_{i,j,k} = Z_{\pi(i),\pi(j),\pi(k)}, \forall \pi \in \bbS_3$), let $Y_{[i,j,k]} = 6Z_{i,j,k}$ if $|[i,j,k]| = 3$, $Y_{[i,j,k]} = 3Z_{i,j,k}$ if $|[i,j,k]| = 2$ and $Y_{[i,j,k]} = Z_{i,j,k}$ if  $|[i,j,k]| = 1$ (as in the algorithm), and define this invertible (over symmetric $Z$'s) linear transformation as $L \in \R^{m \times m_3}$:
\begin{equation} \label{eq:L}
    Z = LY \mper
\end{equation}
One can check that, 

\[\Sym_{6D}(\sum_t C_t^{\otimes 3}) = \Sym_{6D} \cdot W^{\otimes 3} \cdot Z = \Sym_{6D} \cdot W^{\otimes 3}_{uniq} \cdot Y.\]

We will show that the matrix $\Sym_{6D} \cdot W^{\otimes 3}_{uniq}$ has full column rank (Lemma~\ref{lem:desymm-mainsv-hd}), therefore,
\[Y = (\Sym_{6D} \cdot W^{\otimes 3}_{uniq})^\dagger(\Sym_{6D}(\sum_t C_t^{\otimes 3})).\]

Given $Y$ we can recover $Z$ by multiplying by matrix $D$, and finally multiplying by $W^{\ot 3}$ we get the desymmetrized tensor $\sum_t C_t^{\otimes 3}$:
\begin{equation}\label{eq:exact-desymm}
\sum_t C_t^{\otimes 3} = W^{\ot 3} L (\Sym_{6D} \cdot W^{\otimes 3}_{uniq})^\dagger(\Sym_{6D}(\sum_t C_t^{\otimes 3})).
\end{equation}

\begin{lemma}\label{lem:desymm-mainsv-hd}
Let $m,\ell, D \in \N$ such that $m \leq (\frac{\ell}{\polylog(\ell)})^{2D}$.
Then, with probability $1- \ell^{-\Omega(D)}$,
\[\sigma_{m_3}(\Sym_{6D} \cdot W^{\otimes 3}_{uniq}) > \Omega\left(\frac{1}{m^{1.5}}\right).\]
\end{lemma}

Let us first decompose the matrix $W^{\otimes 3}_{uniq}$ into a random matrix times a basis transformation matrix. Define $C$ as the matrix whose columns are the vectors $C_t$ and analogously define $C^{\otimes 3}$ as well as $C^{\otimes 3}_{uniq}$. Let $V$ be a basis transformation matrix between $C$ and $W$: $C \cdot V = W$. We have the following easy to prove lemma:

\begin{lemma}\label{lem:desymm-basis-hd}
\begin{enumerate}
\item  $C^{\otimes 3}_{uniq} \cdot V^{\otimes 3}_{uniq} = W^{\otimes 3}_{uniq}.$
\item With probability $1- \ell^{-\Omega(D)}$ we have that:
$\sigma_{m_3}(V^{\otimes 3}_{uniq}) \geq \frac{1}{m^{1.5}\ell^{3D}}.$
\end{enumerate}
\end{lemma}
\begin{proof}
The proof of (1) is straightforward so let us go to the proof of (2). For any unit vector $v \in \R^{m_3}$ we have that:
\[\|C_{uniq}^{\ot 3}V^{\ot 3}_{uniq}v\|_2 = \|W^{\otimes 3}_{uniq} v\|_2 \geq 1,\]
since $W^{\otimes 3}_{uniq}$ is an orthonormal matrix. We can upper bound the LHS by $\spec{C_{uniq}^{\ot 3}}\|V^{\ot 3}_{uniq}v\|_2$
rearranging which gives us that $\|V^{\ot 3}_{uniq}v\|_2 \geq \frac{1}{\spec{C_{uniq}^{\ot 3}}} \geq \frac{1}{\spec{C}^3} \geq \frac{1}{\|C\|_F^3}$ which implies that $\sigma_{m_3}(V^{\ot 3}_{uniq}) \geq \frac{1}{\|C\|_F^3}$.

We have that $\|C\|_F^2 = \sum_{t \in m}\|B_t(y)^D\|_2^2$ which is less than $O(m\ell^{2D})$ with probability $1- \ell^{-\Omega(D)}$ (Claim~\ref{claim:norm-of-polynomial-powers}). So we get that $\|C\|_F^3 \geq 1/(m^{1.5}\ell^{3D})$ completing the proof of the lemma.
\end{proof}

The following lemma is crucial to the singular value lower bound. We defer the proof to Section~\ref{sec:sval-desymm}.

\begin{lemma}
\label{lem:desymm-singular-value-hd}
    Let $m,\ell, D \in \N$ such that $m \leq (\frac{\ell}{\polylog(\ell)})^{2D}$.
    Then, with probability $1- \ell^{-\Omega(D)}$,
    \begin{equation*}
        \sigma_{m_3} \Paren{\Sym_{6D} \cdot C^{\otimes 3}_{uniq}} \geq \Omega(\ell^{3D}) \mper
    \end{equation*}
\end{lemma}

The above two lemmas immediately imply that $(\Sym_{6D} W^{\ot 3}_{uniq})$ is invertible: 
\begin{proof}[Proof of Lemma~\ref{lem:desymm-mainsv-hd}]
We have that $\Sym_{6D} \cdot W^{\ot 3}_{uniq} = \Sym_{6D} \cdot C^{\ot 3}_{uniq} \cdot V^{\ot 3}_{uniq}$. Since $V^{\ot 3}_{uniq}$ is a square matrix, we can multiply the two singular value lower bounds in Lemmas~\ref{lem:desymm-basis-hd},~\ref{lem:desymm-singular-value-hd} to get:
\begin{equation*}
    \sigma_{m_3}(\Sym_{6D} \cdot W^{\ot 3}_{uniq}) \geq \sigma_{m_3}(\Sym_{6D} \cdot C^{\ot 3}_{uniq}) \cdot \sigma_{m_3}(V^{\ot 3}_{uniq})
    \geq \Omega\left(\frac{1}{m^{1.5}}\right) \mper
    \qedhere
\end{equation*}
\end{proof}

\paragraph{Error Resilience:} We have already proved (equation~\ref{eq:exact-desymm}) that when given the subspace $\calW_D$ we can recover the desymmetrized tensor $\sum_t C_t^{\ot 3}$. We will now show that in the case when the input subspace $\wt{\calW}_D$ is sufficiently close to $\calW_D$ and $\|E\|_F$ is bounded, Algorithm~\ref{algo:desymmetrization-hd} outputs a solution close to $\sum C_t^{\ot 3}$. Roughly we show that instead of solving the linear system in equation~\ref{eq:exact-desymm}, the algorithm solves a least-squares minimization program and since $\sigma_{m_3}(\Sym_{6D} W^{\ot 3}_{uniq})$ is at least $1/\poly(m)$ the algorithm is error resilient.

Before proving the correctness of the algorithm let us prove the following lemma that relates the bases of two subspaces that are close to each other:

\begin{lemma}\label{lem:bases}
Given $d$-dimensional subspaces $\calW$ and $\wt{\calW}$ with $\|\calW - \wt{\calW}\|_F \leq 1$, there exists orthonormal matrices $W,\wt{W} \in \R^{n \times d}$ with $WW^\top = \proj(\calW)$ and $\wt{W}\wt{W}^\top = \wt{\calW}$ such that:
\[\wt{W} = W + \Gamma,\]
with $\|\Gamma\|_F \leq O(d\|\wt{\calW} - \calW\|_F)$.
\end{lemma}

\begin{proof}
Let $E = \proj(\wt{\calW}) - \proj(\calW)$. Let $\wt{W}$ be an orthonormal basis for $\wt{\calW}$: $\wt{W}\wt{W}^\top = \proj(\wt{\calW})$ and similarly let $W'$ be an orthonormal basis for $\calW$: $W'W'^\top = \proj(\calW)$. Then, we have that ${\wt{W}\wt{W}^\top = W'W'^\top + E.}$
Multiplying both sides on the right by $\wt{W}$ we get:
\begin{equation} \label{eq:bases-close}
\wt{W} = W'W'^\top\wt{W} + E\wt{W} = W'\wt{V} + E\wt{W},
\end{equation}
for $\wt{V} = W'^\top\wt{W}$. We will now show that $\wt{V}$ is close to an orthonormal matrix $V \in \R^{d\times d}$. We have:
\[\wt{V}\wt{V}^\top = W'^\top(\wt{W}\wt{W}^\top) W' = W'^\top(W'W'^\top + E) W' = \Id_{d} + W'^\top E W',\]
where, $\spec{W'^\top E W'} = \spec{E}$. By Weyl's theorem (Theorem~\ref{thm:weyl}) we get that the eigenvalues of $\wt{V}\wt{V}^\top$ are in the range $[1-\spec{E}, 1+\spec{E}]$. 

Let $\wt{V}$ have the singular value decomposition $ADB^\top$. Since the diagonal entries of $D^2$ lie in $[1-\spec{E}, 1+\spec{E}]$, we get that the diagonal entries of $D$ are in the range, $[1-\spec{E}, 1+\spec{E}]$ since $\spec{E} \leq 1$, therefore $D = \Id + E_1$, where $E_1$ is a diagonal matrix with $\spec{E_1} \leq \spec{E}$. Therefore we get,
\[\wt{V} = A(\Id + \Gamma)B^\top = AB^\top + A E_1 B^\top = V+E_2,\]
for $V$ equal to the orthonormal matrix $AB^\top$ and: 
\[\|E_2\|_F \leq \spec{A}\|E_1\|_F\spec{B} \leq d\spec{E} 
\leq d\|E\|_F.\] 
Plugging this back into equation~\ref{eq:bases-close} we get:
\[\wt{W} = W'V + W'E_2 + E\wt{W} = W + \Gamma,\]
where $W = W'V$ is an orthonormal basis for $\calW$ ($WW^\top = \proj(\calW)$) and $\|\Gamma\|_F \leq d\|E\|_F+\|E\|_F = O(d\|E\|_F)$.
\end{proof}

\begin{lemma}[Correctness of Algorithm~\ref{algo:desymmetrization-hd}, restatement of Lemma~\ref{lem:analysis-desym-hd}]
    Let $D,m,\ell\in \N$ such that $m \leq (\frac{\ell}{\polylog(\ell)})^{2D}$.
    For each $t\in[m]$, let $B_t$ be a degree-2 homogeneous polynomial in $\ell$ variables with i.i.d.\ $\calN(0,1)$ entries.
    Suppose $\wt{\calW}_D$ is a subspace of $\R^{\ell_{2D}}$ such that $\Norm{\wt{\calW}_D -\calW_D}_F \leq 1/(m^{3.5}\ell^{O(D)})$, then with probability $1- \ell^{-\Omega(D)}$ over the choice of $B_t$'s,
    Algorithm~\ref{algo:desymmetrization-hd} outputs a tensor $\wt{T}$ such that:
    \[ \Norm{\wt{T} - \sum_{t \in [m]} \Paren{\Sym(B_t^{\ot D})}^{\ot 3} }_F \leq 
    \poly(m)(\ell^{O(D)}\|\calW_D - \wt{\calW}_D\|_F + \|E\|_F)\mper\]
\end{lemma}

\begin{proof}
Let $W,\wt{W} \in \R^{\ell_{2D} \times m}$ be a basis for the subspaces $\calW_D,\wt{\calW}_D$ which we will henceforth call $\calW,\wt{\calW}$, given by Lemma~\ref{lem:bases}: $\wt{W} = W+\Gamma'$, with $\spec{\Gamma'} \leq O(m\|\calW - \wt{\calW}\|_F)$. We have that $H = W_{uniq}^{\ot 3}$ and $\wt{H} = \wt{W}_{uniq}^{\ot 3}$ are an orthonormal basis for $\calH = \calW^{\ot 3}_{uniq}$ and $\wt{\calH} = \wt{\calW}^{\ot 3}_{uniq}$ respectively, with $\wt{H}^{\ot 3} - H^{\ot 3}$ defined as $\Gamma$.

Let $B, \wt{B}$ be the coefficient vectors of the polynomials $P\circ M(y)$ and $\wh{P}\circ M(y)$ respectively. Recall the linear transformation $L$ from equation~\ref{eq:L}. Step 4 of the algorithm outputs $\wt{T}$ such that,
\[\wt{T} = \wt{W}^{\ot 3}L(\Sym_{6D} \cdot \wt{H})^\dagger\wt{B}\footnote{Note that this quantity does not depend on the particular choice of basis for $\wt{\calW}$.},\]
and analogously in the $E(x) = 0$ case we get:
\[T = W^{\ot 3}L(\Sym_{6D} \cdot H)^\dagger B = \sum_{t \in [m]} (\Sym_{2D}(B_t^{\ot D}))^{\ot 3},\]
where the last equality follows since $\Sym_{6D} H$ has full column rank (Lemma~\ref{lem:desymm-mainsv-hd}) with high probability and the analysis for equation~\ref{eq:exact-desymm}.

Therefore it suffices to bound $\|T - \wt{T}\|_2$ to prove the lemma. Using $\wt{W}^{\ot 3} = W^{\ot 3}+\Gamma_1$, we have:
\begin{align}
\|T - \wt{T}\|_2 &= \|W^{\ot 3}L(\Sym_{6D} H)^\dagger B - \wt{W}^{\ot 3}L(\Sym_{6D} \wt{H})^\dagger\wt{B}\|_2 \nonumber \\ 
&\leq \|W^{\ot 3}L[(\Sym_{6D} H)^\dagger B - (\Sym_{6D} \wt{H})^\dagger\wt{B}]\|_2 + \|\Gamma_1 (\Sym_{6D} \wt{H})^\dagger\wt{B}\|_2 \nonumber \\ 
&\leq \spec{L}\|(\Sym_{6D} H)^\dagger B - (\Sym_{6D} \wt{H})^\dagger\wt{B}\|_2 + \|\Gamma_1(\Sym_{6D} \wt{H})^\dagger\wt{B}\|_2 \mper \label{eq:desymm1}
\end{align}

Let $E_1 = \Sym_{6D} H - \Sym_{6D} \wt{H}$. We can bound the first term in equation~\ref{eq:desymm1} by using robustness of least-squares minimization for the program $\min_{y \in \R^{m_3}} \|\Sym_{6D} H y - B\|_2$ (with noisy estimates $\wt{H}$ and $\wt{B}$), since $\Sym_{6D} H$ is a matrix with full column rank (Lemma~\ref{lem:desymm-mainsv-hd}). Applying Lemma~\ref{lem:least-sq-robust} we get:
\begin{align*}
\Norm{(\Sym_{6D} H)^\dagger B - (\Sym_{6D} \wt{H})^\dagger\wt{B}}_2 
&\leq \frac{\sqrt{2}\spec{E_1}\|B\|_2 + \sigma_{m_3}(\Sym_{6D} H)\|B - \wt{B}\|_2}{\sigma_{m_3}(\Sym_{6D} H)(\sigma_{m_3}(\Sym_{6D} H) - \spec{E_1})}.
\end{align*}

We will now bound the second term in equation~\ref{eq:desymm1}:
\begin{align*}
\Norm{\Gamma_1(\Sym_{6D} \wt{H})^\dagger\wt{B}}_2 &\leq \spec{\Gamma_1}\spec{(\Sym_{6D} \wt{H})^\dagger}\|\wt{B}\|_2 \\
&\leq \spec{\Gamma_1}\left(\frac{1}{\sigma_{m_3}(\Sym_{6D}H)- \spec{E_1}}\right)(\|B\|_2 + \|\wt{B} - B\|_2)
\end{align*}

Let us bound all the parameters involved above:
\begin{enumerate}
\item $\spec{\Gamma} \leq \spec{\Gamma_1} \leq O(m^2 \|\calW - \wt{\calW}\|_F)$: We know that $W,\wt{W}$ are bases for $\calW,\wt{\calW}$ with $\wt{W} = W+\Gamma'$ with $\|\Gamma'\|_F \leq O(m\|\calW - \wt{\calW}\|_F)$. Let $W$ have columns $w_1,\ldots,w_m$ and $\Gamma'$ have columns $\gamma_1,\ldots,\gamma_m$. We have that $\spec{\Gamma} \leq \spec{\Gamma_1} \leq \|\wt{W}^{\ot 3} - W^{\ot 3}\|_F$, so bounding the latter gives:
\begin{align*}
\|\wt{W}^{\ot 3} - W^{\ot 3}\|_F^2 &= \|(W+\Gamma')^{\ot 3} - W^{\ot 3}\|_F^2 \\
&=\sum_{i,j,k\in [m]} \|(w_i+\gamma_i)\ot(w_j+\gamma_j)\ot(w_k+\gamma_k) - w_i \ot w_j\ot w_k\|_2^2 \\
&\leq \sum_{i,j,k\in [m]} O(\|\gamma_i \ot w_j \ot w_k\|_2^2)\\
&= O(\|\Gamma'\|_F^2 m^2) \mper
\end{align*}
where in the first inequality we used Cauchy-Schwarz and absorbed the terms that have more than $1$ $\gamma$-term into the terms with only one $\gamma$ term, since $\|\Gamma'\|_2 \ll 1$. So we get that $\spec{\Gamma}\leq\spec{\Gamma_1} \leq O(\|\Gamma'\|_F m) \leq O(m^2 \|\calW - \wt{\calW}\|_F)$.
\item $\spec{L} \leq \|L\|_F = O(m^{1.5})$.
\item $\spec{E_1} \leq \spec{\Sym_{6D}}\spec{\Gamma} \leq \sqrt{D^{O(D)}\ell_{6D}} \cdot O(m^2\|\calW - \wt{\calW}\|_F) \leq O(m^2\ell^{O(D)}\|\calW - \wt{\calW}\|_F)$. 
\item $\|B\|_2 \leq O(m\ell^{3D})$: By triangle inequality we get $\|B\|_2 \leq \sum_{t \in [m]}\|B_t(y)^{3D}\|_2$. Each summand is at most  $O(\ell^{3D})$ with probability $1 - \ell^{-\Omega(D)}$ (Claim~\ref{claim:norm-of-polynomial-powers}), therefore implying that with probability $1 - \ell^{-\Omega(D)}$, $\|B\|_2 \leq O(m\ell^{3D})$. 
\item $\|B - \wt{B}\|_2 = \|E \circ M\|_F \leq \|E\|_F$.
\item $\sigma_{m_3}(\Sym_{6D} H) \geq \Omega\left(\frac{1}{m^{1.5}}\right)$ using Lemma~\ref{lem:desymm-mainsv-hd}.
\end{enumerate}

Assuming $\|\calW - \wt{\calW}\|_F < 1/(m^{3.5}\ell^{O(D)})$ and plugging in all the parameters above gives that:
\begin{equation*}
    \Norm{T - \wt{T}}_2 \leq  \poly(m)(\ell^{O(D)}\|\calW - \wt{\calW}\|_F + \|E\|_F) \mper
    \qedhere
\end{equation*}
\end{proof}
\subsection{Proof of Lemma~\ref{lem:analysis-aggregation-hd}: Analysis of aggregating restrictions}
\label{sec:aggregation-hd}

Let us first recall Definition~\ref{def:restriction-matrix}: for a subset $S\subseteq [n]$ and matrix $A\in \R^{n \times n}$, we write $\calR_S(A)$ to be the matrix obtained by zeroing out the $(i,j)$ entry of $A$ if $i$ or $j$ is not in $S$.

\begin{lemma}[Restatement of Lemma~\ref{lem:analysis-aggregation-hd}]
    Let $D,n,\ell,m \in \N$ such that $6D \leq \ell \leq n$.
    There is an $n^{O(D)}$-time computable collection $\calS$ of subsets of $[n]$ such that each $S\in \calS$ satisfies $\ell \leq |S| \leq 6D\ell$ and that
    \begin{equation*}
        \E_{S\sim \calS} \sum_{t=1}^m \Paren{\Sym \Paren{ \calR_{S}(A_t)^{\ot D} }}^{\ot 3}
        = C \circ \sum_{t=1}^m \Paren{\Sym(A_t^{\ot D}) }^{\ot 3}
    \end{equation*}
    where $C \in (\R^n)^{\ot 6D}$ is a fixed tensor whose entries depend only on the entry locations, and each entry of $C$ has value within $((\ell/2n)^{6D}, 1)$.
\end{lemma}

We prove Lemma~\ref{lem:analysis-aggregation-hd} by constructing a pseudorandom family of hash functions similar to a $k$-wise independent hash family.
Specifically, given a parameter $k\in \N$, we construct a family $\calH$ of functions $[n] \to [n]$ that satisfies the following:
\begin{enumerate}
    \item For a subset $T \subseteq [n]$, for any $r\leq k$ and any distinct elements $x_1,\dots, x_r \in [n]$, the probability $\Pr_{h\sim \calH}[ \forall i \leq r,\ h(x_i) \in T]$ only depends on $r$ and $|T|$,
    
    \item For any subset $T\subseteq [n]$ of size $\ell$, the cardinality of $h^{-1}(T)$ is $\Theta(\ell)$ for all $h\in \calH$.
\end{enumerate}
Note that the most standard construction of $k$-wise independent hash functions consists of all degree $k-1$ univariate polynomials, including constant polynomials, over a field $\F$ of size $n$.
This only satisfies the first requirement but not the second due to the constant polynomials.
We make a simple modification to the standard construction to satisfy both requirements.

We first state the following standard fact.

\begin{fact} [Vandermonde matrix]
\label{fact:vandermonde}
    Fix a finite field $\F$.
    Let $r, k\in \N$ such that $r \leq k$, and let $x_1,\dots, x_r$ be distinct values in $\F$.
    Then, the following $r \times k$ Vandermonde matrix
    \begin{equation*}
        \begin{bmatrix}
            1 & x_1 & \cdots & x_1^{k-1} \\
            1 & x_2 & \cdots & x_2^{k-1} \\
            \vdots & \vdots & \ddots & \vdots \\
            1 & x_r & \cdots & x_r^{k-1}
        \end{bmatrix}
    \end{equation*}
    is rank $r$.
\end{fact}

The troublesome functions in the standard construction of $k$-wise independent hash family are the constant polynomials that don't satisfy the second requirement.
Thus, we simply delete those from our hash family.

\begin{lemma} \label{lem:hash-family}
    Fix a finite field $\F$ with $|\F|=q$ and let $k\in \N$, $k\geq 2$. Let $\calH$ be the following family of non-constant polynomials of degree $\leq k-1$:
    \begin{equation*}
        \calH = \{a_0 + a_1 x + \cdots + a_{k-1} x^{k-1} \mid a_0,\dots, a_{k-1}\in \F, a_1,\dots,a_{k-1} \text{ not all zero} \} \mper
    \end{equation*}
    Fix any $T \subseteq \F$ with $|T| \geq 2$.
    For any $r\leq k$, any distinct $x_1,\dots,x_r \in \F$,
    \begin{equation*}
        \Pr_{h\sim H} \Brac{\forall i\in[r],\ h(x_i) \in T} = \frac{q^{k-r} |T|^r - |T|}{q^k-q} \mper
    \end{equation*}
\end{lemma}
\begin{proof}
    For any distinct $x_1,\dots, x_r \in \F$ and any $b_1,\dots,b_r \in T$, we consider the following linear system with variables $(a_0,a_1,\dots,a_{k-1})$,
    \begin{equation}
    \label{eq:hash-functions}
        \begin{bmatrix}
            1 & x_1 & \cdots & x_1^{k-1} \\
            1 & x_2 & \cdots & x_2^{k-1} \\
            \vdots & \vdots & \ddots & \vdots \\
            1 & x_r & \cdots & x_r^{k-1}
        \end{bmatrix}
        \begin{bmatrix}
            a_0 \\ a_1 \\ \vdots \\ a_{k-1}
        \end{bmatrix}
        =
        \begin{bmatrix}
            b_1 \\ b_2 \\ \vdots \\ b_r
        \end{bmatrix}
        \mper
    \end{equation}
    
    Each solution to the above corresponds to a polynomial $h(x) = a_0 + a_1 x + \cdots + a_{k-1} x^{k-1}$ such that $h(x_i) = b_i$ for all $i\in[r]$.
    Let $V \in \F^{r \times k}$ be the Vandermonde matrix in \pref{eq:hash-functions}.
    By Fact~\ref{fact:vandermonde}, $\rank(V) = r$, thus the solutions to \pref{eq:hash-functions} form an affine subspace of dimension $k-r$ which contains $q^{k-r}$ vectors in $\F^{k}$.
    We split into two cases:
    \begin{itemize}
        \item $b_1 = b_2 = \cdots = b_r$: note that $(a_0, a_1,\dots,a_{k-1}) = (b_1,0,\dots,0)$ is a solution to \pref{eq:hash-functions}, which corresponds to the constant polynomial $h(x) = 1$.
        This is also the only constant polynomial that satisfies \pref{eq:hash-functions}.
        
        \item $b_1,\dots,b_r$ not all equal: no constant polynomial satisfies \pref{eq:hash-functions}.
    \end{itemize}
    Thus, there are
    \begin{equation*}
        |T| \cdot (q^{k-r}-1) + (|T|^r-|T|) \cdot q^{k-r} = q^{k-r} |T|^r - |T|
    \end{equation*}
    number of non-constant polynomials $h$ that satisfy $h(x_i) = b_i$ for all $i\in[r]$.
    Since $|\calH| = q^k - q$, this completes the proof.
\end{proof}

With Lemma~\ref{lem:hash-family}, we can construct a desired collection $\calS$ of subsets of $[n]$ where each $S\in \calS$ has bounded cardinality.

\begin{lemma} \label{lem:collection-of-subsets}
    Let $n\in \N$ be a prime power and given $\ell, k \in \N$ such that $2\leq k \leq \ell \leq n$.
    There exists a $n^{O(k)}$ time algorithm that outputs a collection $\calS$ of $n^{k}-n$ subsets of $[n]$ such that each $S\in \calS$ satisfies $\ell \leq |S| \leq (k-1)\ell$ and for any $r\leq k$ and distinct indices $i_1,\dots,i_r \in [n]$,
    \begin{equation*}
        \Pr_{S\sim \calS} \Brac{i_1,\dots,i_r \in S} = \frac{n^{k-r}\ell^r - \ell}{n^k-n} \mper
    \end{equation*}
\end{lemma}
\begin{proof}
    Let $\F$ be a field of size $n$ with a bijective map to $[n]$.
    With a slight abuse of notation, we will use $\F$ and $[n]$ interchangeably.
    Let $\calH$ be the set of non-constant polynomials of degree $\leq k-1$ defined in Lemma~\ref{lem:hash-family} such that $|\calH| = n^k - n$.
    Pick any fixed subset $T\subseteq \F$ of size $\ell$.
    For any $h\in \calH$, let $h^{-1}(T) \coloneqq \{x\in \F: h(x) \in T\}$, which we also view as a subset of $[n]$.
    Define
    \begin{equation*}
        \calS \coloneqq \{h^{-1}(T) \mid h \in \calH\} \mper
    \end{equation*}
    We will prove that $\calS$ is the desired collection of subsets.
    
    First, fix an $h\in \calH$.
    For each $b\in T$, since $h$ is not a constant polynomial and is of degree $\leq k-1$, $h(x) = b$ must have at least one solution and at most $k-1$ solutions.
    Thus,
    \begin{equation*}
        |T| \leq \Abs{h^{-1}(T)} \leq (k-1)|T|\mcom
    \end{equation*}
    meaning each $S\in \calS$ satisfies $\ell \leq |S| \leq (k-1)\ell$.
    
    Next, for $r\leq k$ and any distinct indices $i_1,\dots, i_r \in [n]$, let $x_1,\dots, x_r\in \F$ be their corresponding field elements.
    The probability that $i_1,\dots,i_r \in S$ over $S\sim \calS$ is exactly the probability that $h(x_1),\dots, h(x_r) \in T$ over $h\sim \calH$, hence by Lemma~\ref{lem:hash-family},
    \begin{equation*}
        \Pr_{S\sim \calS} \Brac{i_1,\dots,i_r \in S} = \frac{n^{k-r}\ell^r - \ell}{n^k-n} \mper
        \qedhere
    \end{equation*}
\end{proof}

Lemma~\ref{lem:analysis-aggregation-hd} is a simple corollary of Lemma~\ref{lem:collection-of-subsets}.

\begin{proof}[Proof of Lemma~\ref{lem:analysis-aggregation-hd}]
    The construction $\calS$ from Lemma~\ref{lem:collection-of-subsets} with parameter $k = 6D$ satisfies that all $S \in \calS$ has $\ell \leq |S| \leq 6D\ell$.
    Next, consider a multiset $I = \{i_1,\dots, i_{6D}\} \subseteq [n]$ of size $6D$, which we view as an index of a $6D$-th order symmetric tensor.
    Let $J = \{j_1,\dots, j_r\}$ be the set of $r$ unique elements in $I$.
    
    Since $\calR_S(A_t)$ is just zeroing out entries of $A_t$, for any $S\subseteq [n]$,
    \begin{equation*}
        \Paren{\Sym \Paren{\calR_S(A_t)^{\ot D} }}^{\ot 3}[I]
        = \1(J \subseteq S) \cdot \Sym\Paren{A_t^{\ot 3D}} [I]
    \end{equation*}

    By the guarantee of Lemma~\ref{lem:collection-of-subsets},
    \begin{equation*}
    \begin{aligned}
        \E_{S\sim \calS}\Brac{ \Paren{\Sym \Paren{\calR_S(A_t)^{\ot D} }}^{\ot 3}[I] }
        &= \Pr_{S\sim \calS}[J \subseteq S] \cdot \Sym\Paren{A_t^{\ot 3D}} [I] \\
        &= \frac{n^{6D-r} \ell^{r} - \ell}{n^{6D}-n} \cdot \Sym\Paren{A_t^{\ot 3D}} [I]
    \end{aligned}
    \end{equation*}
    Since $r$ can only range from $1$ to $6D$, $\frac{n^{6D-r} \ell^{r} - \ell}{n^{6D}-n}$ is between $(\ell/2n)^{6D}$ and $1$.
    Moreover, this coefficient only depends on the number of unique elements in $I$.
    This implies that when viewed as a coefficient tensor indexed by $I$, the entries only depend on the entry locations.
    This completes the proof.
\end{proof}

\subsection{Proof of Lemma~\ref{lem:analysis-D-root-hd}: \texorpdfstring{$D$}{D}th roots of polynomials}
\label{sec:D-root-hd}

In this section, we give an algorithm to desymmetrize a single noisy $D$th power and complete the proof of Lemma~\ref{lem:analysis-D-root-hd} that analyzes the last step in the algorithm. We restate it before continuing.

\begin{lemma}[Stable Computation of $D$-th Roots, Lemma~\ref{lem:analysis-D-root-hd} restated]
    Let $D,n \in \N$ and $\delta \geq 0$.
    Let $P \in \R^{n\times n}$ be an unknown symmetric matrix.
    Suppose $\wt{P_D}(x)$ is a homogeneous degree-$D$ polynomial in $n$ variables such that its coefficient tensor satisfies $\Norm{\wt{P_D} - \Sym(P^{\otimes D})}_F \leq \delta$.
    There is an algorithm that runs in $n^{O(D)}$ time and outputs $\wt{Q} \in \R^{n\times n}$ such that if $D$ is odd, then
    \begin{equation*}
        \Norm{\wt{Q} - P}_F \leq O(\sqrt{n}\delta^{1/D}) \mcom
    \end{equation*}
    and if $D$ is even, then
    \begin{equation*}
        \min_{\sigma \in \{\pm1\}} \Norm{ \wt{Q} - \sigma P }_{F} \leq O(n \delta^{1/3D}) \cdot \|P\|_{\max} \mper
    \end{equation*}
\end{lemma}

Our algorithm is based on a sum-of-squares semidefinite relaxation for desymmetrization. We recall some minimal background needed to describe it below.
\subsubsection{Background on Sum-of-Squares}
We will use the sum-of-squares semidefinite programming method to design and analyze the algorithm to prove Lemma~\ref{lem:analysis-D-root-hd}. We direct the reader to the monograph~\cite{TCS-086} and the survey~\cite{BarakS14} for more details.

A \emph{pseudo-distribution} on $\R^n$ is a finitely supported \emph{signed} measure $\mu :\R^n \rightarrow \R$ such that $\sum_{x: \mu(x) \neq 0} \mu(x) = 1$. The associated \emph{pseudo-expectation} is a linear operator that assigns to every polynomial $f:\R^n \rightarrow \R$, the value $\pE_\mu f = \sum_{x: \mu(x) \neq 0} \mu(x) f(x)$ that we call the pseudo-expectation of $f$. We say that a pseudo-distribution $\mu$ on $\R^n$ has degree $d$ if $\pE_\mu[f^2] \geq 0$ for every polynomial $f$ on $\R^n$ of degree $\leq d/2$.

A pseudo-distribution of degree $d$ is said to satisfy a constraint $\{q \geq 0\}$ for any polynomial $q$ of degree $\leq d$ if for every square polynomial $p$ of degree $\leq d-\deg(q)$, $\pE_\mu[ p q] \geq 0$. We say that $\mu$ $\tau$-approximately satisfies a constraint $\{q \geq 0\}$ if for any sum-of-squares polynomial $p$, $\pE_\mu[pq] \geq - \tau \Norm{p}_2$ where $\Norm{p}_2$ is the $\ell_2$ norm of the coefficient vector of $p$. 

We will rely on the following basic fact about the pseudo-expectations. 

\begin{fact}[Pseudo-distribution Jensen] \label{fact:pseudo-expectation Jensen}
For any pseudo-distribution $\mu$ of degree $d$ on $\R^n$ and any polynomial $p$ such that $p^{2k}$ is of degree at most $d$, $\pE_{\mu}[p^{2k}] \geq \pE_{\mu}[p^2]^{k}$. 
\end{fact}

Of use to us is the following basic connection that forms the basis of the sum-of-squares algorithm. 

\begin{fact}[Sum-of-Squares Algorithm, \cite{parrilo2000structured,Lasserre01}] \label{fact:sos-algorithm}
Given a system of degree $\leq d$ polynomial constraints $\{q_i \geq 0\}$ in $n$ variables and the promise that there is a degree-$d$ pseudo-distribution satisfying $\{q_i \geq 0\}$ as constraints, there is a $n^{O(d)} \polylog( 1/\tau)$ time algorithm to find a pseudo-distribution of degree $d$ on $\R^n$ that $\tau$-approximately satisfies the constraints $\{q_i \geq 0\}$.
\end{fact}

An (unconstrained) \emph{sum-of-squares} proof of non-negativity of a polynomial $p$ is an identity of the form $p = \sum_i q_i^2$ for polynomials $q_1, q_2, \ldots$. A \emph{sum-of-squares} proof of non-negativity of a polynomial $p$ on the solution set of a system of polynomial inequalities $\{r_i \geq 0\}_{i \leq R}$ is an identity of the form $p = \sum_i q_i^2 + \sum_{S \subseteq [R]} p_S \prod_{i\in S} r_i$. The \emph{degree} of such a sum-of-squares proof equals the maximum of the degree of $q_i^2$ and $\deg(p_S) + \prod_{i \in S} r_i$ over all $S$ appearing in the sum above. We write $\{q_i \geq 0\} \sststile{t}{} \{p \geq 0\}$ where $t$ is the degree of the sum-of-squares proof.

We will rely on the following basic connection between sum-of-squares proofs:

\begin{fact} \label{fact:proofs-to-pseudo-dist}
Suppose $\{q_i \geq 0\} \sststile{t}{} \{p \geq 0\}$ for some polynomials $q_i$ and $p$. Let $\mu$ be a pseudo-distribution satisfying $\{q_i \geq 0\}$ of degree $\geq t$. Then, $\pE_\mu[p] \geq 0$.
\end{fact}

We will also use the following sum-of-squares version of the almost triangle inequality. 

\begin{fact}[SoS Almost Triangle Inequality] \label{fact:sos-almost-triangle}
For indeterminates $a,b$ and any $k \in \N$, 
\[
\sststile{2k}{} \Set{(a+b)^{2k} \leq 2^{2k} (a^{2k} + b^{2k})}\mper
\]
\end{fact}

\subsubsection{Algorithm for \texorpdfstring{$D=2$}{D=2}: Square root of polynomial}
We first look at the case when $D=2$. Given a noisy version of a degree-$4$ polynomial $A(x)^2$ where $A(x) = x^\top A x$ for an unknown symmetric matrix $A$, we hope to construct $A$.
We remark that we can only recover up to sign.

\begin{lemma}[Taking square root of polynomials]
\label{lem:square-root-of-polynomial}
    Let $P \in \R^{n\times n}$ be an unknown symmetric matrix, and let $\delta = o(1) \cdot \|P\|_{\max}$.
    Suppose $\wt{P_2}(x)$ be a homogeneous degree-4 polynomial in $n$ variables such that its coefficient tensor satisfies $\Norm{\wt{P_2} - \Sym(P^{\ot 2})}_{\max} \leq \delta$. 
    Then, there is an algorithm that runs in $O(n^2)$ time and outputs a symmetric matrix  $\wt{Q} \in \R^{n\times n}$ such that
    \begin{equation*}
        \min_{\sigma \in \{\pm1\}} \Norm{ \wt{Q} - \sigma P }_{\max} \leq O(\delta^{1/6}) \cdot \|P\|_{\max} \mper
    \end{equation*}
\end{lemma}

We will need the following simple fact:

\begin{lemma}
\label{fact:square-root-error}
    Let $a > 0$ and $\delta\in \R$ such that $|\delta| \leq a^2$.
    Then $\Paren{\sqrt{a^2+\delta}-a}^2 \leq |\delta|$.
\end{lemma}
\begin{proof}
    If $\delta \geq 0$, then $\sqrt{a^2+\delta}+a \geq \sqrt{a^2+\delta}-a \geq 0$, and thus $\delta = \Paren{\sqrt{a^2+\delta}+a} \Paren{\sqrt{a^2+\delta}-a} \geq \Paren{\sqrt{a^2+\delta}-a}^2$.
    Similarly if $\delta \leq 0$, then $a + \sqrt{a^2+\delta} \geq a - \sqrt{a^2+\delta} \geq 0$, and similar calculations show that $|\delta| = \Paren{a+\sqrt{a^2+\delta}} \Paren{a-\sqrt{a^2+\delta}} \geq \Paren{a-\sqrt{a^2+\delta}}^2$.
\end{proof}

\begin{proof}[Proof of Lemma~\ref{lem:square-root-of-polynomial}]
    Denote $P_2 \coloneqq \Sym(P^{\ot 2})$, which is unknown to us.
    We will use the following identities for $i \neq j \neq k \neq \ell$:
    \begin{itemize}
        \item $P_2[i,i,i,i] = P[i,i]^2$.
        \item $P_2[i,i,i,j] = P[i,i] P[i,j]$.
        \item $P_2[i,i,j,j] = \frac{1}{3} P[i,i] P[j,j] + \frac{2}{3} P[i,j]^2$.
        \item $P_2[i,i,j,k] = \frac{1}{3} P[i,i] P[j,k] + \frac{2}{3} P[i,j] P[i,k]$.
        \item $P_2[i,j,k,\ell] = \frac{1}{3}(P[i,j] P[k,\ell] + P[i,k] P[j, \ell] + P[i,\ell] P[j,k])$.
    \end{itemize}
    First observe that from $P_2[i,i,i,i]$ we have the approximate magnitudes of the diagonal entries of $P$.
    Our algorithm simply constructs the matrix $\wt{Q}$ entry by entry.
    We split into two cases depending on the diagonal.
    
    \paragraph{$P$ has a large diagonal entry.} Suppose $\wt{P_2}[i,i,i,i] > 4\|P\|_{\max}^2 \delta^{1/3}$ for some $i\in[n]$.
    Then, we do the following:
    \begin{enumerate}
        \item Set $\wt{Q}[i,i] = \sqrt{\wt{P_2}[i,i,i,i]}$.
        \item For $j\neq i$, set $\wt{Q}[i,j] = \frac{\wt{P_2}[i,i,i,j]}{\wt{Q}[i,i]}$.

        \item For $j, k \neq i$, set $\wt{Q}[j,k] = \frac{3\wt{P_2}[i,i,j,k] - 2\wt{Q}[i,j] \wt{Q}[i,k]}{\wt{Q}[i,i]}$.
    \end{enumerate}
    By Fact~\ref{fact:square-root-error}, $\Abs{\wt{Q}[i,i] - \Abs{P[i,i]}} \leq \sqrt{\delta}$.
    We can assume without loss of generality that $\Abs{P[i,i]} = P[i,i]$ because if it's the opposite, we will simply recover $\wt{Q} \approx -P$.
    
    Next, for $j\neq i$, since $P_2[i,i,i,j] = P[i,i] P[i,j]$ and $\wt{Q}[i,i] > 2\|P\|_{\max} \delta^{1/6}$,
    \begin{equation*}
    \begin{aligned}
        \Abs{\wt{Q}[i,j] - P[i,j]}
        &= \Abs{\frac{\wt{P_2}[i,i,i,j] - \wt{Q}[i,i]P[i,j]}{\wt{Q}[i,i]} } \\
        &\leq \Abs{ \frac{(P[i,i]P[i,j] \pm \delta) - (P[i,i]\pm \sqrt{\delta}) P[i,j]}{\wt{Q}[i,i]} } \\
        &\leq \frac{\|P\|_{\max} \sqrt{\delta} + \delta}{\wt{Q}[i,i]}
        \leq  \delta^{1/3} \mper
    \end{aligned}
    \end{equation*}
    Here we assume that $\delta \ll \|P\|_{\max}^2$.
    
    Finally, for $i \neq j,k$ (here $j$ can equal $k$), similar calculations show that
    \begin{equation*}
        \Abs{ \wt{Q}[j,k] - P[j,k] } = \Abs{ \frac{3\wt{P_2}[i,i,j,k] - 2\wt{Q}[i,j] \wt{Q}[i,k] - \wt{Q}[i,i] P[j,k]}{\wt{Q}[i,i]} } \leq \frac{4\|P\|_{\max} \delta^{1/4} + O(\delta^{1/2})}{\wt{Q}[i,i]}
        \leq 4\delta^{1/6} \mper
    \end{equation*}
    
    Therefore, we obtain $\wt{Q}$ such that for all $i,j\in[n]$, $\Abs{ \wt{Q}[i,j] - P[i,j] } \leq 4\delta^{1/6}$.

    \paragraph{$P$ has small diagonal entries.}
    Suppose $\wt{P_2}[i,i,i,i] \leq 4\|P\|_{\max}^2 \delta^{1/3}$ for all $i\in[n]$.
    Since $P_2[i,i,j,j] = P[i,i]P[j,j] + P[i,j]^2$, there must be a pair $(i,j)$ such that $\wt{P_2}[i,i,j,j] \geq \|P\|_{\max}^2(1-O(\delta^{1/3}))$.
    Then we do the following:
    
    \begin{enumerate}
        \item Set $\wt{Q}[i,j] = \sqrt{\wt{P_2}[i,i,j,j]}$.
        
        \item Set $\wt{Q}[k,k] = 0$ for all $k\in[n]$.
        
        \item For $k\neq i \neq j$, set $\wt{Q}[i,k] = \frac{3\wt{P_2}[i,i,j,k]}{2 \wt{Q}[i,j]}$ and $\wt{Q}[j,k] = \frac{3\wt{P_2}[i,j,j,k]}{2 \wt{Q}[i,j]}$.
        
        \item For $k \neq \ell \neq i \neq j$, set $\wt{Q}[k,\ell] = \frac{3\wt{P_2}[i,j,k,\ell] - \wt{Q}[i,k] \wt{Q}[j,\ell] - \wt{Q}[i,\ell] \wt{Q}[j,k]}{\wt{Q}[i,j]}$.
    \end{enumerate}
    
    First, for each $k\in[n]$, $\wt{Q}[k,k] = 0$, thus
    \begin{equation*}
        \Abs{ \wt{Q}[k,k] - P[k,k] } = \Abs{P[k,k]} \leq 2\|P\|_{\max} \delta^{1/6} \mper
    \end{equation*}
    And since $P[i,j]^2 = \wt{P_2}[i,i,j,j] - P[i,i] P[j,j] \pm \delta = \wt{P_2}[i,i,j,j] \pm 4\|P\|_{\max}^2 \delta^{1/3} \pm O(\delta)$, by Fact~\ref{fact:square-root-error},
    \begin{equation*}
        \Abs{ \wt{Q}[i,j] - |P[i,j]| } \leq 2 \|P\|_{\max} \delta^{1/6} \mper
    \end{equation*}
    We can again assume without loss of generality that $|P[i,j]| = P[i,j]$.
    
    Next, for $k\neq i,j$, since $\wt{Q}[i,j] \geq \|P\|_{\max}(1-O(\delta^{1/3}))$ and $|P[i,k]| \leq \|P\|_{\max}$,
    \begin{equation*}
    \begin{aligned}
        \Abs{ \wt{Q}[i,k] - P[i,k] } 
        &= \Abs{ \frac{3 \wt{P_2}[i,i,j,k] - 2\wt{Q}[i,j] P[i,k]}{2\wt{Q}[i,j]} } \\
        &= \Abs{ \frac{ (P[i,i]P[j,k] + 2P[i,j] P[i,k] \pm \delta) - 2 (P[i,j] \pm 2\|P\|_{\max} \delta^{1/6}) P[i,k]}{2\wt{Q}[i,j]} } \\
        &\leq 6\|P\|_{\max} \delta^{1/6} \mper
    \end{aligned}
    \end{equation*}
    For $k \neq \ell \neq i \neq j$, similar calculations show that
    \begin{equation*}
        \Abs{ \wt{Q}[k,\ell] - P[k,\ell] }
        = \Abs{ \frac{3 \wt{P_2}[i,j,k,\ell] - \wt{Q}[i,k] \wt{Q}[j,\ell] - \wt{Q}[i,\ell] \wt{Q}[j,k] - \wt{Q}[i,j] P[k,\ell]}{\wt{Q}[i,j]} }
        \leq O(\|P\|_{\max} \delta^{1/6}) \mper
    \end{equation*}
    Therefore, we obtain $\wt{Q}$ such that for all $i,j\in[n]$, $\Abs{ \wt{Q}[i,j] - P[i,j] } \leq O(\|P\|_{\max}\delta^{1/6})$.
\end{proof}

\subsubsection{Algorithm for the general case: \texorpdfstring{$D$}{D}-th root of polynomial}
We now focus on taking $D$-th of polynomials for $D\geq 2$.
We will use the SoS algorithm.

\begin{lemma}
\label{lem:sum-of-ab}
    Let $a,b$ be indeterminates, and let $k \in \N$ be an even integer, $k\geq 2$. Then,
    \begin{equation*}
        \sststile{k}{a,b} \Set{\sum_{i=0}^k a^i b^{k-i} \geq \frac{1}{2}(a^k + b^k) } \mper
    \end{equation*}
\end{lemma}
\begin{proof}
    Observe that for odd $i \leq k-1$, $a^i b^{k-i} = (ab)\cdot a^{i-1} b^{k-i-1}$ where $a^{i-1} b^{k-i-1}$ is a squared polynomial of degree $k-2$.
    Further, we know that $\sststile{2}{a,b} \Set{ ab \geq -\frac{1}{2}(a^2+b^2) }$.
    Thus, we split $\sum_{i=0}^k a^i b^{k-i}$ into even and odd terms:
    \begin{equation*}
        \sststile{k}{a,b} \sum_{i=0}^k a^i b^{k-i}
        = \sum_{i \text{ even}} a^i b^{k-i} + \sum_{i \text{ odd}} a^i b^{k-i}
        \geq \sum_{j=0}^{k/2} a^{2j} b^{k-2j} - \frac{1}{2}(a^2+b^2) \sum_{j=0}^{k/2-1} a^{2j} b^{k-2-2j} \mcom
    \end{equation*}
    using the fact that $a^{2j} b^{k-2-2j}$ is a squared polynomial.
    Next, observe that
    \begin{equation*}
        \frac{1}{2}(a^2+b^2)(a^{k-2} + a^{k-4} b^2 + \cdots + a^2 b^{k-4} + b^{k-2}) = \frac{1}{2}(a^k + b^k) + (a^{k-2}b^2 + \cdots + a^2 b^{k-2}) \mper
    \end{equation*}
    Thus, we have
    \begin{equation*}
        \sststile{k}{a,b} \Set{ \sum_{i=0}^k a^i b^{k-i} \geq \frac{1}{2}(a^k + b^k) } \mper
        \qedhere
    \end{equation*}
\end{proof}

\begin{lemma} \label{lem:a-minus-b}
    Let $a,b$ be indeterminates, and let $k \in \N$ be an odd integer.
    Then,
    \begin{equation*}
        \sststile{2k}{a,b} \Set{ (a-b)^{2k} \leq 2^{2(k-1)} (a^k - b^k)^2 } \mper
    \end{equation*}
\end{lemma}
\begin{proof}
    The statement is clearly true for $k=1$, so we can assume that $k \geq 3$.
    First note that $(a^k - b^k)^2 = (a-b)^2 \Paren{\sum_{i=0}^{k-1} a^i b^{k-1-i}}^2$.
    Since $k-1$ is even, by Lemma~\ref{lem:sum-of-ab} we have that $\sststile{k-1}{a,b} \Set{ \sum_{i=0}^{k-1} a^i b^{k-1-i} \geq \frac{1}{2}(a^{k-1} + b^{k-1}) }$, hence
    \begin{equation*}
        \sststile{2k}{a,b} \Set{ (a^k - b^k)^2 \geq (a-b)^2 \cdot \frac{1}{4}(a^{k-1} + b^{k-1})^2 } \mper
    \end{equation*}
    Next, since $k-1$ is even, by SoS almost triangle inequality (Fact~\ref{fact:sos-almost-triangle}),
    \begin{equation*}
        \sststile{k-1}{a,b} \Set{ (a-b)^{k-1} \leq 2^{k-2} (a^{k-1} + b^{k-1}) } \mper
    \end{equation*}
    Thus, $\sststile{2k}{a,b} \Set{(a-b)^{2k} \leq (a-b)^2 \cdot 2^{2(k-2)} (a^{k-1}+b^{k-1})^2}$, and we have
    \begin{equation*}
        \sststile{2k}{a,b} \Set{ 2^{2(k-1)} (a^k - b^k)^2 - (a-b)^{2k} \geq 0 } \mper
    \end{equation*}
    This completes the proof.
\end{proof}

\begin{lemma} \label{lem:a-squared-minus-b-squared}
    Let $a,b$ be indeterminates, and let $k \in \N$.
    Then,
    \begin{equation*}
        \sststile{4k}{a,b} \Set{ (a^2 - b^2)^{2k} \leq 2^{2(k-1)} (a^{2k} - b^{2k})^2} \mper
    \end{equation*}
\end{lemma}
\begin{proof}
    First note that $a^{2k} - b^{2k} = (a^2 - b^2) \sum_{i=0}^{k-1} a^{2(k-1-i)} b^{2i}$, thus
    \begin{equation*}
        2^{2(k-1)}(a^{2k} - b^{2k})^2 - (a^2 - b^2)^{2k}
        = (a^2-b^2)^2 \Paren{ \Paren{2^{k-1} \sum_{i=0}^{k-1} a^{2(k-1-i)} b^{2i} }^2 - (a^2-b^2)^{2(k-1)} } \mper
    \end{equation*}
    Next, $(a^2-b^2)^{k-1} = \sum_{i=0}^{k-1} \binom{k-1}{i} a^{2(k-1-i)} b^{2i}$,
    \begin{equation*}
        2^{k-1} \sum_{i=0}^{k-1} a^{2(k-1-i)} b^{2i} \pm (a^2 - b^2)^{k-1}
        = \sum_{i=0}^{k-1} \Paren{2^{k-1} \pm (-1)^i \binom{k-1}{i}} a^{2(k-1-i)} b^{2i} \mcom
    \end{equation*}
    which are both sum-of-squares polynomials since $\binom{k-1}{i} < 2^{k-1}$.
    By the fact that $a^2 - b^2 = (a-b)(a+b)$, we get that $2^{2(k-1)}(a^{2k} - b^{2k})^2 - (a^2 - b^2)^{2k}$ is a sum-of-squares polynomial, completing the proof.
\end{proof}

\begin{proof}[Proof of Lemma~\ref{lem:analysis-D-root-hd}]
    Let $Q$ be an $n \times n$ symmetric matrix-valued indeterminate.
    We solve for a degree-$2D$ pseudo-distribution $\mu$ in $Q$ such that
    \begin{equation*}
        \pE_{\mu} \Norm{ \Sym(Q^{\otimes D}) - \wt{P_D} }_F^2 \leq \delta^2 \mper
    \end{equation*}
    Note that $\Norm{ \Sym(Q^{\otimes D}) - \wt{P_D} }_F^2$ is a degree-$2D$ polynomial in $D$.
    Such a pseudo-distribution is feasible since the matrix $P$ satisfies the inequality.
    Then, since $\sststile{2}{a,b} (a+b)^2 \geq \frac{1}{2}a^2 - b^2$,
    \begin{equation*}
        \pE_{\mu} \Norm{ \Sym(Q^{\otimes D}) - \wt{P_D} }_F^2
        \geq \frac{1}{2} \pE_{\mu} \Norm{ \Sym(Q^{\otimes D}) - \Sym(P^{\otimes D}) }_F^2 - \pE_{\mu} \Norm{ \Sym(P^{\otimes D}) - \wt{P_D} }_F^2 \mper
    \end{equation*}
    Thus, we can conclude that
    \begin{equation} \label{eq:SymQD-minus-SymPD}
        \pE_{\mu} \Norm{ \Sym(Q^{\otimes D}) - \Sym(P^{\otimes D}) }_F^2 \leq 4\delta^2 \mper
    \end{equation}
    Consider a fixed $z\in \R^{n}$ such that $\|z\|_2 = 1$.
    With slight abuse of notation, we denote $Q(z) = z^\top Qz$ as a linear polynomial in $Q$.
    Then, by pseudo-distribution Jensen inequality (Fact~\ref{fact:pseudo-expectation Jensen}),
    \begin{equation} \label{eq:QZ-minus-PZ}
    \begin{aligned}
        \pE_{\mu} \Paren{Q(z)^D - P(z)^D}^2
        &= \pE_{\mu}\Brac{ \Iprod{\Sym(Q^{\otimes D}) - \Sym(P^{\otimes D}), z^{\otimes 2D}}^2 } \\
        &\leq \pE_{\mu}  \Norm{ \Sym(Q^{\otimes D}) - \Sym(P^{\otimes D}) }_F^2 \cdot \|z\|_2^{4D} \\
        &\leq 4\delta^2 \mper
    \end{aligned}
    \end{equation}
    
    \paragraph{$D$ is odd.}
    By pseudo-distribution Jensen inequality (Fact~\ref{fact:pseudo-expectation Jensen}), \pref{lem:a-minus-b} and \pref{eq:QZ-minus-PZ}, we have that
    \begin{equation*}
        \Paren{\pE_{\mu} [Q(z)-P(z)] }^{2D}
        \leq \pE_{\mu} \Paren{Q(z)-P(z)}^{2D}
        \leq 2^{2(D-1)} \cdot \pE_{\mu} \Paren{Q(z)^D - P(z)^D}^2
        \leq 2^{2D} \delta^2 \mper
    \end{equation*}
    The above holds for any unit vector $z \in \R^{n}$, and since $\pE_{\mu}[Q(z)-P(z)] = \Iprod{\pE_{\mu} [Q]- P, zz^\top}$,
    \begin{equation*}
        \Norm{ \pE_{\mu} [Q]- P}_2 = \max_{z:\|z\|_2=1} \Abs{\Iprod{\pE_{\mu} [Q]- P, zz^\top}} \leq 2 \delta^{1/D} \mper
    \end{equation*}
    As Frobenius norm and spectral norm is the same up to a factor $\sqrt{n}$, we get that $\Norm{ \pE_{\mu} [Q]- P}_F \leq 2\sqrt{n} \delta^{1/D}$.
    Thus, setting $\wt{Q} = \pE_{\mu}[Q]$, we get the desired result.
    
    \paragraph{$D$ is even.}
    By pseudo-distribution Jensen inequality (Fact~\ref{fact:pseudo-expectation Jensen}), \pref{lem:a-squared-minus-b-squared} (with $k=D/2$) and \pref{eq:QZ-minus-PZ}, we have that
    \begin{equation*}
    \begin{aligned}
        \pE_{\mu} \Brac{Q(z)^2 - P(z)^2}^D
        &\leq \pE_{\mu} \Brac{\Paren{Q(z)^2 - P(z)^2}^D } \\
        &\leq 2^{D-2} \cdot \pE_{\mu} \Paren{Q(z)^D - P(z)^D}^2
        \leq 2^{D} \delta^2 \mper
    \end{aligned}
    \end{equation*}
    Thus, $\Abs{\pE_{\mu} \Brac{Q(z)^2} - P(z)^2} \leq 2\delta^{2/D}$ holds for all unit vectors $z$.
    This means that we get a tensor $\wt{P_2}$ such that
    \begin{equation*}
        \Norm{ \wt{P_2} - \Sym(P^{\ot 2}) }_{\max} \leq 2\delta^{2/D} \mper
    \end{equation*}
    Now we can apply Lemma~\ref{lem:square-root-of-polynomial} and obtain a symmetric matrix $\wt{Q}\in \R^{n\times n}$ such that
    \begin{equation*}
        \min_{\sigma \in \{\pm1\}} \Norm{ \wt{Q} - \sigma P }_{\max} \leq O(\delta^{1/3D}) \cdot \|P\|_{\max} \mper
    \end{equation*}
    This implies a Frobenius norm bound $\min_{\sigma \in \{\pm1\}} \Norm{ \wt{Q} - \sigma P }_{F} \leq O(n\delta^{1/3D}) \cdot \|P\|_{\max}$.
\end{proof}

\section{Decomposing Power-Sums of High-degree Polynomials}
In this section, we give a generalization of our algorithm for decomposing power-sums of quadratics to power-sums of arbitrary degree polynomials. The outline of the algorithm remains the same with one main difference in the span finding step (Algorithm~\ref{algo:span-finding-hd}).
For the quadratic case, we sample a random polynomial $p(y)$ of degree $2$ and work with the polynomial $p(y)^D$ in span-finding; for generic degree-$K$ polynomials, we will  instead sample a random polynomial of degree $2(K-1)D$.

For the technical analysis, despite following the same outline, there are two major differences from the quadratic case: 1) the reduction from $\calU_D = \calV_D$ to proving feasibility of a linear system (described in Section~\ref{sec:U-equals-V-hd}), and 2) the analysis of $V$ and its null space (described in Section~\ref{sec:analysis-of-V-hd}).
For (1), the difference arises from the structure of the partial derivatives; we will derive the analogous linear system in Section~\ref{sec:U-equals-V-high-deg}.
For (2), the major difference arises from the fact that the equation $\sum B_t(y)^D p_t(y) + G(y)q(y) = 0$ (equation~\ref{eq:null-space-high-deg}) for a random degree $2(K-1)D$ polynomial $G(y)$ and random degree $K$ polynomials $B_t(y)$, has a lot more solutions; we will state the analogous $V$ matrix and identify the correct null space in Section~\ref{sec:analysis-of-V-high-deg}.




The main result of this section is the following theorem that is analogous to Theorem~\ref{thm:main-theorem-hd}.

\begin{theorem} \label{thm:main-theorem-high-degree}
    There is an algorithm that takes input parameters $n,m,D,K\in \N$, an accuracy parameter $\tau>0$, and the coefficient tensor $\wh{P}$ of a degree-$3KD$ polynomial $\wh{P}$ in $n$ variables with total bit complexity $\size(\wh{P})$, runs in time $(\size(\wh{P})n)^{O(KD)} \polylog (1/\tau)$, and outputs a sequence of symmetric tensors $\wt{A}_1,\wt{A}_2,\ldots, \wt{A}_m \in (\R^{n})^{\ot K}$ with the following guarantee.

    Suppose $\wh{P}(x) = \sum_{t=1}^m A_t(x)^{3D} + E(x)$ where each $A_t$ is a symmetric $K$-th order tensor of independent $\calN(0,1)$ entries, $\|E\|_F\leq n^{-O(D)}$, and $m \leq (\frac{n}{\polylog(n)})^{\frac{2KD}{5K-4}}$.
    Then, with probability at least $0.99$ over the draw of $A_t$s and internal randomness of the algorithm, for odd $D$,
    \begin{equation*}
        \min_{\pi \in \bbS_m} \max_{t \in [m]} \Norm{\wt{A}_t - A_{\pi(t)}}_F \leq \poly(n) \Paren{ \|E\|_F^{1/D} + \tau^{1/D} } \mcom
    \end{equation*}
    and for even $D$,
    \begin{equation*}
        \min_{\pi \in \bbS_m} \max_{t \in [m]} \min_{\sigma\in\{\pm1\}} \Norm{\wt{A}_t - \sigma A_{\pi(t)}}_F \leq \poly(n) \Paren{ \|E\|_F^{1/3D} + \tau^{1/3D} } \mper
    \end{equation*}
\end{theorem}

\begin{remark}
    The upper bound on the number of polynomials $m \leq \wt{O}(n^{\frac{2KD}{5K-4}})$ comes from the requirement $m \leq \wt{O}(\ell^{KD/2})$ and $m\ell^{2D(K-1)} \leq \wt{O}(n^{2D})$ from Lemma~\ref{lem:analysis-restricted-partials-high-degree} and Lemma~\ref{lem:span-finding-high-degree}.
    The best $m$ we get is by setting $\ell = n^{\frac{4}{5K-4}}$ (which is $o(n)$ when $K \geq 2$), which gives our bound on $m$.
    Note that for $K=2$, we get $m \leq \wt{O}(n^{2D/3})$.
\end{remark}

The first step of the algorithm is identical: given a noisy version of the polynomial $P(x) = \sum_{t=1}^m A_t(x)^{3D}$ where $A_t(x)$ is a degree-$K$ polynomial, we take $2D$ partial derivatives and restrict to a fixed subset of $\ell=o(n)$ variables, resulting in a polynomial of degree $(3K-2)D$.
We first define the notion of the relevant restriction operator analogous to Definition~\ref{def:restriction-matrix}.

\begin{definition}[Restriction operator]
\label{def:restriction-operator-high-degree}
    Given a restriction matrix $M$ corresponding to a set $S \subseteq [n]$ with $|S| = \ell$, we associate it with an operator $\calR_M$ defined as follows:
    for a $k$-th order tensor $T \in (\R^{n})^{\ot k}$, we write $\calR_M(T) \in (\R^{n})^{\ot k}$ 
    which is the symmetric $k$-th order tensor obtained by zeroing out the entry of $T$ indexed by $I$ if any index in $I$ is not in $S$.
\end{definition}

Given a restriction matrix $M$, let $B_t(y) = A_t \circ M(y)$, and
\begin{equation*}
    \calV_D \coloneqq \spn\Paren{ B_t(y)^D y_T \mid t\in[m], T\in[\ell]^{2(K-1)D} } \mper
\end{equation*}
Same as the analysis of Lemma~\ref{lem:analysis-restricted-partials-hd}, we prove that for small enough $m, \ell$, the linear span $\calU_D$ of the restricted partials of $P$ is \emph{equal} to $\calV_D$.

\begin{lemma}[Analysis of the Subspace of Restricted Partials of $\wh{P}$]
\label{lem:analysis-restricted-partials-high-degree}
    Fix $D,K\in \N$ and let $n,m,\ell \in \N$ such that $m \leq \paren{\frac{\ell}{\polylog(\ell)}}^{KD/2}$ and $m\ell^{2(K-1)D} \leq \paren{\frac{n}{\polylog(n)}}^{2D}$.
    Given $\wh{P} = \sum_{t \in [m]} A_t(x)^{3D} + E(x)$ where each $A_t$ is a degree-$K$ homogeneous polynomial with i.i.d.\ $\calN(0,1)$ entries, and a restriction matrix $M \in \R^{n \times \ell}$, we have that with probability $1- n^{-\Omega(D)}$ over the choice of $A_t$'s, an analogue of  Algorithm~\ref{algo:partial-der-hd} outputs a subspace $\wt{\calV}_D$ of $\R^{\ell_{(3K-2)D}}$ that satisfies: 
    \[ \Norm{\wt{\calV}_D - \calV_D }_F \leq O\Paren{\frac{\|E\|_F}{n^D \ell^{KD/2}}} \mcom \]
    with $B_t(y) = A_t \circ M(y)$. 
\end{lemma}

The proof is almost identical to the proof of Lemma~\ref{lem:analysis-restricted-partials-hd}: we prove $\calU_D = \calV_D$ (Lemma~\ref{lem:U-equals-V-high-deg}) by proving that a different matrix $L$ (Definition~\ref{def:L-matrix-high-degree}) defining the linear system is full row rank (Lemma~\ref{lem:singular-value-of-L-higher-deg}), hence the linear system is feasible.
The only difference is the $L$ matrix since the structure of the partial derivatives is different.

\paragraph{Span finding step: comparison to the quadratic case.}
We would like to obtain an estimate of the subspace $\calW_D = \spn(B_t(y)^D \mid t\in[m])$ given a basis for $\wt{\calV}_D$.
We reemphasize that this is the main difference in the algorithm.
We choose a random polynomial $G$ of degree $2D(K-1)$ and write $\calV_G \coloneqq \spn(G(y) y_S \mid S\in[m]^{KD})$.
Note the difference between $\calV_D$ and $\calV_G$: they are both subspaces of polynomials of degree $(3K-2)D$, but $B_t(y)^D$ has degree $KD$ which is smaller than $2(K-1)D$ (for $K>2$).
Nevertheless, we show that we can still obtain $\calW_D$ by taking the intersection of $\calV_D$ and $\calV_G$ to obtain $\spn(B_t(y)^D G(y))$ (Lemma~\ref{lem:intersection-high-deg}) and follow the remaining steps in Algorithm~\ref{algo:span-finding-hd}.

\begin{lemma}[Extracting Span of $B_t(y)^D$]\label{lem:span-finding-high-degree}
    Fix $D,K\in \N$, $K\geq 2$ and let $n,m,\ell \in \N$ such that $m \leq \paren{\frac{\ell}{\polylog(\ell)}}^{KD/2}$.
    Given degree-$K$ homogeneous polynomials $B_t$ for $t\in[m]$ in $\ell$ variables with coefficients drawn i.i.d from $\calN(0,1)$, with probability $1- \ell^{-\Omega(D)}$ an analogue of Algorithm~\ref{algo:span-finding-hd} outputs $\wt{\calW}_D$ that satisfies:
    \[ \Norm{\wt{\calW}_D - \calW_D}_F \leq m\ell^{O(KD)} \|\calV_D - \wt{\calV}_D\|_F.\]
\end{lemma}

The remaining steps are identical to the quadratic case.

\begin{lemma}[Desymmetrization of Restricted $\wh{P}$ via Least-Squares]
\label{lem:analysis-desym-high-degree}
    Let $D, K,m,\ell\in \N$ such that $m \leq (\frac{\ell}{\polylog(\ell)})^{KD/2}$.
    For each $t\in[m]$, let $B_t$ be a degree-$K$ homogeneous polynomial in $\ell$ variables with i.i.d.\ $\calN(0,1)$ entries.
    Suppose $\wt{\calW}_D$ is a subspace of $\R^{\ell_{KD}}$ such that $\Norm{\wt{\calW}_D -\calW_D} \leq 1/(\poly(m)\ell^{O(KD)})$, then with probability $1- n^{\Omega(D)}$ over the choice of $B_t$'s, an analogue of
    Algorithm~\ref{algo:desymmetrization-hd} outputs a tensor $\wt{T}$ such that:
    \[ \Norm{\wt{T} - \sum_{t \in [m]} \Paren{\Sym(B_t^{\ot D})}^{\ot 3} }_F \leq \poly(m) \Paren{\ell^{O(KD)}\|\calW_D - \wt{\calW}_D\|_F + \|E\|_F} \mper\]
\end{lemma}

We show how to aggregate the desymmetrized estimates above for $n^{O(KD)}$ pseudorandom restriction matrices to obtain the estimate of the unrestricted tensor we need.
The following is identical to Lemma~\ref{lem:analysis-aggregation-hd} except for a small change in the parameters:

\begin{lemma}[Aggregating Pseudorandom Restrictions]
\label{lem:analysis-aggregation-high-degree}
    Let $D,K,n,\ell,m \in \N$ such that $6DK \leq \ell \leq n$.
    There is an $n^{O(KD)}$-time computable collection $\calS$ of subsets of $[n]$ such that each $S\in \calS$ satisfies $\ell \leq |S| \leq 6DK\ell$ and that
    \begin{equation*}
        \E_{S\sim \calS} \sum_{t=1}^m \Paren{\Sym\Paren{\calR_S(A_t)^{\ot D}}}^{\ot 3}
        = C \circ \sum_{t=1}^m \Paren{\Sym(A_t^{\ot D}) }^{\ot 3}
    \end{equation*}
    where $C \in (\R^n)^{\ot 6KD}$ is a fixed tensor whose entries depend only on the entry locations, and each entry of $C$ has value within $((\ell/2n)^{6KD}, 1)$.
\end{lemma}

Given the partially desymmetrized tensor, we will apply the off-the-shelf algorithm for 3rd order tensor decomposition (Fact~\ref{fact:tensor-decomposition-algo}) to obtain $\Sym(A_t^{\ot D})$ for $t\leq m$.
To apply tensor decomposition, we need the following condition number bound:

\begin{lemma}[Condition number]
    Under the same assumptions as Lemma~\ref{lem:analysis-restricted-partials-high-degree},
    let $A_D$ be the $n_{DK} \times m$ matrix whose columns are the coefficient vectors of $A_t(x)^D$ for $t\in[m]$.
    Then, with probability $1 - n^{-\Omega(D)}$, the condition number $\kappa(A_D) \leq O(1)$.
\end{lemma}

Finally, we extract $A_t$ from $\Sym(A_t^{\ot D})$ using the same procedure as Lemma~\ref{lem:analysis-D-root-hd}:

\begin{lemma}[Stable Computation of $D$-th Roots]
    Let $K, D, n \in \N$ and $\delta \geq 0$.
    Let $P \in (\R^n)^{\ot K}$ be an unknown symmetric tensor.
    Suppose $\wt{P_D}(x)$ is a homogeneous degree-$KD$ polynomial in $n$ variables such that its coefficient tensor satisfies $\Norm{\wt{P_D} - \Sym(P^{\otimes D})}_F \leq \delta$.
    There is an algorithm that runs in $n^{O(KD)}$ time and outputs $\wt{Q} \in (\R^n)^{\ot K}$ such that if $D$ is odd, then
    \begin{equation*}
        \Norm{\wt{Q} - P}_F \leq n^{O(K)} \delta^{1/D} \mcom
    \end{equation*}
    and if $D$ is even, then
    \begin{equation*}
        \min_{\sigma \in \{\pm1\}} \Norm{ \wt{Q} - \sigma P }_{F} \leq  n^{O(K)} \delta^{1/3D} \cdot \|P\|_{\max} \mper
    \end{equation*}
\end{lemma}

These sequence of steps and the analyses captured in the above lemmas are enough to immediately complete the proof of Theorem~\ref{thm:main-theorem-high-degree}.

\subsection{Analysis of the partial derivative span}
\label{sec:U-equals-V-high-deg}

It is clear that $\calU_D \subseteq \calV_D$.
We will prove that the two subspaces are in fact equal.
The following lemma is almost identical to Lemma~\ref{lem:U-equals-V-hd}.

\begin{lemma}[$\calU_D = \calV_D$]
\label{lem:U-equals-V-high-deg}
    Fix $K,D\in \N$. Let $\ell, m, n\in \N$ such that $m\ell^{2(K-1)D} \leq (\frac{n}{\polylog(n)})^{2D}$.
    Let $A_1,\dots,A_m \in (\R^n)^{\ot K}$ be random symmetric tensors with independent Gaussian entries, let $M \in \R^{n\times \ell}$ be a fixed restriction matrix.
    Then, with probability $1 - n^{-\Omega(D)}$ over the draw of $A_t$s, we have that $\calU_D = \calV_D$ and further, for any degree-$2(K-1)D$ homogeneous polynomials $Q_1,\dots, Q_m$, there exist coefficients $R = \{R_I\}_{I\in[n]^{2D}}$ such that
    \begin{equation*}
        \sum_{I\subseteq[n],\ |I|=2D} R_I \cdot \sum_{t=1}^m (\partial_{I} A_t^{3D}) \circ M(y) = \sum_{t=1}^m B_t(y)^D Q_t(y) \mcom
    \end{equation*}
    and moreover, $\|R\|_F \leq O(n^{-D}) \sum_{t=1}^m \|Q_t\|_F$.
\end{lemma}

In this section, we proceed to prove the main lemma for estimating the restricted partial span of powers of high-degree polynomials.
Recall that in Section~\ref{sec:U-equals-V-hd}, we reduce the task of proving $\calU_D = \calV_D$ to proving the feasibility of a linear system characterized by the matrix $L$ in Definition~\ref{def:L-matrix-hd}.
We will follow the same steps and get an analogous linear system.

Recall that in the case of powers of quadratic polynomials, the partial derivatives are nicely characterized as a sum of products of linear polynomial of $\langle A_t[i], x\rangle $ and scalars of the form of $A_t[i_1,i_2]$. This, not surprisingly, is no longer true when we go to higher degree polynomials; that being said, we can still get a somewhat ``nice'' characterization of the coefficient vectors that arise in the partial derivatives. Towards this end, we first need to generalize our definition of bucket profile from Definition~\ref{def:bucket-profile}:

\begin{definition}[Generalized bucket profile] \label{def:generalized-bucket-profile}
For any given $K\geq 2$,
We call $\gam = [\gam_1, \dots, \gam_{K} ]$ a \emph{generalized} bucket profile of the $2D$ partial derivative such that \[ 
\sum_{i=1}^K \gam_i\cdot i = 2D \mper
\] 
Let $\Gamma_{2D}$ be the set of bucket profiles, and let $\gamma^* = (2D,0,\dots,0)$ be the \emph{dominant} bucket profile.
Furthermore, we define $\deg(\gam) \coloneqq \sum_{i=1}^{K} (K-i) \gamma_i = K(\sum_{i=1}^K \gamma_i) - 2D$.

For each bucket profile, we associate with a \emph{bucket partition} $(\kappa_1(\gamma),\dots, \kappa_K(\gamma))$ of $\{1, 2, \dots, 2D\}$
such that there are $\gamma_i$ partitions (ordered tuples) of size $i$, sorted in \emph{increasing} order of the partition sizes, and $\kappa_i(\gam)$ is the collection of the $\gam_i$ partitions of size $i$:
\begin{equation*}
\begin{aligned}
    \kappa_1(\gamma) &= \bigparen{ (1),(2),\dots, (\gamma_1) }
    \mcom \quad
    \kappa_2(\gamma) = \bigparen{ (\gamma_1+1, \gamma_1+2), \dots, (\gamma_1 + 2\gamma_2-1, \gamma_1 + 2\gamma_2) } \mcom
\end{aligned}
\end{equation*}
and so on.
Furthermore, we define $(\wt{\kappa}_1(\gamma),\dots,\wt{\kappa}_K(\gamma))$ to be a partition of $\{1,2,\dots, \deg(\gamma)\}$
such that there are $\gamma_i$ partitions (ordered tuples) of size $K-i$, sorted in \emph{decreasing} order of the partition sizes, and $\wt{\kappa}_i(\gam)$ is the collection of the $\gamma_i$ partitions of size $K-i$:
\begin{equation*}
\begin{aligned}
    \wt{\kappa}_1(\gamma) &= \Bigparen{ (1,2,\dots,K-1), \dots, ((\gamma_1-1)(K-1)+1, \dots, \gamma_1(K-1)) } \\
    \wt{\kappa}_2(\gamma) &= \Big( (\gamma_1(K-1) +1, \dots, \gamma_1(K-1) + (K-2) ), \dots, \\
    & \quad \quad (\gamma_1(K-1) + (\gamma_2-1)(K-2)+1, \dots, \gamma_1(K-1) + \gamma_2(K-2)) \Big)
\end{aligned}
\end{equation*}
and so on.
Let $\kappa_i(\gamma) = (S_1,\dots,S_{\gamma_i} )$ and $\wt{\kappa}_i(\gamma) = (T_1,\dots, T_{\gamma_i})$ where $|S_j| + |T_j| = K$ for all $j\in[\gamma_i]$, and we define $\kappa_i \| \wt{\kappa}_i(\gamma)$ to be the collection of tuples $((S_1,T_1), \dots, (S_{\gamma_i}, T_{\gamma_i}))$.

\end{definition}

\begin{remark}[Interpretation of $\gamma$ and the terms in partial derivatives of $A_t^{3D}$]
\label{rem:bucket-profile-interpretation}
    To better understand Definition~\ref{def:generalized-bucket-profile}, we make the following remarks,
    \begin{itemize}
        \item Imagine $A_t^{3D}$ as being $3D$ buckets, and taking $2D$ partial derivatives means dropping $2D$ balls in these buckets such that each bucket contains at most $K$ balls.
        Then, for $i\in[K]$, $\gamma_i$ denotes the number of buckets with $i$ balls, hence $\sum_{i=1}^K \gam_i \cdot i = 2D$.
        Having $i$ balls corresponds to $i$-th derivatives of $A_t$, which is degree-$(K-i)$ polynomials.

        \item Let $I = (i_1,\dots,i_{2D}) \in [n]^{2D}$ be an ordered tuple, and consider $\partial_I A_t^{3D}$ and a bucket profile $\gamma \in \Gamma_{2D}$.
        Note that there are $3D - \sum_{j=1}^K \gamma_j$ empty buckets.
        Thus, $\gamma$ represents the terms in $\partial_I A_t^{3D}$ that are products of $A_t^{3D-\sum_j \gamma_j}$ and $\gamma_j$ polynomials of degree $K-j$.
        These $\gamma_j$ degree-$(K-j)$ polynomials can be represented as $\Iprod{A_t[I_{\pi(S)}], x^{\ot K-j}}$ for $S\in \kappa_j(\gam)$ where $|S| = j$ and for some permutation $\pi\in \bbS_{2D}$.
        Since $|I_{\pi(S)}| = |S| = j$, $A_t[I_{\pi(S)}]$ is an order-$(K-j)$ slice of the tensor $A_t\in (\R^n)^{\ot K}$, i.e.\ $A_t[I_{\pi(S)}]\in (\R^n)^{\ot K-j}$.

        \item
        The product over $j\in[K]$ and $S\in \kappa_j(\gamma)$ (where $|\kappa_j(\gam)| = \gam_j$) can further be written as the product between the tensor $\bigotimes_{j\in[K]} \bigotimes_{S\in\kappa_j(\gam)} A_t[I_{\pi(S)}]$, which is a tensor of order $\sum_{j=1}^{K}(K-j)\gamma_j = \deg(\gamma)$, and the tensor $x^{\ot \deg(\gam)}$.
        Therefore, the term in $\partial_I A_t^{3D}$ corresponding to $\gamma$ can be written as $A_t^{3D-\sum_{j=1}^K \gamma_j} \cdot p_{t,I,\gamma}(x)$ for some polynomial $p_{t,I,\gamma}$ of degree $\deg(\gam)$.
    \end{itemize}
\end{remark}

With Definition~\ref{def:generalized-bucket-profile} and the discussion in Remark~\ref{rem:bucket-profile-interpretation} in mind, we now write out the partial derivatives for $I = \{i_1,\dots, i_{2D}\} \in [n]^{2D}$ as follows,
\begin{equation*}
\begin{aligned}
    \partial_I A_t^{3D}(x) &= \sum_{\gam\in \Gamma_{2D}} \sum_{\pi\in \bbS_{2D}} c(\gam,K) A_t(x)^{3D - \sum_{j=1}^K \gamma_j } \prod_{j \in [K] } \prod_{S\in \kappa_j(\gam)} \langle A_t[I_{\pi(S)}] , x^{\otimes K-j} \rangle \\
    &= \sum_{\gam\in \Gamma_{2D}}c(\gam,K) A_t(x)^{3D - \sum_{j=1}^K \gamma_j } 
     \sum_{\pi\in \bbS_{2D}} \Iprod{ \bigotimes_{j \in [K] } \bigotimes_{S\in \kappa_j(\gam)} A_t[I_{\pi(S)}], x^{\ot \sum_{j=1}^K (K-j) \gamma_j }} \\
    &= \sum_{\gamma \in \Gamma_{2D}} A_t(x)^{3D - \sum_{j=1}^K \gamma_j} \cdot  p_{t,I,\gamma}(x)
\end{aligned}
\end{equation*}
where $c(\gamma,K)$ is a constant and $\deg(p_{t,I,\gamma}) = \sum_{j=1}^K (K-j) \gamma_j = \deg(\gamma)$.

We remark that the exponent $3D - \sum_{j=1}^K \gamma_j$ of $A_t(x)$ is minimized at $D$ when $\gamma = \gamma^* = (2D, 0, \dots, 0)$ and we have $\deg(p_{t,I,\gamma^*}) = 2D\cdot (K-1)$.
Note also that the total degree of $A_t^{3D}$ is $3KD$, hence after taking $2D$ derivatives, the degree of $\partial_I A_t^{3D}$ is $(3K-2)D$.

Next, we set $x = My$ where $M$ is the given restriction matrix, and let $q_{t,I,\gamma}(y) = p_{t,I,\gamma}(My)$.
For a given bucket profile $\gam$, the polynomial $q_{t,I,\gamma}(y)$ is a degree $\deg(\gamma)$ polynomial in variables $y$.
Similar to the strategy in Section~\ref{sec:U-equals-V-hd}, we define the linear system with the focus on $\gamma = \gamma^*$:

\begin{definition}[Linear system for $\calU_D = \calV_D$ for degree-$K$ polynomials]
\label{def:linear-system-for-U-equals-V-high-deg}
    Given arbitrary degree-$(K-1)D$ polynomials $Q_1,\dots,Q_m$, we define the following linear system in variables $\{R_I\}_{I\in[n]^{2D}}$:
    Given arbitrary degree-$2(K-1)D$ homogeneous polynomials $Q_1,\dots,Q_m$, we define the following linear system in variables $\{R_I\}_{I\in [n]^{2D}}$:
    \begin{equation*}
    \begin{aligned}
        &\sum_{I\in[n]^{2D}} R_I \cdot q_{t,I,\gamma^*}(y) = Q_t(y), & & \forall t\in [m], \\
        &\sum_{I\in[n]^{2D}} R_I \cdot q_{t,I,\gamma}(y) = 0, & & \forall \gamma \neq \gamma^*,\ t\in[m] \mper
    \end{aligned}
    \end{equation*}
\end{definition}

Note that $q_{t,I,\gamma}$ is of degree $\deg(\gamma)$, so each $\gamma \in \Gamma_{2D}$ and each multiset $J\in [\ell]^{\deg(\gamma)}$ gives a linear constraint in the linear system above.
Thus, for a $J\in [\ell]^{\deg(\gamma)}$, we would like to determine the coefficient of $y_J$ in $q_{t,I,\gamma}(y)$.
For convenience, we let $J' \in [n]^{\deg(\gamma)}$ be the multiset in $[n]$ if we map $J$ back to $[n]$.

\begin{equation*}
\begin{aligned}
    \wh{q}_{t,I,\gamma}(J) &= c'(\gamma,K) \cdot \sum_{\pi \in \bbS_{2D}} \sum_{\pi' \in \bbS_{\deg(\gamma)}}
        \prod_{i\in [K]} \prod_{S,T \in \kappa_i \| \wt{\kappa}_i(\gamma)} A_t[I_{\pi(S)}, J'_{\pi'(T)}]  \\
    &= c'(\gamma,K) \cdot \sum_{\pi \in \bbS_{2D}} \sum_{\pi' \in \bbS_{\deg(\gamma)}}
        \Paren{ \bigotimes_{i\in[K]} \bigotimes_{T\in \wt{\kappa}_i(\gamma)} A_t[J'_{\pi'(T)}] }[\pi(I)] \\
    &= c''(\gamma,K) \cdot \Sym \Paren{ \sum_{\pi' \in \bbS_{\deg(\gamma)}} \bigotimes_{i\in[K]} \bigotimes_{T\in \wt{\kappa}_i(\gamma)} A_t[J'_{\pi'(T)}] }[I] \mper
\end{aligned}
\end{equation*}

Let us parse through the expression above.
First, $|I| = 2D$ by definition and $|J| = \deg(\gamma)$ since this is the degree of $q_{t,I,\gamma}(y)$.
We recall Definition~\ref{def:generalized-bucket-profile}: for an $i\in [K]$, $\kappa_i \| \wt{\kappa}_i(\gam)$ denotes $\paren{(S_1,T_1),\dots, (S_{\gam_i}, T_{\gam_i})}$ where each $S, T \in \kappa_i \| \wt{\kappa}_i(\gamma)$ satisfy $|S| = i$ and $|T| = K-i$.
Thus, $|I_{\pi(S)}| = i$ and $|J'_{\pi(T)}| = K-i$, and together they index a single entry $A_t[I_{\pi(S)}, J'_{\pi'(T)}]$ of the order-$K$ tensor $A_t$.
Furthermore, $A_t[J'_{\pi'(T)}]$ is an order-$i$ slice of the tensor $A_t$, and thus taking tensor products over $i\in[K]$ and $T\in \wt{\kappa}_i(\gam)$, we get a tensor of order $\sum_{i=1}^K \gamma_i \cdot i = 2D$ as desired.

At a high level, we have written $\wh{q}_{t,I,\gamma}(J)$, which is the coefficient of monomial $y_J$ in $q_{t,I,\gamma}(y)$, as an entry of an order-$2D$ symmetric tensor indexed by $I \in [n]^{2D}$.
Now, we can rewrite $\sum_{I} R_I \cdot \wh{q}_{t,I,\gamma}(J)$ as an inner product between the tensor $R$ and an order-$2D$ symmetric tensor determined by $t, \gamma$ and $J$.
This leads us to the following matrix,

\begin{definition} \label{def:L-matrix-high-degree}
    For a fixed $K\geq 2$, for each $\gam$ in $\Gamma_{2D}$, let $L_\gam$ be the $m\ell_{\deg(\gam)} \times n^{2D}$ matrix  whose each row is indexed by $t\in [m]$ and a multiset $J \in [\ell]^{\deg(\gam)}$:
    \begin{equation*}
        L_{\gamma}[(t,J), \cdot] = \fl \Paren{ \Sym\Paren{ \sum_{\pi\in \bbS_{\deg(\gam)}} \bigotimes_{i\in [K]} \bigotimes_{T\in \wt{\kappa}_i(\gamma)} A_t[J'_{\pi(T)}] }}
    \end{equation*}
    i.e.\ each row is a flattened vector of a symmetric order $2D$ tensor of dimension $n$.
    Moreover, let $r_D = m \sum_{\gamma\in \Gamma_{2D}} \ell_{\deg(\gamma)}$, and define $L$ to be the $r_D \times n^{2D}$ matrix formed by concatenating the rows of $L_{\gamma}$ for $\gamma \in \Gamma_{2D}$.
\end{definition}

Now, we prove a singular value lower bound for $L$. We defer the proof to Section~\ref{sec:sval-U-equals-V}.

\begin{lemma}[Singular value lower bound for $L$ for degree-$K$ polynomials]
\label{lem:singular-value-of-L-higher-deg}
      For a fixed $K$, $D\in \N$, let $m,\ell,n\in \N$ such that $m \ell^{2(K-1)D} \leq (\frac{n}{\polylog(n)})^{2D}$.
    Let $L \in \R^{r_D \times n^{2D}}$ be the matrix defined in Definition~\ref{def:L-matrix-high-degree}.
    Then, with probability $1 - n^{-\Omega(D)}$, $\sigma_{r_D}(L) \geq \Omega(n^{D})$.
\end{lemma}

The proof of Lemma~\ref{lem:U-equals-V-high-deg} follows immediately from Lemma~\ref{lem:singular-value-of-L-higher-deg}.

\subsection{Analysis for \texorpdfstring{$\calV_D$}{V\_D} and the span finding step}
\label{sec:analysis-of-V-high-deg}

Recall that we define
\begin{equation*}
    \calV_D = \{ B_t(y)^D y_T \mid t\in [m], T\in[\ell]^{2(K-1)D}\} \mper
\end{equation*}

In the span finding step, we will sample a random polynomial $G(y)$ of degree $2(K-1)D$.
We need to analyze what sets of polynomials $(p_1,\dots, p_m, q)$ satisfy
\begin{equation}
\label{eq:null-space-high-deg}
    \sum_{t=1}^m B_t(y)^D p_t(y) + G(y) q(y) = 0 \mcom
\end{equation}
where $\deg(p_t) = 2(K-1)D$ and $\deg(q) = KD$.
Thus, we define the following, which is the original $V$ matrix with an extra block.

\begin{definition}[Matrix $\ol{V}$]
\label{def:V-bar-matrix-high-deg}
    Let $G \in (\R^\ell)^{\ot 2(K-1)D}$ be a tensor with i.i.d.\ Gaussian entries.
    We define $\ol{V}$ to be the $\ell_{(3K-2)D} \times (m\ell_{2(K-1)D} + \ell_{KD})$ matrix of the form
    \begin{equation*}
        \ol{V} = \begin{bmatrix}
            V \mid V_G
        \end{bmatrix}
    \end{equation*}
    where $V_G$ is a $\ell_{(3K-2)D} \times \ell_{KD}$ matrix whose columns are coefficient vectors of the polynomials $\ell^{-\frac{(K-2)D}{2}} G(y) y_S$ for $S \in [\ell]^{KD}$.
    Note that the scaling of $\ell^{-\frac{(K-2)D}{2}} G(y) y_S$ is to ensure that the columns of $V$ and $V_G$ have roughly the same norms.
\end{definition}

Observe that in \pref{eq:null-space-high-deg}, the degree of $B_t(y)^D$ is $KD$ but the degree of $p_t(y)$ is $2(K-1)D$ which is larger than $KD$ for $K > 2$.
Therefore, clearly the dimension of the null space of $V$ is much larger than $\binom{m}{2}$ (which is the case for $K=2$).
Furthermore, we have the extra block $V_G$ which gives additional null space.
We define the following sets of $(p_1,\dots,p_m, q)$ which are the ``obvious'' solutions to \pref{eq:null-space-high-deg}.

\begin{definition}[Null space of $\ol{V}$]
    Given degree-$K$ homogeneous polynomials $B_1,\dots, B_m$ and a degree-$2(K-1)D$ polynomial $G$, we define $\calN_1$ to be the following set of $\binom{m}{2} \ell_{(K-2)D}$ tuples of polynomials $(p_1,\dots,p_m, q)$:
    \begin{equation*}
        \calN_1 = \Set{ (p_1,\dots, p_m, q) \mid p_{t_1} = B_{t_2}(y)^D y_T,\ p_{t_2}= -B_{t_1}(y)^D y_T,\
        p_{s} = 0,\
        q=0 \textnormal{ for } s \neq t_1,t_2 }_{\substack{t_1 < t_2 \leq m,\\ T \in [\ell]^{(K-2)D}} }
    \end{equation*}
    and $\calN_2$ to be the following set of $m$ tuples:
    \begin{equation*}
        \calN_2 = \Set{ (p_1,\dots, p_m, q) \mid p_t = G(y), p_{s} = 0, \text{ and } q = B_t(y)^D \text{ for } s\neq t }_{t\in [m]} \mper
    \end{equation*}
    Here each $p_t$ is degree $2(K-1)D$ and $q$ is degree $KD$.
\end{definition}

We next show that $\calN_1, \calN_2$ are in fact the \emph{only} solutions to \pref{eq:null-space-high-deg}.

\begin{lemma}[Singular value lower bound for $\ol{V}$]
\label{lem:V-singular-value-high-degree}
    Fix $K, D\in \N$ and $K\geq 3$.
    Let $m,\ell \in \N$ such that $m \leq (\frac{\ell}{\polylog(\ell)})^{KD/2}$.
    Let $B_1, \dots, B_m$ be degree-$K$ homogeneous polynomials and let $G$ be a degree-$2(K-1)D$ homogeneous polynomial, all in $\ell$ variables and have i.i.d.\ standard Gaussian coefficients.
    Then, with probability $1 - \ell^{-\Omega(D)}$, the set of polynomials $(p_1,\dots, p_m, q)$ with $\deg(p_t) = 2(K-1)D$ and $\deg(q) = KD$ satisfying
    \begin{equation*}
        \sum_{t=1}^m B_t(y)^D p_t(y) + G(y) q(y) = 0
    \end{equation*}
    forms a subspace of dimension $\binom{m}{2}\ell_{(K-2)D} + m$ spanned by $\calN_1 \cup \calN_2$.
    
    Furthermore, the matrix $\ol{V}$ defined in Definition~\ref{def:V-bar-matrix-high-deg} satisfies $\rank(\ol{V}) =  m\ell_{2(K-1)D} - |\calN_1| - |\calN_2|$ and $\sigma_{\rank(\ol{V})}(\ol{V}) \geq \Omega(\ell^{KD/2})$.
\end{lemma}

To prove Lemma~\ref{lem:V-singular-value-high-degree}, we first define the following matrices analogous to Definition~\ref{def:N-matrix} which represent $\calN_1$ and $\calN_2$:

\begin{definition}[Matrix $N$]
\label{def:null-space-of-V-high-deg}
    We define $N$ to be the $\binom{m}{2} \ell_{(K-2)D} \times m \ell_{2(K-1)D}$ matrix whose rows are indexed by $(t_1, t_2, T)$ for $t_1 < t_2 \in [m]$ and $T \in [\ell]^{(K-2)D}$, and each row represents a collection of $m$ degree-$2(K-1)D$ polynomials $(p_1,\dots, p_m)$ such that
    \begin{equation*}
        p_s(y) = 
        \begin{cases}
            B_{t_2}(y)^D y_T & \text{if } s = t_1 \\
            -B_{t_1}(y)^D y_T & \text{if } s = t_2 \\
            0 & \text{otherwise.}
        \end{cases}
    \end{equation*}
\end{definition}

\begin{definition}[Matrix $\ol{N}$]
\label{def:N-bar-matrix-high-deg}
    Let $G \in (\R^\ell)^{\ot 2(K-1)D}$ be a tensor with i.i.d.\ Gaussian entries.
    We define $\ol{N}$ to be the $(\binom{m}{2} \ell_{(K-2)D} + m) \times (m \ell_{2(K-1)D} + \ell_{KD})$ of the form
    \begin{equation*}
        \ol{N} = \begin{bmatrix}
            N & 0 \\
            N_G & N_B
        \end{bmatrix}
    \end{equation*}
    where $N$ is $\binom{m}{2} \ell_{(K-2)D} \times m \ell_{2(K-1)D}$ defined in Definition~\ref{def:N-bar-matrix-high-deg}, $N_G$ is $m\times m\ell_{2(K-1)D}$ such that for $s\in [m]$, $t\in [m]$ and $I \in [\ell]^{2(K-1)D}$,
    \begin{equation*}
        N_G[s, (t,I)] = \begin{cases}
            \ell^{-\frac{(K-2)D}{2}} G[I] & \text{if } s=t, \\
            0    & \text{otherwise,}
        \end{cases}
    \end{equation*}
    and $N_B$ is an $m \times \ell_{KD}$ matrix such that for $s\in[m]$, $J \in [\ell]^{KD}$,
    \begin{equation*}
        N_B[s, J] = -(B_s(y)^D)[J] \mper
    \end{equation*}
\end{definition}

Observe that the rows of $[N\mid 0]$ exactly represent $\calN_1$ and the rows of $[N_G \mid N_B]$ represent $\calN_2$.
Next, we prove the following lemma analogous to Lemma~\ref{lem:rank-of-N}:

\begin{lemma}[Rank of $\ol{N}$]
\label{lem:rank-of-N-high-deg}
    Let $K,D,m,\ell\in\N$ be the same parameters as Lemma~\ref{lem:V-singular-value-high-degree}.
    Let $\ol{N}$ be the $(\binom{m}{2} \ell_{(K-2)D} + m) \times (m \ell_{2(K-1)D} + \ell_{KD})$ matrix defined in Definition~\ref{def:N-bar-matrix-high-deg}.
    Then, with probability $1 - \ell^{-\Omega(D)}$,
    \begin{equation*}
        \rank(\ol{N}) = \binom{m}{2} \ell_{(K-2)D} + m,\quad 
        \sigma_{\rank(\ol{N})}(\ol{N}) \geq \Omega(\ell^{KD/2}) \mper
    \end{equation*}
\end{lemma}

Finally, we show that the row span of $\ol{N}$ \emph{equals} the null space of $\ol{V}$.
Specifically, we prove the following lemma analogous to Lemma~\ref{lem:VV-NN}:

\begin{lemma}
\label{lem:VV-NN-high-deg}
    Let $K,D,m,\ell\in\N$ be the same parameters as Lemma~\ref{lem:V-singular-value-high-degree}.
    Let $\ol{V}$, $\ol{N}$ be the matrices defined in Definition~\ref{def:V-bar-matrix-high-deg} and \ref{def:N-bar-matrix-high-deg}.
    Then with probability $1-\ell^{-\Omega(D)}$,
    \begin{equation*}
        \lambda_{\min}\Paren{\ol{V}^\top \ol{V} + \ol{N}^\top \ol{N}}
        \geq \Omega(\ell^{KD}) \mper
    \end{equation*}
\end{lemma}

We defer the proofs to Section~\ref{sec:sval-analysis-of-V-high-deg}.
Combining the above lemmas, we can complete the proof of Lemma~\ref{lem:V-singular-value-high-degree}, which is almost identical to the proof of Lemma~\ref{lem:V-singular-value-hd}.

\begin{proof}[Proof of Lemma~\ref{lem:V-singular-value-high-degree}]
    By the definition of $\ol{V}$ and $\ol{N}$, $\ol{V}\cdot \ol{N}^\top = 0$ because each row of $\ol{N}$ represents a solution to
    \begin{equation*}
        \sum_{t=1}^m B_t(y)^D p_t(y) + G(y) q(y) = 0 \mper
    \end{equation*}
    This implies that the matrices $\ol{V}^\top \ol{V}$ and $\ol{N}^\top \ol{N}$ have orthogonal column span.
    By Lemma~\ref{lem:rank-of-N-high-deg}, $\rank(\ol{N}) = \binom{m}{2} \ell_{(K-2)D} + m$, which is exactly $|\calN_1| + |\calN_2|$, and by Lemma~\ref{lem:VV-NN-high-deg}, the row span of $\ol{N}$ is exactly the null space of $\ol{V}^\top \ol{V}$.
    This implies that $\rank(\ol{V}) = m\ell_{2(K-1)D} - |\calN_1| - |\calN_2|$ and that $\spn(\calN_1 \cup \calN_2)$ is exactly the set of solutions to $\sum_{t=1}^m B_t(y)^D p_t(y) + G(y) q(y) = 0$.
    
    $\ol{V}^\top \ol{V} + \ol{N}^\top \ol{N}$ having smallest eigenvalue at least $\Omega(\ell^{KD})$ further shows that the smallest singular value of $\ol{V}$ is lower bounded by $\Omega(\ell^{KD/2})$. This completes the proof.
\end{proof}

\paragraph{Span finding step.}
With Lemma~\ref{lem:V-singular-value-high-degree}, we can now analyze our span finding step.
For the random polynomial $G$, let
\begin{equation*}
\begin{aligned}
    \calV_G &\coloneqq \spn\Paren{ G(y) y_S \mid S\in [\ell]^{KD}} \mcom \\
    \calW_G &\coloneqq \spn\Paren{ G(y) B_t(y)^D \mid t\in[m] } \mper
\end{aligned} 
\end{equation*}

It is easy to see that the subspace $\calW_G$ lies inside the intersection of $\calV_D$ and $\calV_G$.
Same as the proof of Lemma~\ref{lem:span-finding-hd}, we will show that in fact the intersection is \emph{equal} to $\calW_G$. The following lemma is analogous to Lemma~\ref{lem:intersection-hd}

\begin{lemma}\label{lem:intersection-high-deg}
    For degree-$K$ homogeneous polynomials $B_t(y), t\in [m]$ and degree-$2(K-1)D$ polynomial $G(y)$ whose coefficients are chosen independently at random from $\calN(0,1)$, we have that with probability $1$ over the draw of $A_t$s and $G$, we have:
    \[\calV_D \cap \calV_G = \calW_G \mcom\]
    with $\dim(\calW_p) = m$.
\end{lemma}
The proof of Lemma~\ref{lem:intersection-high-deg} via Lemma~\ref{lem:V-singular-value-high-degree} is identical to the proof of Lemma~\ref{lem:intersection-hd}. Given $\calW_G$, we can then divide out the polynomial $G(y)$ and get the subspace $\spn(B_t(y)^D \mid t\in[m])$. Analogous to Section~\ref{sec:span-finding-hd} we can make the above algorithm error resilient by taking a robust intersection of the subspaces and then performing robust division.

\section{Singular Value Lower Bounds for Structured Random Matrices}

\subsection{Detailed overview}
\label{sec:overview-singular-val}
\paragraph{A two-line strategy}
Before we delve into the ``real'' challenge of what might come up in our analysis, for the sake of exposition, let's recall our G.O.E.\ toy example $A\in \R^{m\times n}$ for $m\ll n$. To obtain a singular value lower bound for $A$ when $m\ll n$, we notice it can be reduced to showing\[ 
A^\top A \approx (1+o(1)) n\cdot \Id \mcom
\]
which amounts to checking 
\begin{itemize}
     \item Large diagonal: $\diag(A^\top A) = (1+o(1)) n$;
    \item Small off-diagonal: $\spec{\text{off-diagonal}(A^\top A)}\leq o(n)$.
\end{itemize}

For the above example, one may find both components rather immediate as the large diagonal follows via standard Gaussian concentration. While establishing a small spectral norm bound for the off-diagonal block might require slightly more work, there are several techniques to establish the nearly optimal bound of $O(\sqrt{m}\sqrt{n} )$ on its spectral norm. That said, we may hope to carry out the above strategy on potentially more complicated matrices (with correlated entries), and let's dig into one example that arises in our analysis for the simplest setting in our main result for decomposing cubics of quadratic.

\paragraph{Example of structured random matrices} 
To describe the matrix, let $A_1, \dots, A_m$ be random matrices of size $n \times n$ with each entry i.i.d.\ $\cN(0,1)$. We will describe a new matrix whose entries are functions of $A_i$s. Let $A_t[i]\in \R^n$ denote the $i$-th column of matrix $A_t$, and let $S\subseteq [n]$ be a subset of size $\ell \approx \sqrt{n}$. Finally, let $G$ be the matrix defined by: \begin{align*}
    G \coloneqq \begin{bmatrix}\dots \\
    A_t[i]\otimes A_t[j]  \\
    \dots
    \end{bmatrix} \in  \R^{m\binom{\ell}{2} \times n^2} \mper
\end{align*}
In other words, each row of $G$ is indexed by $(t,i,j)$ for $t\in [m]$ and $i\neq j \in [\ell]$ with the corresponding row vector be $A_t[i]\otimes A_t[j]$. And a question we would like to answer is the following, 
\begin{question}
   What is the largest $m$ for which $G$ is full (row) rank with inverse-polynomially bounded smallest singular value?
\end{question}
For starters,  $m \sim n$ is a natural upper bound as $G$ becomes a roughly square matrix when we have $m\ell^2 \approx n^2$. Let's now implement the above strategy. The diagonal entry $(t,i,j)$ of $GG^\top$ is given by $\|A_t[i]\otimes A_t[j]\|_2 \approx n^2 $ by standard Gaussian concentration, and the two-step strategy now prompts us to bound the spectral norm of the off-diagonal part of $GG^\top$, i.e.\ the spectral norm of $GG^\top - n^2\cdot \Id_{m\binom{\ell}{2}}$, by $o(n^2)$.

Observe that easy estimates such as the maximum $\ell_1$ norm of any row do not give a useful bound: we expect the off-diagonal entries to have typical magnitude $n$ and that can be ``charged'' to the diagonal entries only if $m \binom{\ell}{2}\cdot n \lesssim n^2$ or $m \leq O(1)$. However, in this very case, one may notice that we are ignoring the potential cancellation from the signs of the entries (as each entry is a mean $0$ random variable) and we can hope to obtain a tighter bound by exploiting this kind of cancellation. Towards this end,  we appeal to the \emph{trace moment method} to obtain tight bounds for these structured random matrices, and on a high level, for ease of our spectral analysis, we adopt \emph{graph matrix decomposition} to systematically represent such random matrices of correlated entries.

\paragraph{Graphical matrix decomposition and norm bounds}
The theory of graphical matrices generally applies to a matrix whose entry is a polynomial of some underlying input. In our case, the off-digonal part of $GG^\top $ is described as:
\begin{align*}
    GG^\top [(s,a,b), (t,c,d)] &= \Iprod{A_s[a], A_t[c]} \cdot \Iprod{A_s[b], A_t[d]}
    = \sum_{i,j\in [n]} A_s[i,a]A_t[i,c]A_s[j,b]A_t[j,d]
\end{align*}
for $(s,a,b) \neq (t,c,d)$. For illustration, we restrict to the case for $s\neq t\in [m]$ and $a,b,c,d$ distinct elements in $[\ell]$, and we remark that in general these cases warrant a careful analysis. Notice we can further decompose the above polynomial entry as \begin{align*}
 GG^\top [(s,a,b), (t,c,d)] & = \sum_{i\neq j\in [n] } A_s[i,a]A_t[i,c]A_s[j,b]A_t[j,d]\\
 &+ \sum_{i=j\in [n] } A_s[i,a]A_t[i,c]A_s[j,b]A_t[j,d]    
\end{align*}
depending on whether $i,j$ ``collide'' with each other. Pictorially, each term in the above matrix corresponds to one of the following diagrams using graph matrix language:
\begin{center}
    \emph{each vertex in the left/right orange oval  corresponds to a variable in the row/column index for the matrix;\\ and each vertex in the middle (outside the orange ovals) corresponds to an indeterminate in the summation with each hyperedge corresponding to a random variable in the polynomial.}
\end{center}
\begin{figure}[ht!]
     \centering
    \includegraphics[width=160pt]{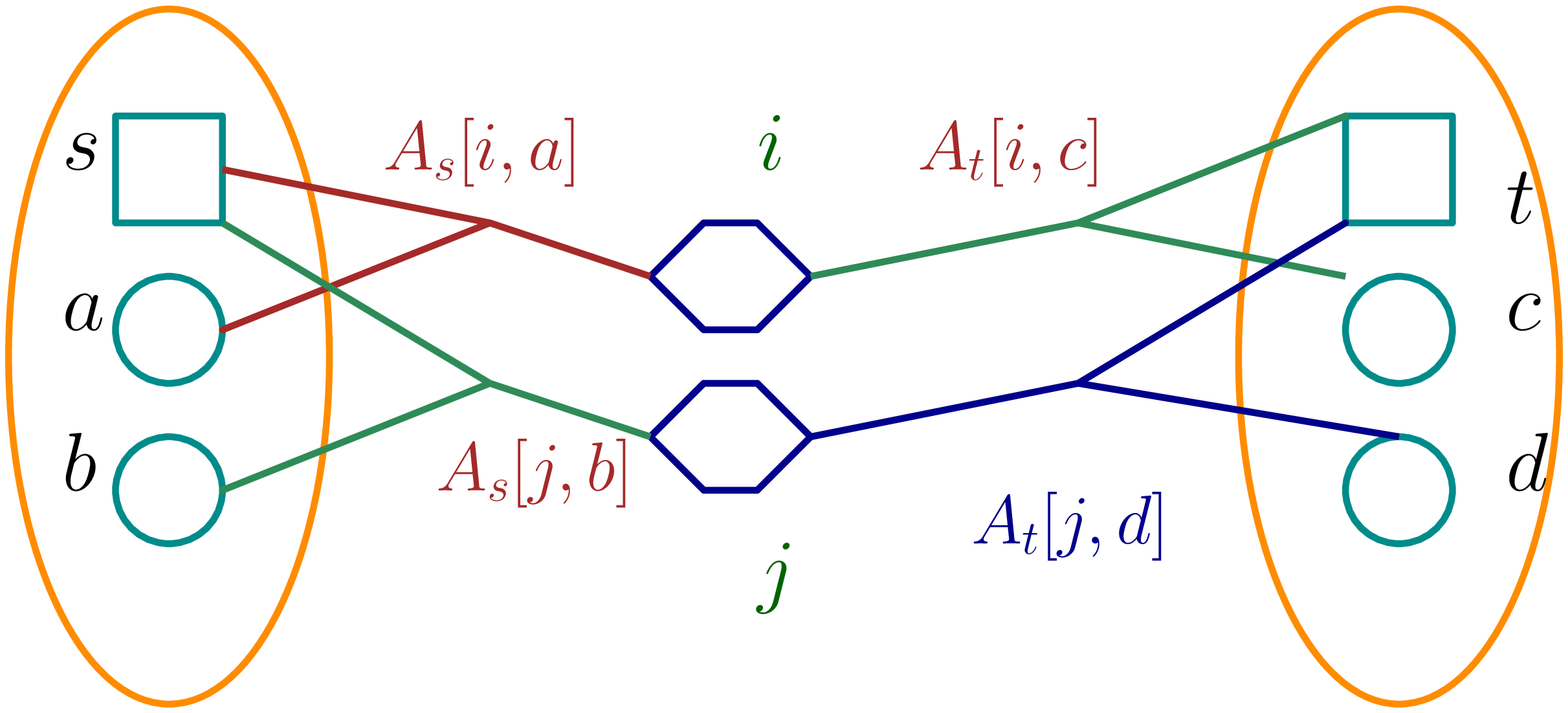}\quad\quad\quad
     \includegraphics[width=160pt]{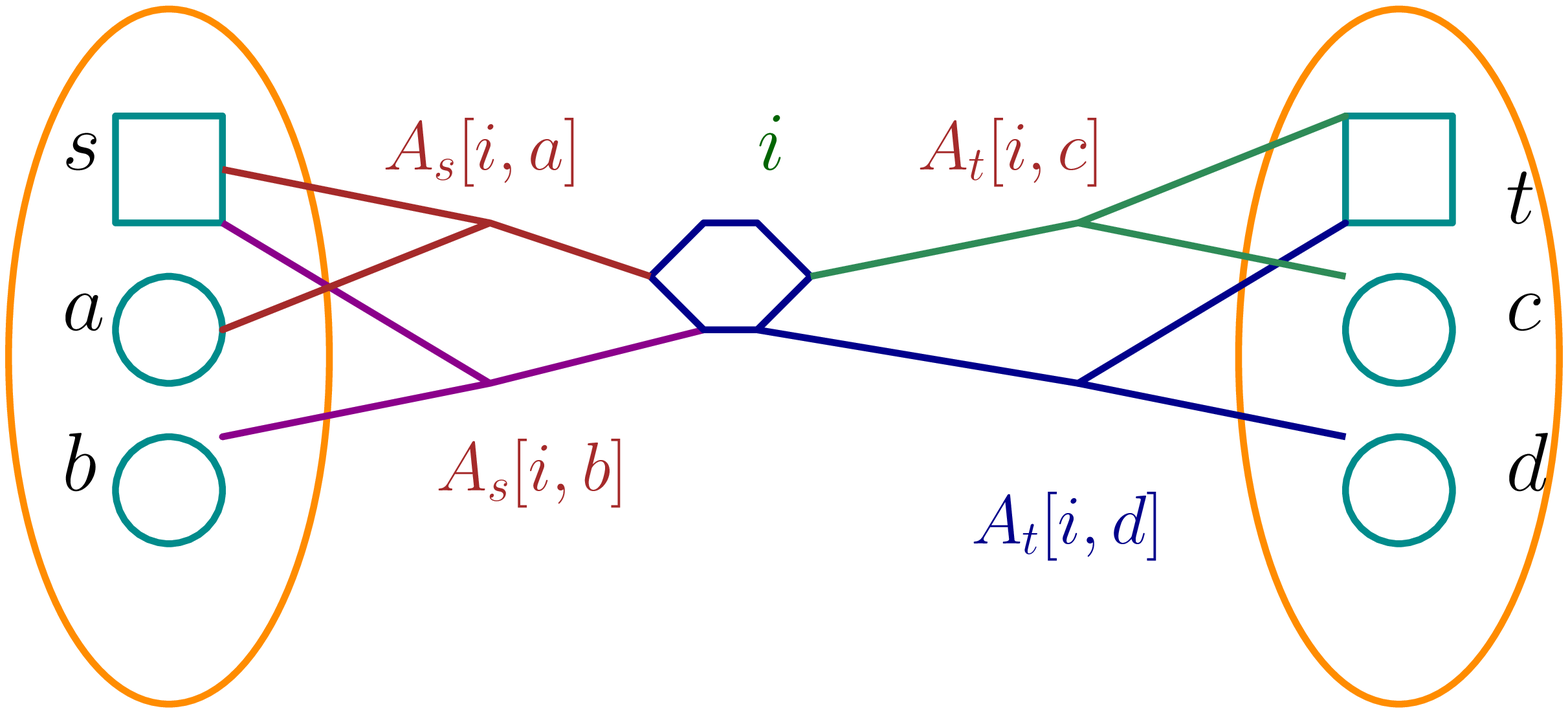}
    \caption{Diagram for $GG^\top [(s,a,b), (t,c,d)] = \sum_{i\neq j\in [n]} A_s[i,a]A_t[i,c]A_s[j,b]A_t[j,d]$.} \label{fig:GG-i-neq-j}
    \caption{Diagram for $GG^\top [(s,a,b), (t,c,d)] = \sum_{i= j\in [n]} A_s[i,a]A_t[i,c]A_s[i,b]A_t[i,d]$.} \label{fig:GG-i-equals-j}
\end{figure}

Once these diagrams are drawn out, we can apply off-the-shelf matrix norm bounds for graph matrices. The key technical idea is that the spectral norm of these random matrices is characterized (up to $\wt{O}(1)$ factors) by a simple, combinatorial quantity of the diagram (see Proposition~\ref{prop:graph-matrix-norm} and Remark~\ref{rem:min-vertex-separator}). Therefore, our job for spectral analysis is essentially estimating the combinatorial quantity for each of the diagrams that may arise in the decomposition. In the setting above, a direct application of this idea in a black-box manner gives us a spectral norm bound $\spec{\offdiag(GG^\top)} \leq \wt{O}(\sqrt{m\ell^2n^2}) = o(n^2)$ assuming $m\ell^2 \ll n^2$ and $\ell \approx \sqrt{n}$ (see Example~\ref{ex:norm-bound}), giving us the $m\approx n$ bound as we anticipate in the beginning.

\subsection{Preliminaries for graph matrices}
\label{sec:graphical-matrices}
We give a lightweight introduction to the theory of graph matrices specialized to our setting in this section, while we defer the interested reader who seeks a thorough introduction or a more formal treatment to its origin in a sequence of works in Sum-of-Squares lower bounds \cite{HKP15, BHK19, graphmatrixbounds}.
Throughout this section, we assume that there is an underlying input matrix $G$.

\begin{definition}[Shape] A shape $\al$ is a tuple $(V(\al), U_\al, V_\al, E(\al))$ associated with a (multi-hyper) graph $(V(\al), E(\al))$ where each vertex in $V(\al)$ is associated with a vertex type that indicates the range of the labels for the particular vertex, and each edge in $E(\al)$ is associated with an edge type specifies a variable of the underlying input matrix $G$ from the information of incidental endpoints of the edge. Moreover, we have $U_\al, V_\al \subseteq V(\al)$ and they may intersect. 
\end{definition}
\begin{remark}
For intuition, $U_\al, V_\al$ are simply the row/column indices of the corresponding random matrix, and we remind the reader that $V_\al$ should be distinguished from $V(\al)$ where $V_\al$ is simply a column/row boundary index, while $V(\al)$ is the set of vertices in the graph.
\end{remark}
\begin{definition}[Edges]
Each edge $e\in E(\al)$ is associated with a tuple $(s(e), E(e))$ such that $s(e)$ specifies its edge type and $E(e)$ specifies its incidental endpoints (including their corresponding vertex type), and we write the corresponding edge variable as $\chi_{s(e), E(e)}(G)$.
\end{definition}
\begin{definition}[Mapping of a shape] Given a shape $\al$, we call a function $\sigma$ a mapping of the shape if \begin{enumerate}
    \item $\sigma$ assigns a label for each vertex according to its specified vertex type;
    \item $\sigma$ is an injective mapping.
\end{enumerate}
    
\end{definition}

\begin{definition}[Boundary consistency for shape]
    Given a shape $\al$, and given boundaries $S,T$ each being a collection of sets of labels,
    we call $\sigma$ a valid labeling for $\al$ with boundary $S, T$ if \begin{enumerate}
        \item For each vertex type, $I(S)$ and $I(T)$ are subsets of labels of type $I$;
        \item $\sigma(U_\al\cap V_\al)$ is consistent with the labels in $S\cap T$ within each collection of vertex type as a set, i.e., for each vertex type $I$, we have $\sigma(I(U_\al\cap V_\al)) =  I(S\cap T)  $;
        \item And for each vertex type $I$, the vertices in $U_\al\setminus V_\al$ take label in $I(S\setminus T)$ as a set, i.e., $\sigma(I(U_\al\setminus V_\al))= I(S\setminus T)$;
        \item and similarly, the vertices in $V_\al\setminus U_\al$  take label in $I(T\setminus S)$ as a set, i.e., $\sigma(I(V_\al\setminus U_\al)) = I(T\setminus S)$.
    \end{enumerate}
\end{definition}
\begin{remark}[Set-indexed graph matrix] \label{rem:set-indexed-graph-matrix}
    Note that we adopt the convention where labeled vertices of a shape appear as a set, i.e., $U_\al\setminus V_\al$, $V_\al\setminus U_\al$ and $U_\al\cap V_\al$ are set-indices, unless otherwise stated. 
\end{remark}
\begin{definition}[Graphical matrix for shape] \label{def:graph-matrix-for-shape}
    Given a shape $\al$, we define its graphical matrix to be the matrix indexed by all possible boundary labelings $S, T$, and for each of its entry, we define \[ 
    M_{\al}[S, T] = \sum_{\substack{\sigma:\text{a mapping of $V(\al)$ }\\ \text{consistent with the boundaries} }} \prod_{e=\in E(\al)}\chi_{s(e), \sigma(E(e))}(G) \mper
    \]
\end{definition}
\begin{remark}
To unpack the above definition of boundary consistency, notice for each vertex type, the vertices in $U_\al $ and $V_\al $ receive labels from $\sigma$ as a set, and we additionally require vertices in $U_\al\cap V_\al$ of each type to receive the same labels under a consistent $\sigma$.
\end{remark}

We will encounter shapes that have multi-edges, i.e.\ the same edge appearing more than once.
Thus, we introduce the following definition.

\begin{definition}[Phantom edge and isolated vertex] \label{def:phantom-edge}
When the underlying input comes from $N(0,1)$ random variables, each edge that appears at least twice (as a multi-edge) is considered a ``phantom'' edge, and a vertex $v\in V(\al)\setminus (U_\al\cup V_\al)$ is considered ``isolated'' if it is not incident to any non-phantom edge in a shape.
\end{definition}

With this set-up, we are now ready to introduce the probabilistic norm bounds for graphical matrices, and prior works have shown that the norm bounds in this regime are governed by a  combinatorial object of the underlying shape called \emph{vertex separator}.

\begin{definition}[Vertex separator] \label{def:vertex-separator}
    For a shape $\al$, a set of vertices $S\subseteq V(\al)$ is a vertex separator if all \emph{non-phantom} paths from $U_\al$ to $V_\al$ pass through $S$.
    Each vertex in $U_\al\cap V_\al$ is by definition in any vertex separator for $\al$.
\end{definition}

\begin{proposition}[Graph matrix norm bound]
\label{prop:graph-matrix-norm}
    Let $\mathsf{wt}$ be a weight function that assigns weight based on the type of vertices s.t.\begin{enumerate}
        \item for a square vertex $\square{i}$, $\mathsf{wt}(\square{i}) = m$;
        \item for a circle vertex $\circle{i}$, $ \mathsf{wt}(\circle{i}) = \ell$;
        \item for a hexagon vertex $\hexagon{i}$, $\mathsf{wt}(\hexagon{i}) = n$.
    \end{enumerate}
    With probability at least $1-O(n^{-100\log n})$, for any shape $\al$ we have \[ 
    \spec{M_\al} \leq \wt{O}\left(\max_{\substack{S\subseteq V(\al) \\ S \text{ a separator }}} \prod_{i\notin S} \sqrt{\mathsf{wt}(i)} \cdot  \prod_{\substack{i\notin S\\\text{isolated} }} \sqrt{\mathsf{wt}(i)} \right) \mper
    \]
\end{proposition}
\begin{proof}
    This is an application of \cite{graphmatrixbounds} and Corollary 6.64 from \cite{sparseindset} in our setting. The main result from \cite{graphmatrixbounds} alone is sufficient, while we point out it is easier to use Corollary 6.64 to handle set-symmetry of our indices within the trace power method directly, which incurs a cost of $\sqrt{|U_\al|!|V_\al|!}$ which is subsumed by the polylog factor as we have $|U_\al|, |V_\al| =O(1)$ throughout our work.
\end{proof}

\begin{remark}[Analyzing minimum vertex separator] \label{rem:min-vertex-separator}
    With Proposition~\ref{prop:graph-matrix-norm}, we can now bound the spectral norm of graph matrices by analyzing their \emph{minimum weight vertex separator} $S_{\min}$.
    Furthermore, notice that the vertex separator only needs to block paths using \emph{non-phantom} edges.
    Thus, when using Proposition~\ref{prop:graph-matrix-norm}, we can essentially remove the phantom edges; this may reduce the minimum vertex separator, hence increasing the norm bound.
\end{remark}

\begin{example}[Norm bound of the shape in Figure~\ref{fig:GG-i-neq-j}] \label{ex:norm-bound}
    Two potential vertex separators are $S_1 = \{\hexagon{i}, \hexagon{j}\}$ and $S_2 = \{\square{s},\circle{a},\circle{b}\} = U_\al$.
    We can see that $S_2$ has smaller weight when $m\ell^2 \ll n^2$ since $\mathsf{wt}(\hexagon{i}) \cdot \mathsf{wt}(\hexagon{j}) = n^2$ but $\mathsf{wt}(\square{s}) \cdot \mathsf{wt}(\circle{a}) \cdot \mathsf{wt}(\circle{b}) = m\ell^2$.
    With some case analysis, one can verify that $S_2$ is indeed the minimum vertex separator of the shape. By Proposition~\ref{prop:graph-matrix-norm}, each vertex outside $S_2$ contributes the square root of its weight, thus we have a bound of $\wt{O}(\sqrt{m\ell^2 n^2})$.
\end{example}

\subsection{Singular value lower bounds for analysis of \texorpdfstring{$V$}{V}}
\label{sec:sval-analysis-of-V}

Recall that given $B_1, \dots, B_m$ which are degree-2 homogeneous polynomials in $\ell$ variables with i.i.d.\ standard Gaussian coefficients, the matrix $V\in \R^{\ell_{4D} \times m\ell_{2D}}$ is the matrix whose columns represent the polynomials $B_t(y)^D y_{j_1} \cdots y_{j_{2D}}$ for $t\in[m]$ and $j_1,\dots, j_{2D} \in [\ell]$.
Furthermore, recall the matrix $N\in \R^{\binom{m}{2} \times m\ell_{2D}}$ defined in Definition~\ref{def:N-matrix} that satisfies $VN^\top = 0$, i.e.\ each row of $N$ represents a collection of $m$ degree-$2D$ polynomials $(p_1,\dots,p_m)$ such that
\begin{equation*}
    \sum_{t=1}^m B_t(y)^D p_t(y) = 0 \mcom
\end{equation*}
and for row index $(s_1, s_2)$ of $N$, $p_t(y) = B_{s_2}(y)^D$ if $t=s_1$, $-B_{s_1}(y)^D$ if $t=s_2$, and 0 otherwise.

\subsubsection{Proof of Lemma~\ref{lem:rank-of-N}: rank of \texorpdfstring{$N$}{N}}
\label{sec:rank-of-N}

\begin{lemma}[Restatement of Lemma~\ref{lem:rank-of-N}: Rank of $N$]
    Let $m, \ell, D \in \N$ such that $m \leq (\frac{\ell}{\polylog(\ell)})^{2D}$.
    Let $N \in \R^{\binom{m}{2} \times m\ell_{2D}}$ be the matrix defined in Definition~\ref{def:N-matrix}.
    Then, with probability $1- \ell^{-\Omega(D)}$,
    \begin{equation*}
        \sigma_{\binom{m}{2}}(N) \geq \Omega(\ell^D) \mper
    \end{equation*}
\end{lemma}
\begin{proof}
    We will show that $NN^\top \in \R^{\binom{m}{2} \times \binom{m}{2}}$ is full rank and that $\lambda_{\min}(NN^\top) \geq \Omega(\ell^{2D})$.
    Notice that the columns of $N$ are indexed by $(t, I)$ where $t\in [m]$ and $I \in [\ell]^{2D}$ is a multiset.
    By Fact~\ref{fact:deleting-rows-sval}, we can delete some columns of $N$ and prove that the resulting matrix $N'$ satisfies $\lambda_{\min}(N' N'^\top) \geq \Omega(\ell^{2D})$, which implies $\lambda_{\min}(NN^\top) \geq \Omega(\ell^{2D})$.
    Thus, we will assume that the matrix $N$ only consists of columns $(t,I)$ where $I$ is a set (no repeated elements), i.e., each row consists of multilinear polynomials.
    
    We decompose $NN^\top$ into diagonal and off-diagonal components.

    \subparagraph{Diagonal entries.} The diagonal entry at $s_1 < s_2$ can be written as
    \begin{equation*}
    \begin{aligned}
        NN^\top[(s_1,s_2), (s_1,s_2)] &= \Norm{\multilinear(B_{s_1}^{D})}^2 + \Norm{\multilinear(B_{s_2}^{D})}^2 \mcom
    \end{aligned}
    \end{equation*}
    where we take the multilinear parts of $B_{s_1}^D$ and $B_{s_2}^D$ because we deleted the columns corresponding to non-multilinear monomials.
    This however does not change the norms up to constant factors, hence by Claim~\ref{claim:norm-of-polynomial-powers}, the diagonal entries are $\Omega(\ell^{2D})$.

    \subparagraph{Off-diagonal part.}
    $NN^\top[(s_1,s_2),(t_1,t_2)]$ is zero when $s_1,s_2,t_1,t_2$ are all different, and if $s_1 \neq t_1$ but $s_2 = t_2$ (for e.g.), then
    \begin{equation*}
    \begin{aligned}
        NN^\top[(s_1,s_2),(t_1,t_2)]
        &= \Iprod{\multilinear(B_{s_1}^D), \multilinear(B_{t_1}^D)} \\
        &= \Theta(1) \sum_{I \in [\ell]^{2D}} \Sym(B_{s_1}^{\ot D})[I] \cdot \Sym(B_{t_1}^{\ot D})[I] \\
        &= \Theta(1) \sum_{\pi \in \bbS_{2D}} \sum_{\substack{I \in [\ell]^{2D} \\ \text{all distinct}}} B_{s_1}^{\ot D}[I] \cdot B_{t_1}^{\ot D}[\pi(I)] \mper
    \end{aligned}
    \end{equation*}
    We view the above as a summation of graphical matrices, each corresponding to a permutation $\pi$.
    We then apply graphical matrix norm bounds (Proposition~\ref{prop:graph-matrix-norm}) introduced in Section~\ref{sec:graphical-matrices} to bound the spectral norms of them.
    Given $\pi$, the shape that arises can be described as the following,
    \begin{enumerate}
        \item A tripartite graph with exactly $2D$ circle vertices $\circle{i_1},\dots, \circle{i_{2D}}$ in the middle (that take labels in $[\ell]$)  , i.e., $V(\al)\setminus (U_\al\cup V_\al)$;
        \item $U_\al, V_\al$ both have two square vertices (that take labels in $[m]$): $U_\al = \{\square{s}, \square{u}\}$ and $V_\al = \{\square{t}, \square{u}\}$, where $\square{u}$ is in the intersection $U_\al \cap V_\al$;
        \item On the left, we add hyperedges in order: $(\square{s}, \circle{i_1}, \circle{i_2}), (\square{s}, \circle{i_3}, \circle{i_4})$, and so on;
        on the right, we add $(\square{t}, \circle{i_{\pi(1)}}, \circle{i_{\pi(2)}}), (\square{t}, \circle{i_{\pi(3)}}, \circle{i_{\pi(4)}})$, and so on.
    \end{enumerate}

    Now we analyze the minimum vertex separator.
    Observe that for every $\circle{i}$ in the middle, there is a path $\square{s} \rightarrow \circle{i} \rightarrow \square{t}$.
    Thus, any vertex separator of such shape must contain either $\square{s}$, $\square{t}$, or all of the circle vertices.
    Clearly, when $m \ll \ell^{2D}$, the minimum weight vertex separator is $\{\square{s}\}$ or $\{\square{t}\}$.
    Thus, by Proposition~\ref{prop:graph-matrix-norm}, we get a norm bound of $\wt{O}\left(\sqrt{m\ell^{2D}}\right)$.

    Therefore, we can write
    \begin{equation*}
        NN^\top \succeq \Omega(\ell^{2D}) \cdot \Id + \offdiag(NN^\top) \mcom
    \end{equation*}
    where $\spec{\offdiag(NN^\top)} \leq \wt{O}(\sqrt{m} \ell^{D}) = o(\ell^{2D})$ given that $m \leq (\frac{\ell}{\polylog(\ell)})^{2D}$.
    Thus, the diagonal dominates, and we have $\lambda_{\min}(NN^\top) \geq \Omega(\ell^{2D})$, hence $\sigma_{\binom{m}{2}}(N) \geq \Omega(\ell^D)$.
\end{proof}

\subsubsection{Proof of Lemma~\ref{lem:VV-NN}: singular value of \texorpdfstring{$V^\top V + N^\top N$}{VV+NN}}
\label{sec:singular-value-of-VV-NN}

Lemma~\ref{lem:VV-NN} is a special case ($K=2$) of the following lemma, thus the reader can focus on this case for simplicity.

\begin{lemma}[Generalization of Lemma~\ref{lem:VV-NN}]
\label{lem:VV-NN-generalized}
    Fix $D, K\in \N$ with $K\geq 2$.
    Let $m, \ell \in \N$ such that $m \leq (\frac{\ell}{\polylog(\ell)})^{KD}$ if $D\leq 2$, and $m\leq (\frac{\ell}{\polylog(\ell)})^{KD/2}$ if $D>2$, and let $N$ be the $\binom{m}{2}\ell_{(K-2)D} \times m\ell_{2(K-1)D}$ matrix defined in Definition~\ref{def:null-space-of-V-high-deg}.
    Then with probability $1-\ell^{-\Omega(D)}$,
    \begin{equation*}
        \lambda_{\min}\Paren{V^\top V + N^\top N}
        \geq \Omega(\ell^{KD}) \mper
    \end{equation*}
\end{lemma}

Recall that the matrix $V$ has dimension $\ell_{(3K-2)D} \times m\ell_{2(K-1)D}$, each column representing a degree-$((3K-2)D)$ polynomial $B_s(y)^D y_I$ where $I\in [\ell]^{2(K-1)D}$ is a multiset.
Specifically, the rows of $V$ are indexed by multisets $H \in [\ell]^{(3K-2)D}$, and for $s\in[m]$, $I\in [\ell]^{2(K-1)D}$, the entry $V[H, (s,I)]$ is the coefficient of the monomial $y_H$ in the polynomial $B_s(y)^D y_I$.
Thus, $V[H,(s,I)] = 0$ if $I \not\subset H$.

The entries of $V^\top V$ are indexed by $(s,I), (t,J)$ where $s,t\in [m]$ and $I, J \in [\ell]^{2(K-1)D}$.
Specifically,
\begin{equation*}
    V^\top V[(s,I), (t,J)] = \Iprod{B_s(y)^D y_I, B_t(y)^D y_J}
\end{equation*}
where we recall (from Section~\ref{sec:prelims}) this is the inner product between the coefficient vectors of the two polynomials.
\begin{remark}[Intersection of $I,J$]
\label{rem:I-J-intersection}
    Remark that $B_s(y)^D y_I, B_t(y)^D y_J$ are polynomials of degree $(3K-2)D$ while $|I| = |J| = 2(K-1)D$, thus if $|I\cap J| < (K-2)D$, then the polynomials have no common monomial and $V^\top V[(s,I),(t,J)]$ must be zero (note that $I$ and $J$ can be disjoint when $K=2$).
\end{remark}

\paragraph{Proof outline}
First observe that both $V^\top V$ and $N^\top N$ cannot be full rank: $V^\top V$ has a null space (namely row span of $N$ by definition) and $N^\top N$ is not full rank simply because $\binom{m}{2} \ell_{(K-2)D} < m\ell_{2(K-1)D}$ (since $m \ll \ell^{KD}$).
Thus, even though $V^\top V$ and $N^\top N$ have large diagonal entries, we cannot hope to bound the spectral norm of their off-diagonal parts separately and charge them to the diagonals (this would contradict that they are not full rank).
In fact, we will see that there are off-diagonal components with large spectral norm.

Nevertheless, we can still proceed to analyze $V^\top V$ and see what graph matrices arise.
Observe that $V^\top V$ has a natural block structure indexed by $(s,t)$, and each block has dimension $\ell_{2(K-1)D} \times \ell_{2(K-1)D}$.
We will analyze the diagonal blocks ($s=t$) and off-diagonal blocks ($s\neq t$) separately.
First, we show the following,
\begin{lemma}[Diagonal blocks of $V^\top V$] \label{lem:VV-diag-blocks}
    Let $m,\ell,D\in\N$ be the same parameters as Lemma~\ref{lem:VV-NN-generalized}, then
    \begin{equation*}
        \diagblocks(V^\top V) = \Omega(\ell^{KD}) \cdot \Id + P_1 + E_1
    \end{equation*}
    where $P_1 \succeq 0$ and $\spec{E_1} \leq o(\ell^{KD})$.
\end{lemma}

For the off-diagonal blocks of $V^\top V$, it turns out that most graph matrices that arise have small spectral norm, except for one specific matrix defined as follows,

\begin{definition}[$W$ matrix] \label{def:W-matrix}
    Let $W$ be the matrix with the same dimensions and indices as $V^\top V$ such that $W[(s,I),(t,J)] = V^\top V[(s,I),(t,J)]$ if $s\neq t$ and $|I\cap J| = (K-2)D$ (the minimum intersection; see Remark~\ref{rem:I-J-intersection}), and 0 otherwise.
\end{definition}


\begin{lemma}[Off-diagonal blocks of $V^\top V$] \label{lem:VV-off-diag-blocks}
    Let $m,\ell,D\in\N$ be the same parameters as Lemma~\ref{lem:VV-NN-generalized}, then
    \begin{equation*}
        \offdiagblocks(V^\top V) = W + E_2
    \end{equation*}
    where $\spec{E_2} \leq o(\ell^{KD})$.
\end{lemma}

Finally, we turn to $N^\top N$.
Interestingly, we show that $N^\top N$ has a $-W$ component that cancels out the one from $\offdiagblocks(V^\top V)$.

\begin{lemma}[Analysis of $N^\top N$] \label{lem:NN}
    Let $m,\ell,D\in\N$ be the same parameters as Lemma~\ref{lem:VV-NN-generalized}, then
    \begin{equation*}
        N^\top N = P_2 - W + E_3
    \end{equation*}
    where $P_2 \succeq 0$ and $\spec{E_3} \leq o(\ell^{KD})$.
\end{lemma}

The combination of Lemmas~\ref{lem:VV-diag-blocks}, \ref{lem:VV-off-diag-blocks} and \ref{lem:NN} immediately imply Lemma~\ref{lem:VV-NN-generalized}.
\begin{proof}[Proof of Lemma~\ref{lem:VV-NN-generalized}]
    Combining Lemmas~\ref{lem:VV-diag-blocks}, \ref{lem:VV-off-diag-blocks} and \ref{lem:NN}, we have
    \begin{equation*}
        V^\top V + N^\top N \succeq \Omega(\ell^{KD}) \cdot \Id + P_1 + P_2 + E_1 + E_2
    \end{equation*}
    where $P_1, P_2\succeq 0$ and $\spec{E_1}, \spec{E_2} \leq o(\ell^{KD})$.
    This completes the proof.
\end{proof}

\paragraph{Graph matrices that arise from $V^\top V$}
We write out the entries of $V^\top V$ explicitly.
For entry $V^\top V[(s,I), (t,J)]$ where $s,t\in [m]$ and $I,J \in [\ell]^{2(K-1)D}$,
\begin{equation} \label{eq:VV-entry}
\begin{aligned}
    V^\top V[(s,I), (t,J)]
    &= \Iprod{B_s(y)^D y_I, B_t(y)^D y_J} \\
    &= \sum_{\substack{H \in [\ell]^{(3K-2)D} \\ I \cup J\subset H}} \Theta(1) \cdot (B_s(y)^{D} y_{I})[H] \cdot (B_t(y)^{D} y_{J})[H] \\
    &= \sum_{\substack{H \in [\ell]^{(3K-2)D} \\ I \cup J\subset H}} \sum_{\pi,\pi'\in \bbS_{KD}} \Theta(1) \cdot 
    B_s^{\ot D}[\pi(H \setminus I)] \cdot B_t^{\ot D}[\pi'(H \setminus J)]
    \mper
\end{aligned}
\end{equation}
Here, the summation is over ordered tuples $H$ of size $(3K-2)D$ containing both $I$ and $J$.

Let us parse Equation~\pref{eq:VV-entry}.
Let $r = |I\cap J|$, $L = I\cap J$, and let $I' = I \setminus L$, $J' = J \setminus L$ and $H' = H \setminus (I\cup J)$ (recall that these are multiset operations).
First observe that $|I'| = |J'| = 2(K-1)D-r$ and $|I \cup J| = 4(K-1)D - r$, hence $|H'| = r - (K-2)D$.
Thus, $r$ can range from $(K-2)D$ to $2(K-1)D$ (note that $r$ cannot be 0 if $K > 2$).
We can think of $H'$ as the ``free indices'' in the summation over $H$ since $H$ must contain $I \cup J$.
In particular, if $r = (K-2)D$ (the minimum intersection), then there is only a single $H$, namely $I \cup J$.
We can now define the graph matrices that arise from \pref{eq:VV-entry}.

\begin{definition}[Graph matrices from $V^\top V$]
\label{def:VV-graph-matrices}
    First assume that $I'$, $J'$, $H'$ and $L$ all have distinct elements.
    For a fixed $r\in \{(K-2)D,\dots,2(K-1)D\}$ and permutations $\pi, \pi' \in \bbS_{KD}$, the graph matrix is described as follows,

    \begin{enumerate}
        \item A graph with exactly $(3K-2)D$ circle vertices (that take labels in $[\ell]$), denoted as follows (slightly abusing notation by writing $I', J', H'$ as ordered tuples of circle vertices):
        \begin{itemize}
            \item $I' = (\circle{a_1},\dots, \circle{a_{|I'|}})$ where $|I'| = 2(K-1)D - r$;
            \item $J' = (\circle{b_1},\dots, \circle{b_{|J'|}})$ where $|J'| = 2(K-1)D - r$;
            \item $H' = (\circle{a_{|I'|+1}} = \circle{b_{|J'|+1}}, \circle{a_{|I'|+2}} = \circle{b_{|J'|+2}}, \dots, \circle{a_{KD}} = \circle{b_{KD}})$;
            \item $L = (\circle{c_1},\dots, \circle{c_r})$, which corresponds to $I\cap J$.
        \end{itemize}

        \item Two square vertices $\square{s}$ and $\square{t}$ (that take labels in $[m]$).

        \item $U_\al = \{\square{s}\} \cup I' \cup L$ and $V_\al = \{\square{t}\} \cup J' \cup L$;
        they are treated as sets.

        \item On the left, we add $D$ hyperedges:
        $(\square{s}, \circle{b_{\pi((i-1)K+1)}}, \dots, \circle{b_{\pi(iK)}})$ for $i=1,\dots, D$;
        on the right, we add
        $(\square{s}, \circle{b_{\pi'((i-1)K+1)}}, \dots, \circle{b_{\pi'(iK)}})$ for $i=1,\dots, D$.
        Each hyperedge touches a square vertex ($\square{s}$ or $\square{t}$) and $K$ circle vertices, representing an element in $B_s$ or $B_t$.
        
    \end{enumerate}
    When there are repeated elements, all such components can be viewed as \emph{collapses} of the graph matrices described above.
    For example, we allow the circle vertices to collapse, with the exception that $I'$ and $J'$ remain disjoint (by definition).
    We also allow $\square{s}$ and $\square{t}$ to collapse, representing the diagonal blocks where $s=t$.
\end{definition}

Let's make some crucial observations:
(1) $U_\al \cap V_\al = L = \{\circle{c_1},\dots, \circle{c_r}\}$ and they are not connected to any edge,
(2) $I' \subset U_\al$ but have edges connected to $\square{t}$, and similarly $J' \subset V_\al$ but have edges connected to $\square{s}$,
(3) $H'$ are outside of $U_\al \cup V_\al$ and have edges connected to both $\square{s}$ and $\square{t}$,
and finally (4) there are $D$ hyperedges on the left (resp.\ right) connected to $\square{s}$ (resp.\ $\square{t}$).

Furthermore, we note that $U_\al$ and $V_\al$ are both treated as sets since $I$ and $J$ in \pref{eq:VV-entry} are both (multi)sets (recall Remark~\ref{rem:set-indexed-graph-matrix}).

\paragraph{Diagonal blocks of $V^\top V$: Proof of Lemma~\ref{lem:VV-diag-blocks}}

\begin{proof}[Proof of Lemma~\ref{lem:VV-diag-blocks}]
    The diagonal blocks of $V^\top V$ consist of all graph matrices in Definition~\ref{def:VV-graph-matrices} with $s=t$.
    For the diagonal entries, for index $(s,I)$ we have
    \begin{equation} \label{eq:VV-diag}
        V^\top V[(s,I), (s,I)] = \Norm{B_s(y)^D y_I}^2 = \Norm{B_s(y)^D}^2 = \Omega(\ell^{KD})
    \end{equation}
    by Claim~\ref{claim:norm-of-polynomial-powers}.

    We next analyze the off-diagonal entries (of the diagonal blocks).
    Fix $r = |I\cap J|$.
    When $r = 2(K-1)D$, i.e.\ $I=J$, this is simply the diagonal entries, so we focus on the components with $r \in \{(K-2)D,\dots,2(K-1)D-1\}$ (recall that $r \geq (K-2)D$ from Remark~\ref{rem:I-J-intersection}).

    First, we note that there are ``troublesome'' components with large spectral norm which cannot be charged to the diagonal.
    For example when $r=(K-2)D$, $|I'| + |J'| = 2KD$ and $|H'| = 0$, hence the minimum vertex separator $S$ consists of only $\square{s}$ and there can be $2KD$ circle vertices outside of $S$, giving a bound of $\wt{O}(\ell^{KD})$, while the diagonal is only $\Omega(\ell^{KD})$.
    Figure~\ref{fig:my_label} is an example of such a shape.

    \begin{figure}[h!]
        \centering
        \includegraphics[width=160pt]{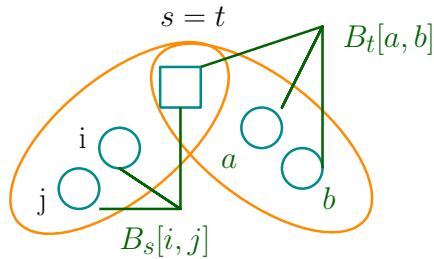}
        \caption{$M_\al[(s,i,j), (t,a,b)]= B_s[i.j]\cdot B_t[a,b]$.}
        \label{fig:my_label}
    \end{figure}


    The crucial observation is that these components are ``roughly'' PSD.
    To see this, we first need to define the following,

    \begin{definition}[Half-gram shape] \label{def:half-gram-shape}
        For a fixed $r \in \{(K-2)D,\dots,2(K-1)D-1\}$ and $\pi\in \bbS_{KD}$, we define a half-gram shape $\beta_{\pi}$ as follows,
        \begin{enumerate}
            \item A bipartite graph with one square vertex $\square{s}$ and exactly $2D+k$ circle vertices:
            \begin{itemize}
                \item $I' = \{\circle{a_1},\dots, \circle{a_{|I'|}} \}$ where $|I'| = 2(K-1)D-r$;
                \item $H' = \{\circle{a_{|I'|+1}}, \dots, \circle{a_{KD}}\}$;
                \item $L = \{\circle{c_1},\dots, \circle{c_r}\}$.
            \end{itemize}

            \item $U_{\beta_{\pi}} = \{\square{s}\} \cup I' \cup L$ and $V_{\beta_{\pi}} = (\square{s}, H', L)$.
            Here, $U_{\beta_{\pi}}$ is treated as sets but $V_{\beta_{\pi}}$ is treated as ordered tuples.

            \item We add $D$ hyperedges:
            $(\square{s}, \circle{a_{\pi((i-1)K+1)}}, \dots, \circle{a_{\pi(iK)}})$ for $i=1,\dots, D$.
        \end{enumerate}
        Furthermore, we define
        \begin{equation*}
            M_{\beta} \coloneqq \sum_{\pi\in\bbS_{KD}} M_{\beta_{\pi}} \mper
        \end{equation*}
    \end{definition}


    With the definitions in hand, we can now observe that for each $r$ and permutations $\pi,\pi'\in\bbS_{KD}$, any ``troublesome'' shape $\al_{\pi,\pi'}$ from Definition~\ref{def:VV-graph-matrices} has two unique half-gram shapes $\beta_{\pi}$ and $\beta_{\pi'}$ such that $M_{\al_{\pi,\pi'}} \approx M_{\beta_{\pi'}} M_{\beta_{\pi}}^\top$ (note the order of $\pi'$ and $\pi$).
    Summing over all permutations $\pi,\pi'\in \bbS_{KD}$ (as in \pref{eq:VV-entry}), we get that 
    \begin{equation*}
        \sum_{\pi,\pi'\in \bbS_{KD}} M_{\al_{\pi,\pi'}} \approx M_{\beta} M_{\beta}^\top \mper
    \end{equation*}

    However, this is only \emph{approximately} true!
    This subtlety arises because $\beta_{\pi'}\circ\beta_{\pi}^\top$ may form extra intersection terms. That being said, intersections, as anticipated in many situations, can be bounded in spectral norm and charged to our diagonal \pref{eq:VV-diag}.
    For starters, we first observe that for shape $\al \coloneqq \alpha_{\pi,\pi'}$, the shapes that arise in the intersection terms of $\beta_{\pi'}\circ\beta_\pi^\top$ can be described by merging vertices in $I'$ with vertices in $J'$ (recall from Definition~\ref{def:VV-graph-matrices} that $I' = U_\al \setminus V_\al$ and $J' = V_\al \setminus U_\al$).
    We now prove that all such intersection terms have small spectral norm compared to the diagonal $\Omega(\ell^{KD})$.
    
    \begin{claim}[Bounding intersection shapes]\label{claim:norm-bound-gram-decompos}
        Let $\beta_{\pi}$ and $\beta_{\pi'}$ be the half-gram shapes of shape $\alpha$ with $r = |I\cap J| \in \{(K-2)D,\dots,2(K-1)D-1\}$.
        Let $\tau$ be an intersection shape that arises from the intersection terms of $\beta_{\pi'} \circ\beta_{\pi}^\top$. 
        Then, we have
        \[ 
        \spec{M_{\tau}} \leq \wt{O}\left(\ell^{KD-1} \right) \mper
        \]
    \end{claim}
    \begin{proof}
        The intersection terms of $\beta_{\pi'}\circ\beta_\pi^\top$ can be described by merging vertices in $I'$ with vertices in $J'$ of the graph matrices from Definition~\ref{def:VV-graph-matrices}.
        We will first show that without merging any vertex, we have $\spec{M_\al} \leq \wt{O}(\ell^{KD})$, and then show that intersection strictly decreases the norm from there, giving us the desired bound.
        \begin{itemize}
            \item When $r=(K-2)D$, we have $H' = \emptyset$. The minimum vertex separator $S$ consists of only $\square{s}$ and there can be $|I' \cup J'| = 2KD$ circle vertices outside of $S$, which gives a bound of $\wt{O}(\ell^{2D})$.

            \item When $r > (K-2)D$, we have $|H'| = r-(K-2)D$.
            Observe that each vertex in $H'$ has two incident edges, hence depending on $\pi,\pi'$, there can be phantom edges which make some vertices isolated (recall Definition~\ref{def:phantom-edge}).
            In the worst case, all edges incident to $H'$ are phantom edges and all vertices in $H'$ are isolated.
            Since $|I' \cup J'| = 4(K-1)D-2r$, Proposition~\ref{prop:graph-matrix-norm} gives a bound of $\wt{O}(\ell^{2(K-1)D-r}\cdot \ell^{r-(K-2)D}) = \wt{O}(\ell^{KD})$.
        \end{itemize}
        The intersection terms can be obtained by merging vertices in $U_\al\setminus V_\al$ and $V_\al\setminus U_\al$.
        Notice that each merge transforms $2$ vertices that originally contribute $\sqrt{\ell}$ factor from the previous bound to be inside $U_\tau \cap V_\tau$, the mandatory separator. This causes a drop by a factor of $\ell$, yielding our desired bound.
    \end{proof}

    Thus, Claim~\ref{claim:norm-bound-gram-decompos} states that
    \begin{equation*}
        \sum_{\pi,\pi'\in\bbS_{KD}} M_{\alpha_{\pi,\pi'}} = M_\beta M_\beta^\top + E
    \end{equation*}
    where the error $E$ consists of all (constant number of) intersection terms and $\spec{E} \leq \wt{O}(\ell^{KD-1}) = o(\ell^{KD})$.
    From \pref{eq:VV-diag}, the diagonal of $V^\top V$ contributes $\Omega(\ell^{KD}) \cdot \Id$.
    This completes the proof.
\end{proof}

\paragraph{Off-diagonal blocks of $V^\top V$: Proof of Lemma~\ref{lem:VV-off-diag-blocks}}
\begin{proof}[Proof of Lemma~\ref{lem:VV-off-diag-blocks}]
    The off-diagonal blocks of $V^\top V$ consist of shapes and collapsed shapes from Definition~\ref{def:VV-graph-matrices} where $s\neq t$.
    We first look at the case when $D \leq 2$ and $r = |I\cap J| > (K-2)D$ (recall Remark~\ref{rem:I-J-intersection} that $|I\cap J| \geq (K-2)D$).

    \subparagraph{When $D \leq 2$ and $r > (K-2)D$.}
    We analyze the minimum vertex separator $S$ depending on $r = |I\cap J|$.
    Recall the definitions of $I', J', H'$ and $L$ from Definition~\ref{def:VV-graph-matrices} where $|I'|, |J'| \leq 2(K-1)D - r$ and $|H'| \leq r-(K-2)D$.
    First, $L \subseteq U_\al \cap V_\al$ must always be in $S$.
    Next, if $\square{s}$ is not in the vertex separator, then since $J' \subset V_\al$ and is connected to $\square{s}$, all of $J'$ must be in the vertex separator.
    Similarly, if $\square{t}$ is not in the vertex separator, then $I'\subset S$.
    We consider the following types of vertex separators,
    \begin{enumerate}
        \item $\square{s}, \square{t} \in S$: there can be $|I' \cup J' \cup H'| \leq (3K-2)D-r < 2KD$ circle vertices outside the vertex separator, hence this gives a bound of $\wt{O}((\sqrt{\ell})^{(3K-2)D-r}) = o(\ell^{KD})$.

        \item One of $\square{s}$ or $\square{t}$ is in $S$ (\WLOG assume $\square{s}\in S$ and $\square{t} \notin S$): in this case $I' \subset S$, thus $\{\square{t}\} \cup J' \cup H'$ can be outside $S$ and $|J' \cup H'| \leq KD$, giving a bound of $\wt{O}(\sqrt{m} (\sqrt{\ell})^{KD}) = \wt{O}(\sqrt{m} \ell^{KD/2})$.

        \item $\square{s}, \square{t}\notin S$: $I', J', H'$ must all be in $S$, giving a bound of $\wt{O}(m)$.
    \end{enumerate}
    When $m \leq (\frac{\ell}{\polylog(\ell)})^{KD}$ and $r > (K-2)D$, the bounds above are all $o(\ell^{KD})$.

    \subparagraph{When $D>2$ and $r > (K-2)D$.}
    The subtlety here is that since $|I'|, |J'|$ can each have $2(K-1)D-r \geq 2K$ circle vertices (this won't happen when $D\leq 2$ because $|I'| < KD$).
    Observe that $2K$ vertices in $I'$ (2 hyperedges) can collapse into $K$ vertices and create a phantom edge, essentially removing the 2 edges (recall Definition~\ref{def:phantom-edge} and Remark~\ref{rem:min-vertex-separator}).
    In this case, even if we don't pick $\square{t}$ in the vertex separator, these circle vertices aren't connected to $\square{t}$ hence are not required to be in the vertex separator.

    Thus, we need to look at Case 2 and 3 of the analysis of $D\leq 2$, i.e.\ when $\square{s}$ and $\square{t}$ are not both in $S$.
    Let $I'_{\ph}$ be the set of circle vertices in $I'$ incident to phantom edges (hence not connected to $\square{t}$), and define $J'_{\ph}, H'_{\ph}$ similarly.
    There can be at most $2(K-1)D-r$ circle vertices in $I'$ or $J'$ and $r-(K-2)D$ in $H'$ before the collapse, so
    \begin{equation*}
    \begin{aligned}
        |I'_{\ph}|, |J'_{\ph}| &\leq \Floor{\frac{1}{2}(2(K-1)D-r)} \leq \Floor{\frac{kD-1}{2}} < \frac{KD}{2} \mcom \\
        |H'_{\ph}| & \leq \Floor{\frac{1}{2}(r-(K-2)D)} \mper
    \end{aligned}
    \end{equation*}
    We again study the vertex separator $S$:
    \begin{enumerate}
        \item One of $\square{s}$ or $\square{t}$ is in $S$ (\WLOG assume $\square{s}\in S$ and $\square{t} \notin S$):
        in this case $I' \setminus I'_{\ph} \subset U_\al$ is connected to $\square{t} \in V_\al$, thus $I'\setminus I'_{\ph}$ must be in the vertex separator.
        Thus, $\{\square{t}\} \cup J' \cup H'\cup I'_{\ph}$ can be outside of $S$ and $|J' \cup H' \cup I'_{\ph}| < KD + KD/2$.
        Here it doesn't matter whether $H'$ has phantom edges, as they contribute the same to the norm bound.
        Thus, we have a bound of $\wt{O}(\sqrt{m} (\sqrt{\ell})^{3KD/2}) = \wt{O}(\sqrt{m} \ell^{3KD/4})$.

        \item $\square{s}, \square{t}\notin S$: $I'\setminus I'_{\ph}, J'\setminus J'_{\ph}$ and $H'\setminus H'_{\ph}$ must be in $S$, so $\{\square{s},\square{t}\} \cup I'_{\ph} \cup J'_{\ph} \cup H'_{\ph}$ can be outside of $S$, and $H'_{\ph}$ are considered isolated in Proposition~\ref{prop:graph-matrix-norm} (recall from Definition~\ref{def:phantom-edge} that $I'_{\ph}, J'_{\ph} \subset U_\al \cup V_\al$ hence they are not isolated).
        This gives a bound of $\wt{O}(m (\sqrt{\ell})^{\floor{2(K-1)D-r}} \cdot \ell^{\floor{\frac{1}{2}(r-(K-2)D)}}) \leq \wt{O}(m \ell^{KD/2})$.
    \end{enumerate}
    When $m \leq (\frac{\ell}{\polylog(\ell)})^{KD/2}$, the bounds above are all $o(\ell^{KD})$.

    \subparagraph{When $r=(K-2)D$.}
    In this case this component corresponds to the entries $V^\top V[(s,I), (t,J)]$ with $s\neq t$ and $|I\cap J| = (K-2)D$.
    But this is exactly the $W$ matrix defined in Definition~\ref{def:W-matrix}.

    \subparagraph{Completing the proof.}
    Combining the above, we see that for $r>(K-2)D$, when $D\leq 2$, $m \leq (\frac{\ell}{\polylog(\ell)})^{KD}$ and when $D > 2$, $m \leq (\frac{\ell}{\polylog(\ell)})^{KD/2}$, all such graph matrices have spectral norm bounded by $o(\ell^{KD})$.
    For $r=(K-2)D$ this component is exactly $W$.
    Thus, $\offdiagblocks(V^\top V) = W + E$ with $\spec{E} = o(\ell^{KD})$, completing the proof.
\end{proof}

\paragraph{Graph matrices that arise from $N^\top N$}
Recall that the $\binom{m}{2}\ell_{(K-2)D} \times m\ell_{2(K-1)D}$ matrix $N$ defined in Definition~\ref{def:null-space-of-V-high-deg} has rows indexed by $(t_1, t_2, L)$ for $t_1 < t_2 \in [m]$, multiset $L\in [\ell]^{(K-2)D}$, and columns indexed by $(s, I)$ for $s\in[m]$ and multiset $I \in [\ell]^{2(K-1)D}$, and the entry is nonzero only if $s\in \{t_1,t_2\}$ and $L \subset I$.

Consider the entry $N^\top N[(s,I), (t,J)]$ for $s < t \in [m]$ and multisets $I, J\in [\ell]^{2(K-1)D}$,
\begin{equation} \label{eq:NN-entry}
\begin{aligned}
    N^\top N[(s,I), (t,J)]
    &= \sum_{t_1 < t_2, L\in [\ell]^{(K-2)D}} N[(t_1,t_2, L), (s,I)] \cdot N[(t_1,t_2, L), (t,J)] \\
    &= \sum_{L\in [\ell]^{(K-2)D}} N[(s,t, L),(s,I)] \cdot N[(s,t, L), (t,J)] \\
    &= - \sum_{L\in [\ell]^{(K-2)D}, L\subset I, J} (B_t(y)^{D} y_L) [I] \cdot (B_s(y)^{D}y_L)[J] \mper
\end{aligned}
\end{equation}
Note that there is only one pair $(t_1,t_2)$ in the summation that's nonzero, which is $(t_1,t_2) = (s,t)$.
Moreover, since $L \subset I, J$, if $|I\cap J| < (K-2)D$ then the entry is zero.
We emphasize that the negative sign in \pref{eq:NN-entry} is crucial; as we will see later, it allows us to cancel out the $W$ matrix from Lemma~\ref{lem:VV-off-diag-blocks}.
We can now define the (unsigned) graph matrices that arise from \pref{eq:NN-entry}.

\begin{definition}[Graph matrices from $N^\top N$]
\label{def:NN-graph-matrices}
    Fix $r \in \{(K-2)D,\dots,2(K-1)D\}$ and permutations $\pi,\pi'\in \bbS_{KD}$, the graph matrix is described as follows,
    \begin{enumerate}
        \item A graph with exactly $4(K-1)D-r$ circle vertices (that take labels in $[\ell]$), denoted as follows,
        \begin{itemize}
            \item $I' = (\circle{a_1},\dots, \circle{a_{|I'|}})$ where $|I'| = 2(K-1)D - r$;
            \item $J' = (\circle{b_1},\dots, \circle{b_{|J'|}})$ where $|J'| = 2(K-1)D - r$;
            \item $H' = (\circle{a_{|I'|+1}} = \circle{b_{|J'|+1}}, \dots, \circle{a_{KD}} = \circle{b_{KD}})$; $|H'| = r-(K-2)D$;
            \item $L = (\circle{c_1},\dots, \circle{c_{|L|}})$ where $|L| = (K-2)D$.
        \end{itemize}

        \item Two square vertices $\square{s}$ and $\square{t}$ (that take labels in $[m]$).

        \item $U_\al = \{\square{s}\} \cup I' \cup H' \cup L$ and $V_\al = \{\square{t}\} \cup J' \cup H' \cup L$;
        they are treated as sets.

        \item On the left, we add $D$ hyperedges:
        $(\square{s}, \circle{b_{\pi((i-1)K+1)}}, \dots, \circle{b_{\pi(iK)}})$ for $i=1,\dots, D$;
        on the right, we add
        $(\square{s}, \circle{b_{\pi'((i-1)K+1)}}, \dots, \circle{b_{\pi'(iK)}})$ for $i=1,\dots, D$.
        Each hyperedge touches a square vertex ($\square{s}$ or $\square{t}$) and $K$ circle vertices, representing an element in $B_s$ or $B_t$.

    \end{enumerate}
    When $I$ and $J$ have repeated elements, all such components can be viewed as \emph{collapses} of the graph matrices described above.
\end{definition}

\begin{remark}
    The graph matrices are defined without signs, whereas there is a crucial negative sign in \pref{eq:NN-entry}.
    However, for spectral norm bounds the sign does not matter.
    Note also the similarities between Definitions~\ref{def:VV-graph-matrices} and \ref{def:NN-graph-matrices}.
\end{remark}

\paragraph{Analysis of $N^\top N$: Proof of Lemma~\ref{lem:NN}}

\begin{proof}[Proof of Lemma~\ref{lem:NN}]
    Like $V^\top V$, $N^\top N$ also has the same $m\times m$ block structure indexed by $s,t\in [m]$.
    First, the diagonal blocks ($s=t$) are clearly positive semidefinite (here we do not need to show that they are full rank).

    Now we analyze the off-diagonal blocks ($s\neq t$).
    Observe from \pref{eq:NN-entry} that when $|I\cap J| = (K-2)D$, $L$ must be $I\cap J$, i.e.\ there is only one nonzero term in the summation, and we have
    \begin{equation*}
        N^\top N[(s,I), (t,J)]
        = - B_t^D[I\setminus J] \cdot B_s^D[J \setminus I]
        = -\Iprod{B_s(y)^D y_I, B_t(y)^D y_J}
        = -V^\top V[(s,I),(t,J)] \mcom
    \end{equation*}
    but this is exactly the $-W$ matrix from Definition~\ref{def:W-matrix}!
    Notice the crucial negative sign here.

    When $r = |I\cap J| > (K-2)D$, we bound the spectral norm of the graph matrices from Definition~\ref{def:NN-graph-matrices}.
    The analysis is almost identical to the analysis of the off-diagonal blocks of $V^\top V$ (proof of Lemma~\ref{lem:VV-off-diag-blocks}).

    Fix $k \in \{(K-2)D+1,\dots,2(K-1)D\}$ and $\pi,\pi' \in \bbS_{KD}$.
    When $D \leq 2$ and $m \leq (\frac{\ell}{\polylog(\ell)})^{KD}$, it is easy to see that $\spec{M_\al} \leq \wt{O}(m) \leq o(\ell^{KD})$ by including all circle vertices in the vertex separator.

    When $D > 2$ and $m \leq (\frac{\ell}{\polylog(\ell)})^{KD/2}$, we again define $I'_{\ph}, J'_{\ph}$ and $H'_{\ph}$ to be the circle vertices in $I', J', H'$ disconnected from $\square{t}$ and $\square{s}$ due to phantom edges after collapsing, same as the proof of Lemma~\ref{lem:VV-off-diag-blocks}.
    We know that $|I'_{\ph}|, |J'_{\ph}| \leq \Floor{\frac{1}{2}(2(K-1)D-r)} < \frac{KD}{2}$ and $|H'_{\ph}| \leq \Floor{\frac{1}{2}(r-(K-2)D)}$.
    We now study the minimum vertex separator $S$:
    \begin{enumerate}
        \item One of $\square{s}$ or $\square{t}$ is in $S$ (\WLOG assume $\square{s}\in S$ and $\square{t} \notin S$):
        in this case $\{\square{t}\} \cup J' \cup H' \cup I'_{\ph}$ can be outside of $S$ and $|J' \cup H' \cup I'_{\ph}| < KD + KD/2$.
        Here it doesn't matter whether $H'$ has phantom edges, as they contribute the same to the norm bound.
        Thus, we have a bound of $\wt{O}(\sqrt{m} (\sqrt{\ell})^{3KD/2}) = \wt{O}(\sqrt{m} \ell^{3KD/4})$.

        \item $\square{s}, \square{t}\notin S$: $I'\setminus I'_{\ph}, J'\setminus J'_{\ph}$ and $H'\setminus H'_{\ph}$ must be in $S$, so $\{\square{s},\square{t}\} \cup I'_{\ph} \cup J'_{\ph} \cup H'_{\ph}$ can be outside of $S$, and $H'_{\ph}$ is considered isolated in Proposition~\ref{prop:graph-matrix-norm}.
        This gives a bound of $\wt{O}(m(\sqrt{\ell})^{\floor{2(K-1)D-r}} \cdot \ell^{\floor{\frac{1}{2}(r-(K-2)D)}}) \leq \wt{O}(m \ell^{KD/2})$.
    \end{enumerate}
    When $m \leq (\frac{\ell}{\polylog(\ell)})^{KD/2}$, the bounds above are all $o(\ell^{KD})$.

    Therefore, we can write
    $N^\top N = P - W + E$
    where $P \succeq 0$, $W$ is the matrix from Definitiom~\ref{def:W-matrix} and $\spec{E} = o(\ell^{KD})$. This completes the proof.
\end{proof}

\subsubsection{Proof of Lemma~\ref{lem:singular-value-of-L} and \ref{lem:singular-value-of-L-higher-deg}: singular value of \texorpdfstring{$L$}{L}}
\label{sec:sval-U-equals-V}

Lemma~\ref{lem:singular-value-of-L} is a special case ($K=2$) of Lemma~\ref{lem:singular-value-of-L-higher-deg}, which we restate below.
The reader is encouraged to focus on the $K=2$ case for simplicity.

\begin{lemma}[Restatement of Lemma~\ref{lem:singular-value-of-L-higher-deg}]
    Fix $K, D\in \N$, and let $m,\ell,n\in \N$ such that $m \ell^{2(K-1)D} \leq (\frac{n}{\polylog(n)})^{2D}$.
    Let $L$ be the $r_D \times n^{2D}$ matrix defined in Definition~\ref{def:L-matrix-high-degree}.
    Then, with probability $1 - n^{-\Omega(D)}$, $\sigma_{r_D}(L) \geq \Omega(n^D)$.
\end{lemma}

We first recall some definitions.
In Definition~\ref{def:generalized-bucket-profile}, we defined a bucket profile to be $\gamma = (\gamma_1, \gamma_2,\dots, \gamma_K)$ such that $\sum_{i=1}^K i\gamma_i = 2D$ and $\deg(\gamma) = \sum_{i=1}^K (K-i)\gamma_i = K(\sum_{i=1}^K \gamma_i)-2D$, and defined $\Gamma_{2D}$ to be the set of bucket profiles.
Note that $\deg(\gamma) \leq 2(K-1)D$ and is achieved when $\gamma = \gamma^* = (2D,0,\dots,0)$.
For each $\gamma\in \Gamma_{2D}$, $L_{\gamma}$ is a $m\ell_{\deg(\gamma)} \times n^{2D}$ matrix where the rows are indexed by $t\in[m]$ and multiset $J\in [\ell]^{\deg(\gamma)}$, and each row is the flattened vector of a $2D$-th order tensor of dimension $n$ (Definition~\ref{def:L-matrix-high-degree}).
Let $r_D = m \sum_{\gamma\in \Gamma_{2D}} \ell_{\deg(\gamma)}$;
the $L$ matrix is the $r_D \times n^{2D}$ matrix formed by concatenating the rows of $L_{\gamma}$ for all $\gamma \in \Gamma_{2D}$.

Furthermore, recall that $M\in \R^{n\times \ell}$ is a restriction matrix (Definition~\ref{def:restriction-matrix}) whose columns consist of standard unit vectors.
In this section, we can assume without loss of generality that $M$'s columns are $e_1,\dots, e_\ell$.
Then, the $((s,I), H)$ entry of $L_{\gamma}$, where $t\in[m]$, $I\in [\ell]^{\deg(\gamma)}$ is a multiset and $H \in [n]^{2D}$ is an \emph{ordered tuple}, is
\begin{equation*}
    L_{\gamma}[(s,I), H] =
    \sum_{\pi_1\in \bbS_{2D}} \sum_{\pi_2\in\bbS_{\deg(\gam)}}
    \prod_{i\in[K]} \prod_{S,T\in \kappa_i\| \wt{\kappa}_i(\gamma)} A_t[H_{\pi_1(S)}, I_{\pi_2(T)}] \mcom
\end{equation*}
where we recall that $\kappa_i\| \wt{\kappa}_i(\gamma) = ((S_1,T_1),\dots,(S_{\gamma_i}, T_{\gamma_i}))$ such that $(S,T)\in \kappa_i\| \wt{\kappa}_i(\gamma)$ satisfy $|S| = i$ and $|T| = K-i$.
Note also that $\{S\in \kappa_i(\gamma)\}_{i\in[K]}$ forms a partition of $\{1,2,\dots,2D\}$ since $\sum_{i=1}^K \sum_{S\in \kappa_i(\gamma)} |S| = \sum_{i=1}^K i\gamma_i = 2D$.
Similarly, $\{T\in \wt{\kappa}_i\}_{i\in[K]}$ forms a partition of $\{1,2,\dots,\deg(\gamma)\}$ since $\sum_{i=1}^K \sum_{T\in\wt{\kappa}_i(\gamma)} |T| = \sum_{i=1}^K (K-i)\gamma_i = \deg(\gamma)$.
These parameters are consistent with $|H| = 2D$ and $|I| = \deg(\gamma)$.

\paragraph{Graph matrices that arise from $LL^\top$}
$LL^\top$ has a block structure indexed by bucket profiles $\gamma,\gamma' \in \Gamma_{2D}$ where the $(\gamma,\gamma')$ block is $L_{\gamma} L_{\gamma'}^\top$.
For $s,t\in[m]$ and $I\in[\ell]^{\deg(\gamma)}$, $J\in [\ell]^{\deg(\gamma')}$, the entry $L_{\gamma} L_{\gamma'}[(s,I), (t,J)] = \sum_{H\in[n]^{2D}} L_{\gamma}[(s,I), H] L_{\gamma'}[(t,J), H]$.
We now describe the graph matrices that arise here.

\begin{definition}[Graph matrices from $LL^\top$]
\label{def:LL-graph-matrices}
    Fix bucket profiles $\gamma, \gamma'\in \Gamma_{2D}$, permutations $\pi_1,\pi_1'\in \bbS_{2D}$, $\pi_2 \in \bbS_{\deg(\gamma)}$ and $\pi_2'\in \bbS_{\deg(\gamma')}$, the graph matrix is a tripartite graph described as follows,
    \begin{enumerate}
        \item On the left, $U_\al$ contains a square vertex $\square{s}$ (that takes a label in $[m]$) and $\deg(\gamma)$ circle vertices $I = \{\circle{a_1},\dots, \circle{a_{\deg(\gamma)}}\}$ (that take labels in $[\ell]$).

        \item On the right, $V_\al$ contains a square vertex $\square{t}$ and $\deg(\gamma')$ circle vertices $J = \{\circle{b_1},\dots, \circle{b_{\deg(\gamma')}}\}$.

        \item In the middle, there are $2D$ hexagon vertices: $H = (\hexagon{c_1}, \dots, \hexagon{c_{2D}})$ (that take labels in $[n]$).

        \item There are $K$ types of hyperedges: for $i\in[K]$, a type-$i$ edge connects a square vertex, $K-i$ circle vertices and $i$ hexagon vertices:
        \begin{itemize}
            \item on the left, for each $i\in[K]$ there are $\gamma_i$ type-$i$ hyperedges:
            for each $(S,T)\in \kappa_i \| \wt{\kappa}_i(\gamma)$ we add an edge $\{\square{s}\} \cup \{\circle{a_j}\}_{j\in\pi_2(T)} \cup \{\hexagon{c_k}\}_{k\in \pi_1(S)}$ (recall that $|S| = i$ and $|T| = K-i$);

            \item on the right, for each $i\in[K]$ there are $\gamma_i'$ type-$i$ hyperedges:
            for each $(S,T)\in \kappa_i \| \wt{\kappa}_i(\gamma')$ we add an edge $\{\square{t}\} \cup \{\circle{a_j}\}_{j\in\pi_2(T)} \cup \{\hexagon{c_k}\}_{k\in \pi_1(S)}$.
        \end{itemize}
        In total, there are $\sum_{i=1}^K (\gamma_i + \gamma_i')$ hyperedges.

        \item There may be ``surprise'' collapses which we define as follows: hexagon vertices in $H$ may collide with circle vertices in $U_{\al}$ or $V_{\al}$ as both the circle and hexagon vertices index the same matrices $A_1,\dots,A_m$; we call each such shape a surprise collapsed shape.
    \end{enumerate}
\end{definition}

For intuition, the following diagram illustrates an example shape that shows up in the case of $D=1$ and $K=2$.
This particular shape comes from bucket profiles
$\gamma = (2,0)$ and $\gamma' = (0,1)$ (recall that $\sum_{i=1}^K i\gamma_i = 2D$), and it has two type-1 edges corresponding to the random variables $A_s[i,c_1] A_s[j,c_2]$ and a type-2 edge corresponding to the term $A_t[c_1,c_2]$.
\begin{figure}[h]
    \centering
    \includegraphics[width=180pt]{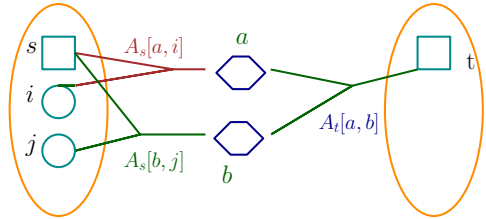}
    \caption{$LL^\top[(s,i,j), (t)] = \sum_{a\neq b\in [n]} A_s[a,i]A_s[b,j]A_t[a,b] $ for $i\neq j\in [\ell]$, and $s\neq t \in [m]$.}
    \label{fig:LL-shape}
\end{figure}

\paragraph{Proof of Lemma~\ref{lem:singular-value-of-L-higher-deg}}

\begin{proof}[Proof of Lemma~\ref{lem:singular-value-of-L-higher-deg}]
    We would like to prove that $\lambda_{\min}(LL^\top) \geq \Omega(n^{2D})$, hence by Fact~\ref{fact:deleting-rows-sval} it suffices for us to prove a singular value lower bound of $L$ with some columns removed. In particular, since the columns of $L$ are indexed by ordered tuples $H \in [n]^{2D}$, we will restrict to our attention to the columns where $H$ has distinct elements.
    This means that we do not have any collapses among the $2D$ hexagon vertices in the graph matrices from Definition~\ref{def:LL-graph-matrices}.

    It suffices for us to prove that $LL^\top$ has large diagonal entries, and then complete the proof by upper bounding the spectral norm of the off-diagonal components.
    Note that there are two types of shapes in the off-diagonal: $\square{s}= \square{t}$ and $\square{s}\neq \square{t}$.
    We will analyze them separately.

    \subparagraph{Diagonal entries.}
    For any bucket profile $\gam$, the diagonal entries of $L_{\gam} L_{\gam}^\top$ correspond to the shapes from Definition~\ref{def:LL-graph-matrices} where $U_\al = V_\al$.
    The dominating term is when all edges are phantom edges (double-edges; recall Definition~\ref{def:phantom-edge}), leaving all $2D$ hexagon vertices isolated.
    All such phantom edges are in fact square terms $A_s[i,c]^2$, hence by Gaussian concentration we have
    \[ 
    L_{\gam}L_{\gam}^\top[(s,I), (s,I)] \geq \Omega( n^{2D})
    \]
    for any of its indices $(s,I)$ where $s\in [m]$ and $I\in [\ell]^{\deg(\gamma)}$.

    \subparagraph{Off-diagonal entries, $s\neq t$.}
    Fix bucket profiles $\gamma,\gamma'\in \Gamma_{2D}$, we analyze the graph matrices from Definition~\ref{def:LL-graph-matrices} with $s\neq t$.
    In this case, there is no isolated vertex, and since the hexagon vertices don't collapse, there is no phantom edge even if any circle vertices collapse.
    Recall that for each $i\in[K]$ there are $\gamma_i$, $\gamma_i'$ type-$i$ edges on the left and right side, respectively.
    Furthermore, note that all $\hexagon{c}\in H$ are connected to both $\square{s}$ and $\square{t}$.

    We now analyze the minimum vertex separator $S$.
    First, for $\hexagon{c}\in H$, we define $N_L(\hexagon{c})\subseteq I$ and $N_R(\hexagon{c})\subseteq J$ to be the set of circle vertices incident to $\hexagon{c}$ in $I$ and $J$, respectively (they can be empty).
    Observe that there are paths $\circle{a} \to \hexagon{c} \to \circle{b}$ for all $\circle{a}\in N_L(\hexagon{c})$ and $\circle{b}\in N_R(\hexagon{c})$,
    First, unless $\square{s}, \square{t}\notin S$ (in which case $H$ is forced to be in $S$; we will see this case later), since $\ell^{K-1} \ll n$, it is always lower weight to pick the circle vertices instead of $\hexagon{c}$.  Thus, we will never include hexagon vertices in $S$.
    Moreover, observe that we must either include all of $N_L(\hexagon{c})$ or all of $N_R(\hexagon{c})$ in the minimum vertex separator, as excluding part of $N_L(\hexagon{c})$ will force all of $N_R(\hexagon{c})$ to be in $S$ and vice versa.

    Recall that by removing certain columns of $L$, we can assume that no hexagon vertices collapse.
    Let us first assume that no circle vertices between $I$ and $J$ collapse and there is no surprise collapse.
    We consider the following cases,
    \begin{enumerate}
        \item $\square{s}, \square{t}\notin S$: for every $\hexagon{c}\in H$, there is a path $\square{s} \rightarrow \hexagon{c} \rightarrow \square{t}$.
        Thus, since $\square{s},\square{t}\notin S$, all hexagon vertices in $H$ must be in the vertex separator, and $\{\square{s},\square{t}\} \cup I \cup J$ can be outside of $S$.
        $|I\cup J| \leq \deg(\gamma) + \deg(\gamma') \leq 4(K-1)D$, giving a bound of $\wt{O}(m(\sqrt{\ell})^{4(K-1)D}) = \wt{O}(m\ell^{2(K-1)D})$.

        \item One of $\square{s}$ or $\square{t}$ is in $S$ (\WLOG assume $\square{s}\in S$ and $\square{t}\notin S$): for each $\hexagon{c}\in H$ incident to a non-type-$K$ edge on the left, there is a path $\circle{a}\rightarrow \hexagon{c} \rightarrow \square{t}$ where $\circle{a}\in U_\al$.
        For such a path, we must include $\circle{a}$ in $S$.
        Thus, we have that $I$ must be in $S$, and $\{\square{t}\} \cup J \cup H$ can be outside of $S$, giving a bound of $\wt{O}(\sqrt{m} (\sqrt{\ell})^{\deg(\gamma')} (\sqrt{n})^{2D}) \leq \wt{O}(\sqrt{m}\ell^{(K-1)D} n^D)$ since $\deg(\gamma') \leq 2(K-1)D$.

        \item $\square{s},\square{t}\in S$: 
        in this case, as discussed earlier we will never include any hexagon vertex in $S$.
        For each $\hexagon{c} \in H$, we have $|N_L(\hexagon{c})|$ and $|N_R(\hexagon{c})| \leq K-1$.
        From the earlier discussion, we must include either all of $N_L(\hexagon{c})$ or all of $N_R(\hexagon{c})$ in $S$.
        Thus, each $\hexagon{c}\in H$ may introduce at most $K-1$ circle vertices outside of $S$, in total at most $(K-1)\cdot |H| = 2(K-1)D$.
        This gives a norm bound of $\wt{O}((\sqrt{\ell})^{2(K-1)D} (\sqrt{n})^{2D}) = \wt{O}(\ell^{(K-1)D} n^D)$.
    \end{enumerate}
    When $m\ell^{2(K-1)D} \leq (\frac{n}{\polylog(n)})^{2D}$, the bounds above are all $o(n^{2D})$, negligible compared to the diagonal.

    Now, we consider shapes with collapsed circle vertices between $I$ and $J$.
    When such collapses occur, then each collapsed circle vertex becomes a vertex in $U_\al \cap V_\al$, the mandatory vertex separator.
    Since the hexagon vertices don't collapse, such collapses won't introduce any phantom edges, hence they only decrease the spectral norm.

    \subparagraph{Off-diagonal entries, $s = t$.}
    In this case, we may have phantom edges which result in isolated hexagon vertices in $H$ (recall Definition~\ref{def:phantom-edge}).
    Let $H_{\iso}$ be the set of isolated vertices in $H$.
    Let's assume that there is no surprise collapse.

    Since $\square{s}=\square{t} \in U_\al\cap V_\al$, it must be in the vertex separator $S$, thus our analysis is similar to Case 3 of the previous analysis.
    We define $N_L(\hexagon{c})$ and $N_R(\hexagon{c})$ as before.
    For each $\hexagon{c}\in H_{\iso}$, by definition of phantom edges it must be that $N_L(\hexagon{c}) = N_R(\hexagon{c})$, i.e.\ they collapse and are thus in $U_\al \cap V_\al$, the mandatory separator.
    For each $\hexagon{c} \in H \setminus H_{\iso}$, since $|N_R(\hexagon{c})|$, $|N_R(\hexagon{c})| \leq K-1$, by the same analysis $\hexagon{c}$ can introduce at most $K-1$ circle vertices outside of $S$.
    This gives a norm bound of $\wt{O}(n^{|H_{\iso}|} \cdot (\sqrt{n})^{2D-|H_{\iso}|} \cdot (\sqrt{\ell})^{(K-1)(2D-|H_{\iso}|)} ) = \wt{O}( n^D (n\ell^{-K+1})^{\frac{1}{2}|H_{\iso}|} \ell^{(K-1)D})$.

    Fortunately, $H_{\iso} \subsetneq H$ (otherwise $U_\al = V_\al$ which is the diagonal) hence $|H_{\iso}| \leq 2D-1$.
    Since $\ell^{K-1} \leq o(n)$, the bound is maximized when $|H_{\iso}| = 2D-1$, giving a bound of $\wt{O}(n^{2D-\frac{1}{2}} \ell^{\frac{K-1}{2}})$, which is $o(n^{2D})$ since $\ell^{K-1} \leq o(n)$.
    
    \subparagraph{Handling surprise collapses.}
    We first remind the reader where surprise collapses may arise (i.e., within the left/right shape). This happens when a hexagon vertex (which takes labels in $[n]$) takes a label in $[\ell]$ hence merge with a circle vertex.

    We now give a charging argument for surprise collapsed shape to bound their norm by the shape with $1$ fewer surprise collapse, and then combining with our bounds for uncollapsed shapes completes the proof. Observe that in our case, each collapse may only decrease the norm bound, as it essentially reduces a factor from $n$ to $\ell$ in the norm bound (in some cases it removes the $n$ factor entirely due to the collapsed vertex being in the separator).
    Moreover, since there are always $2D$ distinct vertices in $H$, the collapse does not introduce (new) phantom edges that may produce extra isolated vertices. Hence, each surprise collapse may only decrease the norm. 

    \subparagraph{Completing the proof.}
    We have shown that the diagonal entries of $LL^\top$ is $\Omega(n^{2D})$, while all other shapes in the off-diagonal has spectral norm $o(n^{2D})$ when $m\ell^{2(K-1)D} \leq (\frac{n}{\polylog(n)})^{2D}$.
    Thus, we can write
    \begin{equation*}
        LL^\top = \Omega(n^{2D})\cdot \Id + E
    \end{equation*}
    where $\spec{E} = o(n^{2D})$.
    This completes the proof.
\end{proof}

\subsection{Singular value lower bounds for desymmetrization}
\label{sec:sval-desymm}

\begin{lemma}[Restatement of Lemma~\ref{lem:desymm-singular-value-hd}]
    Let $m,\ell, D \in \N$ such that $m \leq (\frac{\ell}{\polylog(\ell)})^{2D}$.
    Then, with probability $1- \ell^{-\Omega(D)}$,
    \begin{equation*}
        \sigma_{m_3} \Paren{\Sym_{6D} \cdot C^{\otimes 3}_{uniq}} \geq \Omega(\ell^{3D}) \mper
    \end{equation*}
\end{lemma}

We first recall from Section~\ref{sec:desymm-hd} that $C$ is the $\ell_{2D} \times m$ matrix whose columns are the coefficient vectors of $B_1^D, \dots, B_m^D$, and $C^{\ot 3}_{uniq}$ is the $(\ell_{2D})^3 \times m_3$ matrix with columns $C_i \ot C_j \ot C_k$ for $i\leq j\leq k\in [m]$, the ``unique'' columns of $C^{\ot 3}$.
Then, $\Sym_{6D} \cdot C^{\ot 3}_{uniq}$ is an $\ell_{6D} \times m_3$ matrix where we symmetrize each column of $C^{\ot 3}_{uniq}$.

We let $H \coloneqq \Sym_{6D} \cdot C^{\ot 3}_{uniq}$ for convenience.
Fix multiset $I \in [\ell]^{6D}$ and let $I_1, I_2, I_3$ be the ordered partition of $I$, each of size $2D$.
Fix $t_1 \leq t_2 \leq t_3 \in [m]$, we have
\begin{equation*}
\begin{aligned}
    H[I, (t_1,t_2,t_3)]
    &= \Theta(1)\sum_{\pi\in S_{6D}} \Sym(B_{t_1}^{\otimes D})[\pi(I_1)] \cdot \Sym(B_{t_2}^{\otimes D})[\pi(I_2)] \cdot \Sym(B_{t_3}^{\otimes D})[\pi(I_3)] \\
    &= \Theta(1)\sum_{\pi\in S_{6D}} (B_{t_1}^{\otimes D})[\pi(I_1)] \cdot (B_{t_2}^{\otimes D})[\pi(I_2)] \cdot (B_{t_3}^{\otimes D})[\pi(I_3)]
    \mper
\end{aligned}
\end{equation*}
We can thus define the graph matrices that arise from $H^\top H$.

\begin{definition}[Graph matrices from $H^\top H$]
\label{def:HH-graph-matrices}
    Fix permutations $\pi, \pi' \in \bbS_{6D}$, the graph matrix can be described as follows,
    \begin{enumerate}
        \item A tripartite graph with the following 6 square vertices (that take labels in $[m]$), $6D$ circle vertices (that take labels in $[\ell]$) and $6D$ edges:
        \begin{itemize}
            \item $U_\al = \{\square{s_1},\square{s_2}, \square{s_3}\}$;

            \item $V_\al = \{\square{t_1},\square{t_2}, \square{t_3}\}$;

            \item $W = (\circle{i_1}, \circle{i_2}, \dots, \circle{i_{6D}})$.
        \end{itemize}

        \item On the left, for each $j\in \{1,2,3\}$, we add $D$ edges: $(\square{s_j}, \circle{i_{\pi(2k-1)}}, \circle{i_{\pi(2k)}})$ for $k = (j-1)D + 1,\dots, jD$.
        Similarly on the right, for each $j\in \{1,2,3\}$, we add $D$ edges: $(\square{t_j}, \circle{i_{\pi'(2k-1)}}, \circle{i_{\pi'(2k)}})$ for $k = (j-1)D + 1,\dots, jD$.
        In total $6D$ edges.
    \end{enumerate}
    All other components of $H^\top H$ can be viewed as collapses of the graph matrices defined above.
\end{definition}

\subparagraph{Diagonal entries of $H^\top H$.}
We show that the diagonal entries are $\Omega(\ell^{6D})$ in the following lemma, whose proof is almost identical to Claim~\ref{claim:norm-of-polynomial-powers} and follows from standard Gaussian concentration results.

\begin{lemma}
\label{lem:desymm-diagonal-entries}
    Fix $D, K\in \N$.
    Let $B_1(y), B_2(y), B_3(y)$ be degree-$K$ homogeneous polynomials in $\ell$ variables such that the coefficients of the $B_i$s are sampled i.i.d.\ from $\calN(0,1)$.
    Then, with probability at least $1- \ell^{-\Omega(DK)}$,
    \begin{equation*}
    \Norm{B_{1}(y)^D B_{2}(y)^D B_{3}(y)^D}^2 = \Theta(\ell^{3DK}) \mper
    \end{equation*}
\end{lemma}
\begin{proof}
    Similar to the proof of Claim~\ref{claim:norm-of-polynomial-powers}, in expectation over the coefficients of $B_1,B_2,B_3$,
    \begin{equation*}
        \E_B \Norm{B_1^D B_2^D B_3^D}^2
        = \Theta(1) \sum_{I\in [\ell]^{3DK}} \E_B \Paren{B_1^D B_2^D B_3^D [I]}^2 = \Theta(\ell^{3DK}) \mper
    \end{equation*}
    The statement of the lemma follows by standard Gaussian concentration on low-degree polynomials of Gaussians~\cite{SS12}.
\end{proof}

\subparagraph{Completing the proof of Lemma~\ref{lem:desymm-singular-value-hd}.}
Lemma~\ref{lem:desymm-diagonal-entries} shows that the diagonal entries of $H^\top H$ are $\Omega(\ell^{6D})$.
We now complete the proof by bounding the norm of the off-diagonal components.

\begin{proof}[Proof of Lemma~\ref{lem:desymm-singular-value-hd}]
    It suffices to prove that $\lambda_{\min}(H^\top H) \geq \Omega(\ell^{6D})$.
    By Fact~\ref{fact:deleting-rows-sval}, it sufﬁces for us to prove a singular value lower bound of $H$ with some rows removed.
    Recall that the rows of $H$ are indexed by multisets $I \in [\ell]^{6D}$.
    We will restrict to the rows of $H$ whose indices have distinct elements.
    This means that the $6D$ circle vertices in the graph matrices from Definition~\ref{def:HH-graph-matrices} do not collapse.

    We split the matrix into diagonal and off-diagonal components.
    For the diagonal entries, by Lemma~\ref{lem:desymm-diagonal-entries}, the diagonal contributes $\Omega(\ell^{6D}) \cdot \Id$.
    For the off-diagonal components, we will bound the spectral norm of the graph matrices from Definition~\ref{def:HH-graph-matrices}.

    Let's start by considering the case $U_\al\cap V_\al = \emptyset$, in which case either $U_\al$ or $V_\al$ must be the minimum vertex separator given $m \ll \ell^{2D}$, thus giving a bound of $\wt{O}(\sqrt{\ell^{6D} m^3})$.

    We now consider the shapes where $U_\al$ and $V_\al$ collapse.
    Since each square vertex is connected to $2D$ circle vertices, merging two square vertices might produce phantom edges, making $2D$ circle verties isolated in the worst case.
    On the other hand, for such a collapse, two square vertices merge into $U_\al \cap V_\al$, the mandatory vertex separator.
    Thus, such a collapse increases the spectral norm by a factor of $\sqrt{\ell^{2D}/m}$.

    There can only be two collapses between $U_\al$ and $V_\al$, otherwise $U_\al = V_\al$ which becomes the diagonal.
    Thus, we can bound the spectral norm by $\wt{O}(\sqrt{\ell^{6D} m^3} \cdot \ell^{2D}/m) = \wt{O}(\ell^{5D} \sqrt{m}) = o(\ell^{6D})$ when $m \leq (\frac{\ell}{\polylog(\ell)})^{2D}$.

    Wrapping up, the above shows that
    \begin{equation*}
        H^\top H = \Omega(\ell^{6D}) \cdot \Id + E
    \end{equation*}
    with $\|E\|_2 \leq o(\ell^{6D})$.
    This completes the proof.
\end{proof}

\subsection{Singular value lower bounds for \texorpdfstring{$V$}{V} of  higher degree polynomials}
\label{sec:sval-analysis-of-V-high-deg}

\begin{lemma}[Restatement of Lemma~\ref{lem:VV-NN-high-deg}]
\label{lem:VV-NN-sec6}
    Fix $K,D\in \N$ and $K \geq 3$.
    Let $m, \ell \in \N$ such that $m \leq (\frac{\ell}{\polylog(\ell)})^{KD/2}$.
    Let $\ol{V}$, $\ol{N}$ be the matrices defined in Definition~\ref{def:V-bar-matrix-high-deg} and \ref{def:N-bar-matrix-high-deg}.
    Then with probability $1-\ell^{-\Omega(D)}$,
    \begin{equation*}
        \lambda_{\min}\Paren{\ol{V}^\top \ol{V} + \ol{N}^\top \ol{N}}
        \geq \Omega(\ell^{KD}) \mper
    \end{equation*}
\end{lemma}

Recall that $\ol{V}$ and $\ol{N}$ are generalizations of the $V$ and $N$ matrices for the case of cubics of quadratics, with additional matrices padded to $V$ and $N$: $\ol{V} = [V\mid V_G]$ and $\ol{N} = \begin{bmatrix} N & 0 \\ N_G & N_B \end{bmatrix}$.
Thus,
\begin{equation}
\label{eq:VV-NN-bar-structure}
    \ol{V}^\top \ol{V} =
    \begin{bmatrix}
        V^\top V & V^\top V_G \\
        V_G^\top V & V_G^\top V_G
    \end{bmatrix}
    ,\quad
    \ol{N}^\top \ol{N} =
    \begin{bmatrix}
        N^\top N + N_G^\top N_G & N_G^\top N_B \\
        N_B^\top N_G & N_B^\top N_B
    \end{bmatrix}
    \mper
\end{equation}

\paragraph{Proof outline}
We will follow our usual strategy of pulling out the diagonal matrix and bound the norm of the off-diagonal blocks.

In Section~\ref{sec:singular-value-of-VV-NN}, we have already shown in Lemma~\ref{lem:VV-NN-generalized} that $V^\top V + N^\top N \succeq \Omega(\ell^{KD})\cdot \Id$.
Thus, we proceed to analyze the rest of the matrix blocks.

For $V^\top V_G$ and $V_G^\top V_G$, it turns out that all graph matrices have negligible spectral norms compared to the diagonal, except for one specific matrix deﬁned as follows.

\begin{definition}[$\ol{W}$ matrix]
\label{def:W-bar-matrix}
    Let $W$ be the matrix with the same dimension as $V^\top V_G$ such that $W[(s,I), J] = V^\top V_G[(s,I), J] = \ell^{-\frac{(K-2)D}{2}} G[I] B_s^D[J]$ if $I \cap J = \emptyset$, and $0$ otherwise.
\end{definition}

\begin{lemma} \label{lem:VG-VG}
    Let $K,D,m,\ell\in\N$ be the same parameters as Lemma~\ref{lem:VV-NN-sec6}, then
    \begin{equation*}
        V_G^\top V_G = \Omega(\ell^{KD})\cdot \Id + E_1 \mcom \quad
        V^\top V_G = \ol{W} + E_2
    \end{equation*}
    where $\spec{E_1}, \spec{E_2}\leq o(\ell^{KD})$.
\end{lemma}

Recall that in the analysis of $V^\top V + N^\top N$ in Section~\ref{sec:singular-value-of-VV-NN}, we saw that there is a large norm matrix $W$ in $V^\top V$ that is canceled by $N^\top N$.
Here, the same phenomenon occurs: the block $N_G^\top N_B$ has a $-\ol{W}$ component that cancels out the $\ol{W}$ in $V^\top V_G$.

\begin{lemma} \label{lem:NG-NB}
    Let $K,D,m,\ell\in\N$ be the same parameters as Lemma~\ref{lem:VV-NN-sec6}, then
    \begin{equation*}
        N_G^\top N_B = -\ol{W} + E_3
    \end{equation*}
    where $\spec{E_3} \leq o(\ell^{KD})$.
\end{lemma}

We can now complete the proof.
\begin{proof}[Proof of Lemma~\ref{lem:VV-NN-sec6}]
    By Lemma~\ref{lem:VV-NN-generalized}, we have $V^\top V + N^\top N \succeq \Omega(\ell^{KD})\cdot \Id$, and by Lemma~\ref{lem:VG-VG}, we have $V_G^\top V_G \succeq \Omega(\ell^{KD})\cdot \Id$.
    Moreover, combining Lemmas~\ref{lem:VG-VG}, \ref{lem:NG-NB} and \pref{eq:VV-NN-bar-structure}, we have
    \begin{equation*}
        \ol{V}^\top \ol{V} + \ol{N}^\top \ol{N} \succeq
        \begin{bmatrix}
            \Omega(\ell^{KD}) \cdot \Id & E_2 + E_3 \\
            E_2^\top + E_3^\top & \Omega(\ell^{KD}) \cdot \Id \\
        \end{bmatrix}
    \end{equation*}
    where $\spec{E_2},\spec{E_3} = o(\ell^{KD})$.
    This completes the proof.
\end{proof}

\paragraph{Graph matrices that arise from $V^\top V_G$}
The following definition describes the graph matrices from $V^\top V_G$.
We will see later that the graph matrices arising from $V_G^\top V_G$ and $N_G^\top N_B$ are all slight modifications of these shapes.

\begin{definition}[Graph matrices from $V^\top V_G$]
\label{def:VVG-graph-matrices}
    Fix $r = |I\cap J| \in \{0,1,\dots, KD\}$ and a permutation $\pi\in \bbS_{KD}$, the graph matrices are described as follows,
    \begin{enumerate}
        \item A graph with a square vertex $\square{s}$ (that take labels in $[m]$) and $(3K-2)D$ circle vertices (that take labels in $[\ell]$):
        \begin{itemize}
            \item $I' = (\circle{a_1},\dots, \circle{a_{|I'|}})$ where $|I'| = 2(K-1)D-r$;
            \item $J' = (\circle{b_1},\dots, \circle{b_{|J'|}})$ where $|J'| = KD-r$;
            \item $H = (\circle{b_{|J'|+1}}, \dots, \circle{b_{KD}})$, and $|H| = r$;
            \item $L = (\circle{c_{1}}, \dots, \circle{c_{r}})$.
        \end{itemize}

        \item $U_\al = \{\square{s}\} \cup I'\cup L$ and $V_\al = J' \cup L$.

        \item There are $D$ hyperedges:
        $(\square{s}, \circle{b_{\pi((i-1)K+1)}}, \dots, \circle{b_{\pi(iK)}})$ for $i=1,\dots, D$, each touching $\square{s}$ and $K$ circle vertices;
        we add another hyperedge $I'\cup H$, touching $2(K-1)D$ circle vertices.
    \end{enumerate}
\end{definition}

\paragraph{Analysis of $V_G^\top V_G$ and $V^\top V_G$: Proof of Lemma~\ref{lem:VG-VG}}

\begin{proof}[Proof of Lemma~\ref{lem:VG-VG}]
    We first start with $V^\top V_G$, the upper right submatrix of $\ol{V}^\top \ol{V}$.
    Recall from Definition~\ref{def:V-bar-matrix-high-deg} that
    the columns of $V$, indexed by $s\in[m]$ and multisets $I\in [\ell]^{2(K-1)D}$, are coefficient vectors of the degree-$(3K-2)D$ polynomials $B_s(y)^D y_I$.
    Moreover, the columns of $V_G$, indexed by $J\in [\ell]^{KD}$, are coefficient vectors of $\ell^{-\frac{(K-2)D}{2}} G(y) y_J$, where $G(y)$ is a degree-$2(K-1)D$ polynomial with i.i.d.\ Gaussian coefficients (note the scaling factor here).
    The entry of $V^\top V_G$ is
    \begin{equation*}
        V^\top V_G[(s,I),J] = \ell^{-\frac{(K-2)D}{2}} \Iprod{B_s(y)^D y_I, G(y) y_J} \mper
    \end{equation*}
    When $I\cap J = \emptyset$, the above is simply $V^\top V_G[(s,I), J] = \ell^{-\frac{(K-2)D}{2}} \cdot G[I] B_s^D[J]$, which is exactly the $\ol{W}$ matrix from Definition~\ref{def:W-bar-matrix}.

    We now analyze the minimum vertex separator $S$.
    If $\square{s}\in S$, then $I'\cup J' \cup H$ can be outside $S$, giving a bound of $\wt{O}((\sqrt{\ell})^{(3K-2)D-r})$ (here possible phantom edges from collapses within $J'$ don't matter as they contribute the same factors).
    Multiplied by the $\ell^{-\frac{(K-2)D}{2}}$ factor, we have a bound of $\wt{O}(\ell^{KD - r/2}) = o(\ell^{KD})$ since $r>0$.

    If $\square{s}\notin S$, same as the proof of Lemma~\ref{lem:VV-off-diag-blocks}, we define $J_{\ph}'$ to be the vertices in $J'$ that are disconnected from $\square{s}$ due to phantom edges.
    We know that $|J'_{\ph}| \leq \Floor{\frac{1}{2}(KD-r)} < KD/2$.
    Then, $J\setminus J'_{\ph}$ must be in $S$, and $\{\square{s}\} \cup I' \cup J'_{\ph} \cup H$ can be outside of $S$, giving a bound of $\wt{O}(\sqrt{m}(\sqrt{\ell})^{(\frac{5}{2}K-2)D})$.
    Multiplied by the $\ell^{-\frac{(K-2)D}{2}}$ factor, we have a bound of $\wt{O}(\sqrt{m}\ell^{3KD/4})$, which is $o(\ell^{KD})$ when $m \leq (\frac{\ell}{\polylog(\ell)})^{KD/2}$.
    This proves that $V^\top V_G = \ol{W} + E$ where $\spec{E} \leq o(\ell^{KD})$.

    \subparagraph{Analysis of $V_G^\top V_G$.}
    We write out the entries of $V_G^\top V_G$:
    for multisets $I,J\in [\ell]^{KD}$,
    \begin{equation*}
        V_G^\top V_G[I,J] = \ell^{-(K-2)D} \Iprod{G(y)y_I, G(y) y_J} \mper
    \end{equation*}

    The diagonal entries $V_G^\top V_G[I,I] = \ell^{-(K-2)D}\Norm{G}^2 = \Omega(\ell^{KD})$ by Claim~\ref{claim:norm-of-polynomial-powers}.
    For the off-diagonal entries of $V_G^\top V_G$, fix $r = |I\cap J|\in \{0,1,\dots, KD-1\}$, the graph matrices that arise are very similar to the ones from $V^\top V_G$ in Definition~\ref{def:VVG-graph-matrices}; we note the small differences below:
    \begin{enumerate}
        \item There is no square vertex, and there are also $(3K-2)D$ circle vertices but with $|I'| = |J'| = KD-r$ and $|H| = (K-2)D-r$.
        \item $U_\al = I'\cup L$ and $V_\al = J' \cup L$.
        \item There are only 2 hyperedges: $I' \cup H$ and $J' \cup H$, each touching $2(K-1)D$ circle vertices.
    \end{enumerate}
    It is clear that when $r < KD$, the minimum vertex separator is just $U_\al \cap V_\al = L$ and there is no phantom edge,
    giving a bound of $\wt{O}((\sqrt{\ell})^{(3K-2)D-r})$.
    Multiplied by the factor $\ell^{-(K-2)D}$, we have a bound of $\wt{O}(\ell^{KD/2+D}) = o(\ell^{KD})$ when $K \geq 3$.
    This completes the proof.
\end{proof}

\paragraph{Analysis of $N_G^\top N_B$: Proof of Lemma~\ref{lem:NG-NB}}

\begin{proof}[Proof of Lemma~\ref{lem:NG-NB}]
    Recall from Definition~\ref{def:N-bar-matrix-high-deg} that
    for $s,t\in[m]$ and $I\in [\ell]^{2(K-1)D}$,
    $N_G[s,(t,I)] = \ell^{-\frac{(K-2)D}{2}} G[I]$ if $s=t$ and 0 otherwise.
    Moreover, for $s\in[m]$ and $J\in [\ell]^{KD}$, $N_B[s,J] = -(B_s(y)^D)[J]$.
    Note the crucial negative sign here.
    Thus,
    \begin{equation*}
        N_G^\top N_B[(t,I), J]
        = \sum_{s\in[m]} N_G[s,(t,I)] \cdot N_B[s,J]
        = - \ell^{-\frac{(K-2)D}{2}} G[I] \cdot B_t^D[J]\mper
    \end{equation*}

    When $I\cap J = \emptyset$, this is exactly $-\ol{W}$ from Definition~\ref{def:W-bar-matrix}.
    When $r = |I\cap J| > 0$, the graph matrices that arise are almost identical to those from $V^\top V_G$ in Definition~\ref{def:VVG-graph-matrices}; we note the differences below:
    \begin{enumerate}
        \item There are only $(3K-2)D - r$ circle vertices: $I', J', H$ are the same, but $L = \emptyset$.
        \item $U_\al = \{\square{s}\} \cup I' \cup H$ and $V_\al = J' \cup H$.
        \item The hyperedges are exactly the same.
    \end{enumerate}
    The analysis of the minimum vertex separator is almost identical to $V^\top V_G$ as well.
    Observe that the graph matrices here are the same as in Definition~\ref{def:VVG-graph-matrices} except that there are no circle vertices outside of $U_\al \cup V_\al$ (here $H \subseteq U_\al \cap V_\al$ and $L=\emptyset$).
    This only decreases the spectral norm as there are strictly fewer vertices outside of the vertex separator when $r > 0$.
    Thus, all such shapes have norm $o(\ell^{KD})$, and this completes the proof.
\end{proof}

\section*{Acknowledgements} P.K.\ thanks Ankit Garg for preliminary discussions on power-sum decomposition. M.B.\ thanks Akash Sengupta for several helpful discussions and for suggesting that a version of Lemma~\ref{lem:V-singular-value-hd} should be true.

\newpage
\bibliographystyle{alpha}
\bibliography{main}

\appendix
\section{Non-Identifiability for Power-Sum Decomposition}

\subsection{Non-identifiability of sum of cubics of linearly independent quadratics}
\label{sec:identifiability-cubics-of-quadratics}
\begin{lemma}
Let $a = \sqrt{6}$ and consider the following distinct sets of bivariate quadratic polynomials in variables $x,y$:
\begin{equation*}
\begin{aligned}
    \calS_1 &= \{x^2 + a xy,\ x^2 + y^2,\ y^2 + axy\} \mcom \\
    \calS_2 &= \{x^2,\ x^2 + axy + y^2,\ y^2\} \mper
\end{aligned}
\end{equation*}
Then, the polynomials in each set have linearly independent coefficient matrices but the sum of cubics of polynomials in either sets is equal.
\end{lemma}
\begin{proof} 
It is easy to verify that in both sets, the coefficient matrices of the polynomials are linearly independent.
The sum of cubics of $\calS_1$ is
\begin{equation*}
    \sum_{p\in \calS_1} p(x,y)^3 = 2x^6 + 3ax^5 y + 3(a^2+1) x^4 y^2 + 2a^3 x^3 y^3 + 3(a^2+1) x^2 y^4 + 3a x y^5+ 2y^6
\end{equation*}
whereas
\begin{equation*}
    \sum_{p\in \calS_2} p(x,y)^3 = 2x^6 + 3ax^5 y + 3(a^2+1) x^4 y^2 + (6a+a^3) x^3 y^3 + 3(a^2+1) x^2 y^4 + 3a x y^5 + 2y^6
\end{equation*}
Thus by setting $a = \sqrt{6}$, we have $2a^3 = 6a +a^3$, meaning $\sum_{p\in \calS_1} p(x,y)^3 = \sum_{p\in \calS_2} p(x,y)^3$.
\end{proof}

\subsection{Non-identifiability of sum of squares of quadratics}
\label{app:quadratic-failure}
We observe that sum-of-squares of even two random homogeneous quadratics cannot be uniquely decomposed. 

\begin{lemma}[Non-Indentifiability of Generic Sum of Quadratics of Quadratics]
Let $A_1, A_2$ be $n \times n$ matrices of independent $\cN(0,1)$ entries up to symmetry. Then, with probability $1$ over the draw of $A_1, A_2$, there exist symmetric $A_1',A_2'$ such that $\Norm{A_i'-A_j}_2\geq 1/n^{O(1)}$ for every $i,j$ such that $(x^{\top} A_1 x)^2 + (x^{\top} A_2 x)^2 = (x^{\top} A_1' x)^2 + (x^{\top} A_2' x)^2$ for every $x$.
\end{lemma}
\begin{proof}
Let $V_1, V_2$ be the vectorization of upper-triangular entries of $A_1, A_2$ respectively. Since the coefficient tensor of $(x^{\top} A_1 x)^2 + (x^{\top} A_2 x)^2$ is a linear transformation (scaling of $\Sym$ operation) applied to $A_1^{\otimes 2} + A_2^{\otimes 2}$, it is enough to find $V_1', V_2'$ distinct form $V_1, V_2$ such that $V_1^{\otimes 2} + V_2^{\otimes 2} = {V_1'}^{\otimes 2} + {V_2'}^{\otimes 2}$. The (since $V_1, V_2$ are random Gaussian, the rank decomposition is unique w.p. $1$) orthogonal decomposition of the matrix $V_1 V_1^{\top} + V_2V_2^{\top}$ uses orthogonal vectors $V_1',V_2'$ that are different from $V_1, V_2$ (in fact must have a distance of at least $1/n \Norm{V_1}_2$). Taking $A_i'$ to be the matrix whose upper triangular entries are given by $V_i'$ for $i=1,2$ completes the proof.
\end{proof}


\section{Analysis of our Algorithm for Generic \texorpdfstring{$A_t$}{At}s}
\label{sec:generic-A-appendix}
In this section, we derive Theorem~\ref{thm:main-theorem-hd-intro-smoothed} as a corollary of our proof of Theorem~\ref{thm:main-theorem-hd-intro} combined with some elementary algebraic considerations.

We will rely on the following lemma that shows that whenever a matrix with low-degree polynomial entries in some variable $A$ has full column rank for some real assignment to variables $A$, it must in fact have non-trivially lower bounded singular value for any $A'$ after a small random perturbation in each entry. Specifically, we prove:

\begin{lemma} \label{lem:sing-value-generic}
Let $\mathcal{G}$ be a product distribution on $N$ dimensional vectors such that the marginal of any coordinate of $G$ is distributed so that no single point has probability $\geq 2^{-N^{O(1)}}$ (for e.g., uniform distribution on $N^{O(1)}$ bit rational numbers in any constant length interval suffices). Let $M(A)$ be a $R \times S$ matrix such that each entry of $M(A)$ is a degree-$d$ polynomial in the $N$-dimensional vector $A$ with each entry upper bounded by $2^{N^{O(1)}}$. Suppose there is a point $A' \in \bbR^N$ such that $M(A')$ has full rank $R$. Then, for any vector $B\in \R^N$ with rational entries of bit complexity at most $N^{O(1)}$, 
\[
\Pr_{G\sim \calG}[M(B+G) \text{ has $R $-th singular value } \geq 2^{-(SN)^{O(1)}}] \geq 1- 2^{-N^{O(1)}} \mper
 \]

\end{lemma}

Our proof relies on the following variant of the classical Schwartz-Zippel lemma and a simple observation about eigenvalues of matrices with polynomial bit complexity entries. 

\begin{fact}[Corollary of Generalized DeMillo–Lipton–Zippel Lemma, Theorem 4.6 in~\cite{MR3788163}] \label{fact:Schwartz-Zippel}
Let $p(x_1, x_2, \ldots, x_n)$ be an $n$-variate degree-$d$ polynomial over any field $\bbF$. Suppose $p$ is not identically equal to $0$. Let $S_1, S_2, \ldots, S_n$ be finite subsets of $\bbF$ of size $s \geq dn^2$. Then, if $x_i \sim S_i$ is chosen uniformly at random and independently for every $i$, then, 
\[
\Pr[ p(x) = 0]\leq dn/s \mper
\]
\end{fact}

\begin{lemma}[Gapped Eigenvalues from Polynomial Bit Complexity] \label{lem:gapped-eigenvalue}
Let $A$ be a $n \times r$ matrix of $N$-bit rational entries. Suppose $A$ has rank $r$. Then, the $r$-th smallest singular value of $A$ is at least $2^{-O(Nn^3)}$.  
\end{lemma}
\begin{proof}
Let $B = A^{\top}A$ and let $B'$ be the matrix of integers obtained by clearing the denominators of the rational numbers appearing in the entries of $B$. The bit complexity by $B'$ is then larger than that of $B$ by at most an additive $Nn^{2}$ and is thus at most $4Nn^{2}$. Further, by the Gershgorin circle theorem, the largest eigenvalue of $B'$ is at most $n2^{4Nn^2}$. 

Since $B'$ is a symmetric matrix with integer entries and has full rank, the determinant of $B'$, $\det(B')$ is a non-zero integer and thus at least $1$ in magnitude. Since $\det(B')$ is the product of all $r$ eigenvalues of $B'$ each of which is at most $n 2^{4 Nn^2}$, the smallest eigenvalue must be at least $n^{-n} 2^{-4Nn^{3}}\leq 2^{-5Nn^{3}}$ and large enough $n$. This completes the proof. 
\end{proof}

\begin{proof}[Proof of Lemma \ref{lem:sing-value-generic}]
For any fixed $B\in\R^N$, consider the determinant $\det(Q)$ of the $R \times R$ matrix $Q= M(A+B)M(A+B)^{\top}$. This is a polynomial of degree $2Rd$ in $A$. For $A^* = A'-B$, from the hypothesis, $M(A^*+B)M(A^*+B)^{\top}$ has full rank $R$. Thus, $\det(Q)$ is not identically equal to $0$ as a polynomial of $A$.

Let $G\in\R^N$ be sampled from $\calG$.
For each entry $i$ of $G$, let $S_i$ be the support of the distribution that $G_i$ is drawn from. Then, we know that this support is of size at least $2^{N^{O(1)}}$. Thus, by the generalized De-Millo-Lipton-Zippel lemma (Fact~\ref{fact:Schwartz-Zippel}), the probability that $M(B+G)M(B+G)^{\top}$ is singular is at most $2^{-N^{O(1)}}$. 

Further, the entries of $M(B+G)M(B+G)^{\top}$ have bit complexity at most $N^{O(1)}$. Thus, by Lemma~\ref{lem:gapped-eigenvalue}, whenever the matrix $M(B+G)M(B+G)^{\top}$ is non-singular, it's smallest eigenvalue is lower bounded by $2^{-(SN)^{O(1)}}$.
\end{proof}

We can now finish the proof of Theorem~\ref{thm:main-theorem-hd-intro-smoothed}. 

\begin{proof}[Proof of Theorem~\ref{thm:main-theorem-hd-intro-smoothed}]
We first note that the randomness of $A_t$ in the analysis of our algorithm (Algorithm~\ref{algo:overall-algo-hd}) is only used to infer inverse polynomial singular value lower bounds on the matrices arising in various steps of our algorithm. For sum of $d=3D$th power of degree-$K$ polynomial $A_t$, every matrix that arises in our analysis has entries that are at most degree-$O(KD)$ polynomials in $A_t$s and have dimension $n^{O(KD)} \times n^{O(KD)}$. Thus, each such matrix is a degree $O(KD)$ function of $A_t$s that is full-rank for random $A_t$. Thus, applying Lemma~\ref{lem:sing-value-generic} yields that the resulting matrix in fact has condition number upper bounded by at most $2^{n^{O(KD)}}$ with probability $1-2^{-n^{O(KD)}}$ over the choice of $n^{O(KD)}$-bit smoothing of $A_t$s. In particular $n^{O(KD)}$ bit truncation of all real numbers occurring in our analysis is enough. As a result, our analysis of Algorithm~\ref{algo:overall-algo-hd} succeeds with a worse estimation error of $2^{n^{O(KD)}} \Norm{E}_F^{1/d}+ \poly(n)\tau^{1/d}$ (where $\tau$ is the numerical accuracy parameter). In particular, when $E = 0$, which is the assumption in \pref{thm:main-theorem-hd-intro-smoothed}, we obtain polynomial time (in input size and $\log 1/\tau$) algorithm for computing $\tau$-approximate estimates of $A_t$s up to permutations (and signs for even $D$).
\end{proof}

\end{document}